\PassOptionsToPackage{dvipsnames}{xcolor}
\documentclass[journal]{vgtc}                
\ifpdf
  \pdfoutput=1\relax                   
  \usepackage{graphicx}                
  \DeclareGraphicsExtensions{.png,.pdf,.jpg,.jpeg} 
\else
  \usepackage{graphicx}                
  \DeclareGraphicsExtensions{.eps}     
\fi%

\graphicspath{{figures/}{pictures/}{images/}{./}} 

\usepackage{microtype}                 
\PassOptionsToPackage{warn}{textcomp}  
\usepackage{textcomp}                  
\usepackage{mathptmx}                  
\usepackage{times}                     
\usepackage{cite}                      
\usepackage{tabu}                      
\usepackage{booktabs}                  

\usepackage{hyperref}
\usepackage{amsfonts}
\usepackage{subfig}
\usepackage[normalem]{ulem}
\usepackage{enumitem}
\usepackage{setspace}
\usepackage{graphicx}
\usepackage{amsthm}
\usepackage{multicol}
\usepackage{multirow}
\usepackage[table]{xcolor}
\usepackage{stfloats}

\newtheorem{theorem}{Theorem}

\newcommand {\mm}[1] {\ifmmode{#1}\else{\mbox{\(#1\)}}\fi}

\newcommand{\Rspace}        {{\mathbb R}}
\newcommand{\Xspace}        {{\mathbb X}}

\newcommand{\Hgroup}        {{\mathsf H}}

\newcommand{\edgeComplex}{{\mathsf E \mathsf d \mathsf g\mathsf e}}

\newcommand{\para}[1]{\textbf{#1.}}

\makeatletter
\let\orgdescriptionlabel\descriptionlabel
\renewcommand*{\descriptionlabel}[1]{%
  \let\orglabel\label
  \let\label\@gobble
  \phantomsection
  \edef\@currentlabel{#1\unskip}%
  \let\label\orglabel
  \orgdescriptionlabel{#1}%
}
\makeatother

\usepackage{tikz}
\usepackage{onimage}
\usepackage{transparent}
\usepackage{tcolorbox}

\tikzset{
    image label/.style={
        every node/.style={
            fill=white,
            opacity=0.7,
            text=black,
            font=\tiny,
            anchor=south east,
            xshift=-0.0cm,
            yshift=-0.15cm,
            at={(1,0)}
        }
    }
}

\newcommand{\demoURL}{\textit{{\footnotesize\url{https://usfdatavisualization.github.io/UntangleFDL/}}}}

\newcommand{\sourceURL}{\textit{{\footnotesize\url{https://github.com/USFDataVisualization/UntangleFDL/}}}}

\definecolor{ourMethodColor}{HTML}{F2F2F2}
\definecolor{hlineColor}{HTML}{808080}
\definecolor{LCMCColorLight}{HTML}{F8F9FD}
\definecolor{LCMCColorDark}{HTML}{EEF2FA}

\newcolumntype{L}[1]{>{\raggedright\let\newline\\\arraybackslash\hspace{0pt}}m{#1}}
\newcolumntype{C}[1]{>{\centering\let\newline\\\arraybackslash\hspace{0pt}}m{#1}}
\newcolumntype{R}[1]{>{\raggedleft\let\newline\\\arraybackslash\hspace{0pt}}m{#1}}

\onlineid{1303}

\vgtccategory{Research}
\vgtcpapertype{algorithm/technique}

\title{Untangling Force-Directed Layouts Using Persistent Homology}

\author{Bhavana Doppalapudi, Bei Wang, and Paul Rosen}
\authorfooter{
\item
 B. Doppalapudi was with the University of South Florida, Tampa, FL, 33610. E-mail: bdoppalapudi@usf.edu
\item
 B. Wang was with the University of Utah, Salt Lake City, UT, 84112.
E-mail: beiwang@sci.utah.edu
\item
P. Rosen was with the University of Utah, Salt Lake City, UT, 84112.
E-mail: prosen@sci.utah.edu
}

\shortauthortitle{Doppalapudi \MakeLowercase{\textit{et al.}}: Untangling Force-Directed Layouts}

\abstract{Force-directed layouts belong to a popular class of methods used to position nodes in a node-link diagram. However, they typically lack direct consideration of global structures, which can result in visual clutter and the overlap of unrelated structures. In this paper, we use the principles of persistent homology to untangle force-directed layouts thus mitigating these issues. First, we devise a new method to use 0-dimensional persistent homology to efficiently generate an initial graph layout. The approach results in faster convergence and better quality graph layouts. Second, we provide a new definition and an efficient algorithm for 1-dimensional persistent homology features (i.e., tunnels/cycles) on graphs. We provide users the ability to interact with the 1-dimensional features by highlighting them and adding cycle-emphasizing forces to the layout. Finally, we evaluate our approach with 32 synthetic and real-world graphs by computing various metrics, e.g., co-ranking, edge crossing, etc., to demonstrate the efficacy of our proposed method.%
}


\teaser{
  \centering
  
  {\begin{minipage}[t]{0.04\linewidth}
  \hspace{3pt}
  \rotatebox{90}{\footnotesize \hspace{7pt} Random Layout}
  \hspace{2pt}
  
  \vspace{10pt}
  \rotatebox{90}{\footnotesize \hspace{11pt} Radial Layout}
  \rotatebox{90}{\footnotesize \hspace{10pt} (our approach)}
  
  \vspace{10pt}
  \rotatebox{90}{\footnotesize \hspace{9pt} Layered Layout}
  \rotatebox{90}{\footnotesize \hspace{10pt} (our approach)}
  \end{minipage}}
  \hspace{-6pt}
  \fcolorbox{red}{white}{%
  \begin{minipage}[t]{0.157\linewidth}
      \begin{tikzonimage}[trim=15pt 0 15pt 0, clip, width=\linewidth]{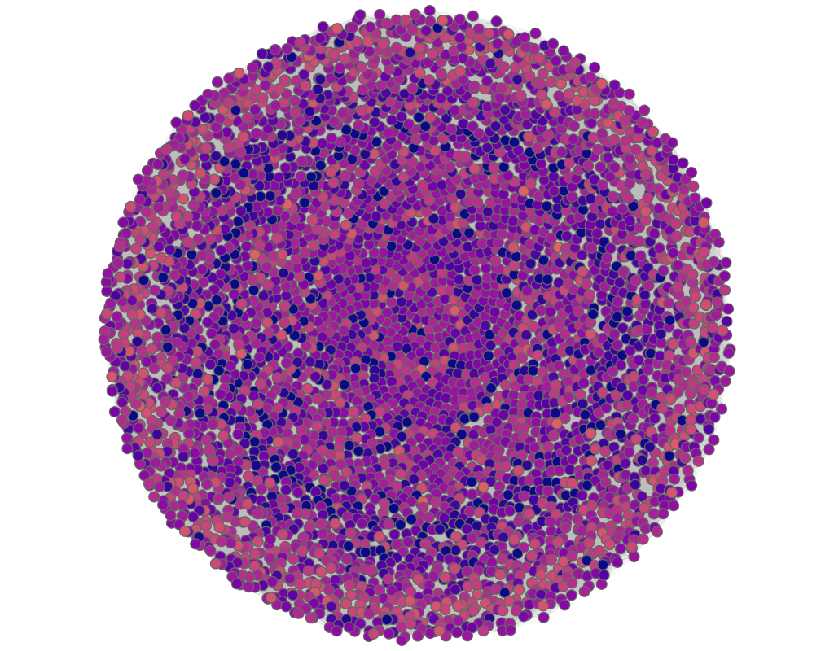}[image label]
        \node{$Q_{LCMC}$: 0.006};
    \end{tikzonimage}
    
    \vspace{1pt}
    \begin{tikzonimage}[trim=15pt 0 15pt 0, clip, width=\linewidth]{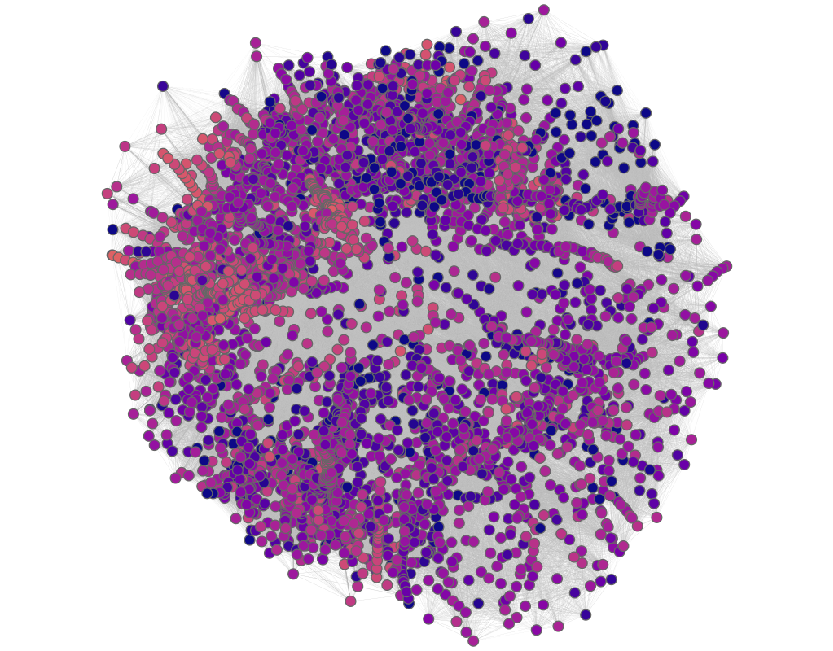}[image label]
        \node{$Q_{LCMC}$: 0.108};
    \end{tikzonimage}
    \begin{tikzonimage}[trim=15pt 0 15pt 0, clip, width=\linewidth]{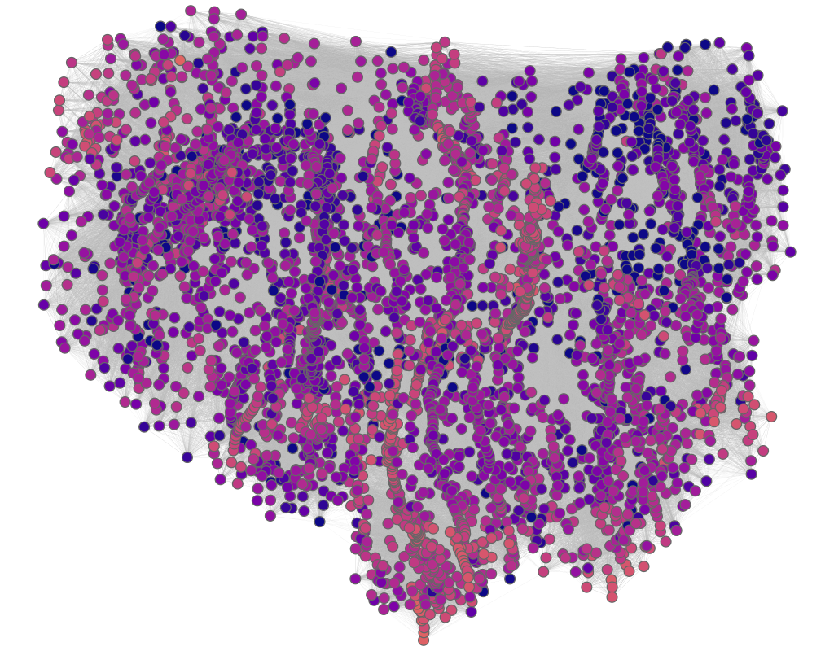}[image label]
        \node{$Q_{LCMC}$: 0.152};
    \end{tikzonimage}
  \end{minipage}}
  \hspace{-3pt}
  \begin{minipage}[t]{0.157\linewidth}
    \hspace{-5pt}
    \begin{tikzonimage}[trim=15pt 0 15pt 0, clip, width=\linewidth]{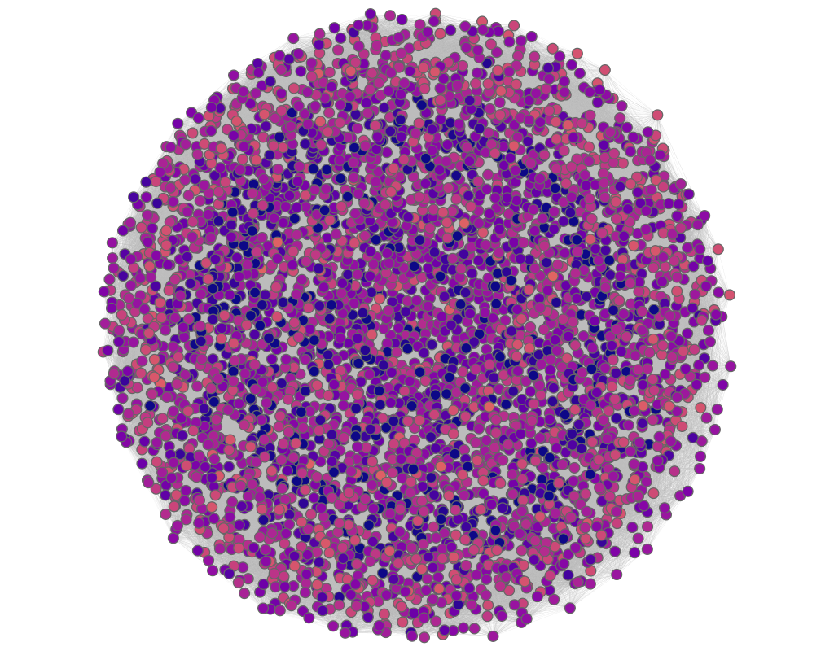}[image label]
        \node{$Q_{LCMC}$: 0.005};
    \end{tikzonimage}
    \fcolorbox{blue}{white}{ 
    \hspace{-4pt}
    \begin{minipage}[t]{\linewidth}
    \begin{tikzonimage}[trim=15pt 0 15pt 0, clip, width=\linewidth]{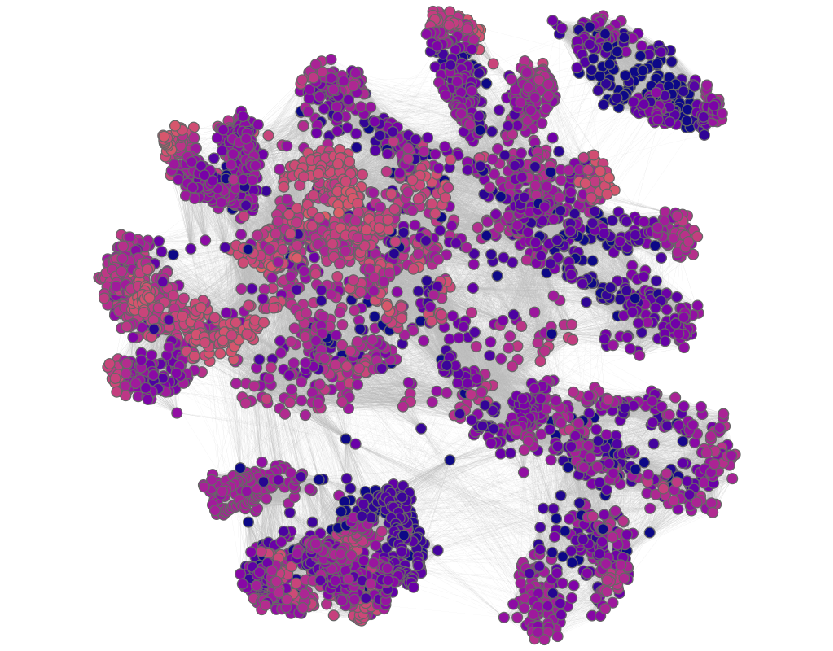}[image label]
        \node{$Q_{LCMC}$: 0.173};
    \end{tikzonimage}
    \begin{tikzonimage}[trim=15pt 0 15pt 0, clip, width=\linewidth]{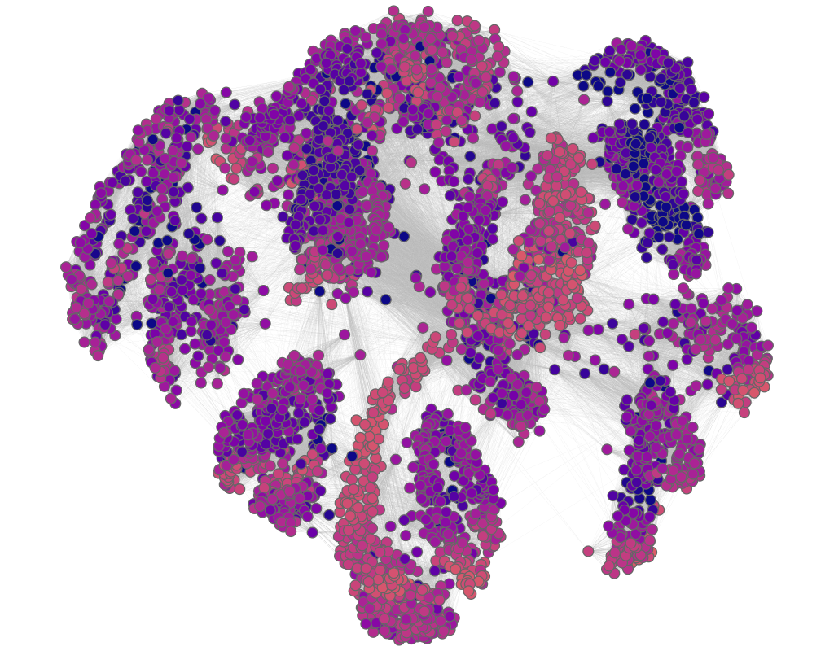}[image label]
        \node{$Q_{LCMC}$: 0.172};
    \end{tikzonimage}
      \end{minipage}}
  \end{minipage}
  \hspace{-5pt}
  \begin{minipage}[t]{0.157\linewidth}
  \begin{tikzonimage}[trim=15pt 0 15pt 0, clip, width=\linewidth]{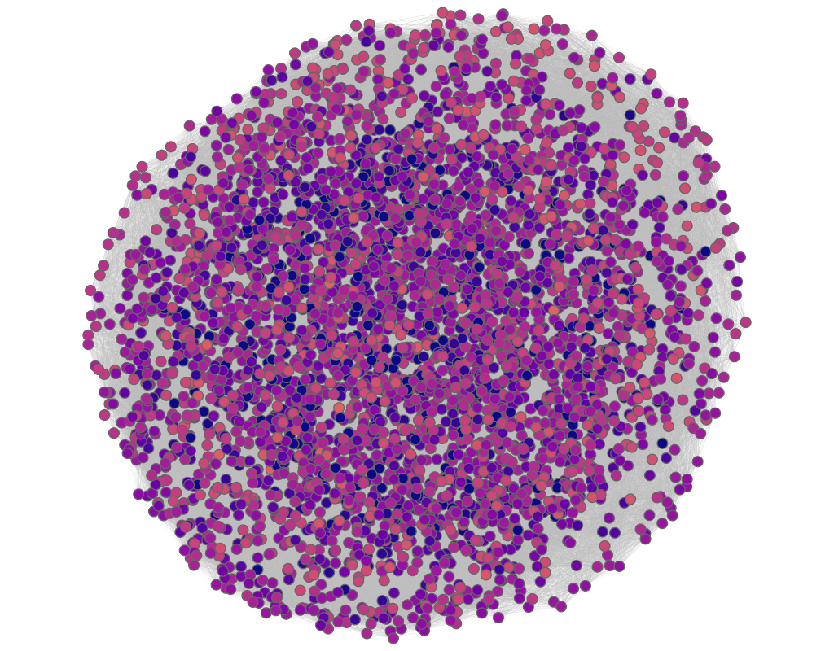}[image label]
        \node{$Q_{LCMC}$: 0.006};
    \end{tikzonimage}
    \begin{tikzonimage}[trim=15pt 0 15pt 0, clip, width=\linewidth]{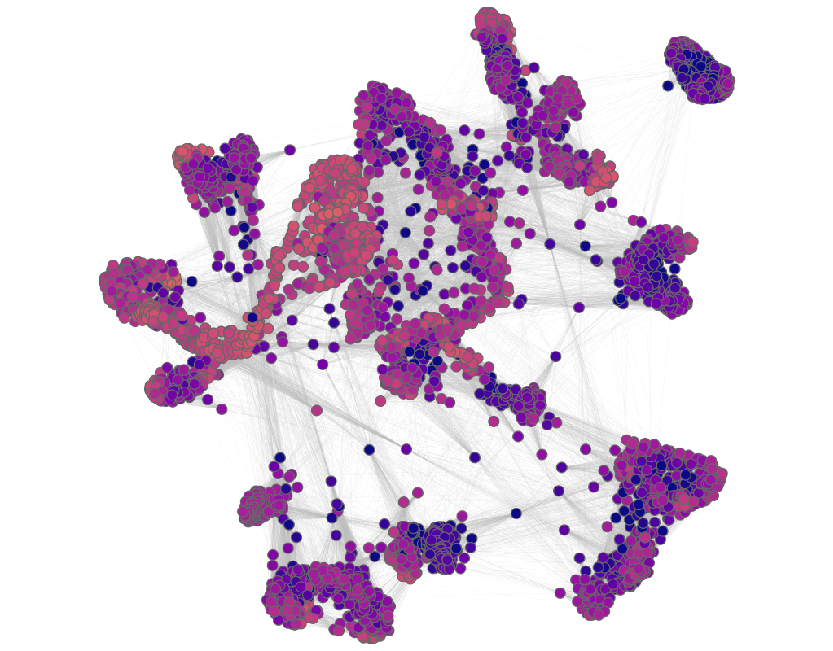}[image label]
        \node{$Q_{LCMC}$: 0.228};
    \end{tikzonimage}
    \begin{tikzonimage}[trim=15pt 0 15pt 0, clip, width=\linewidth]{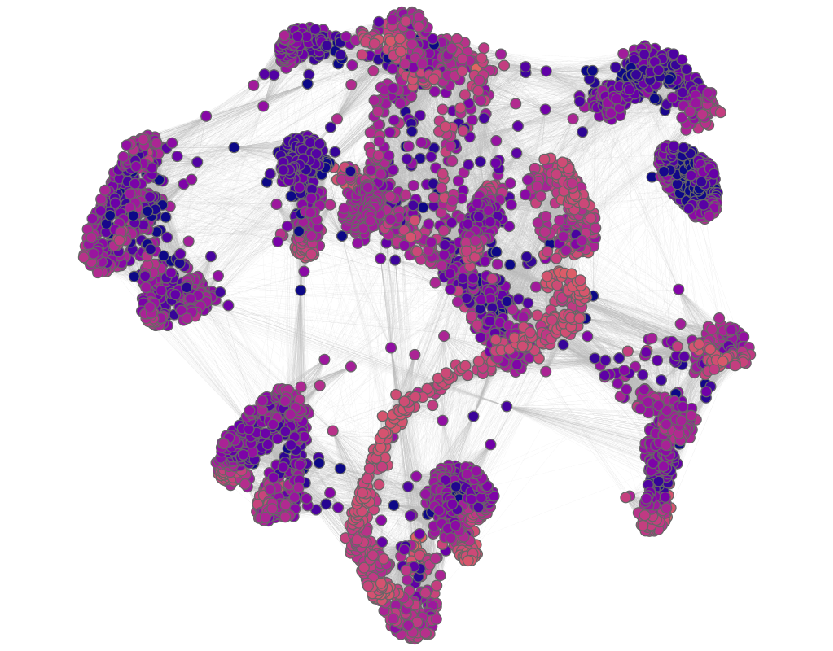}[image label]
        \node{$Q_{LCMC}$: 0.233};
    \end{tikzonimage}
  \end{minipage}
  \hspace{-4pt}
  \begin{minipage}[t]{0.157\linewidth}
  \fcolorbox{blue}{white}{%
  \begin{minipage}[t]{\linewidth}
    \begin{tikzonimage}[trim=15pt 0 15pt 0, clip, width=\linewidth]{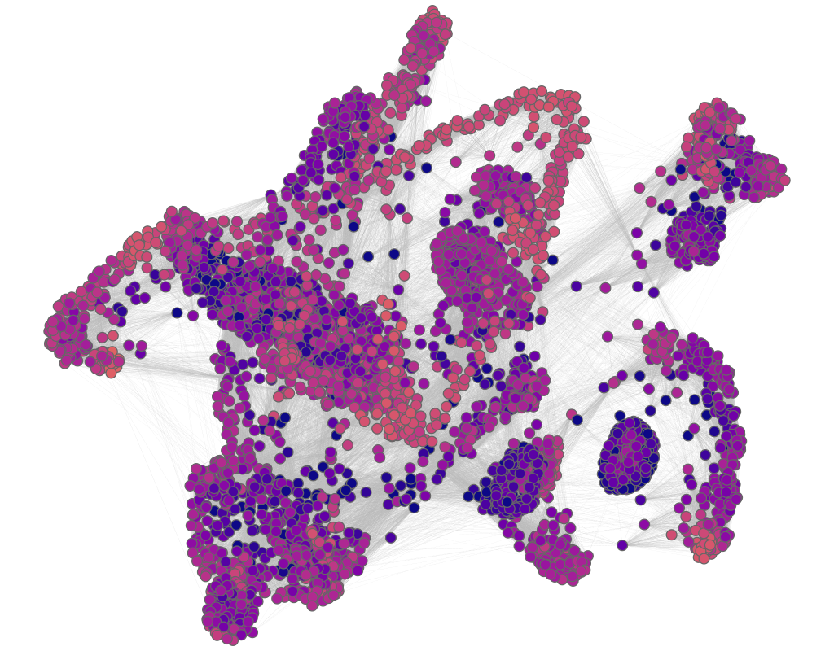}[image label]
        \node{$Q_{LCMC}$: 0.139};
    \end{tikzonimage}
    \end{minipage}}
    \begin{tikzonimage}[trim=15pt 0 15pt 0, clip, width=\linewidth]{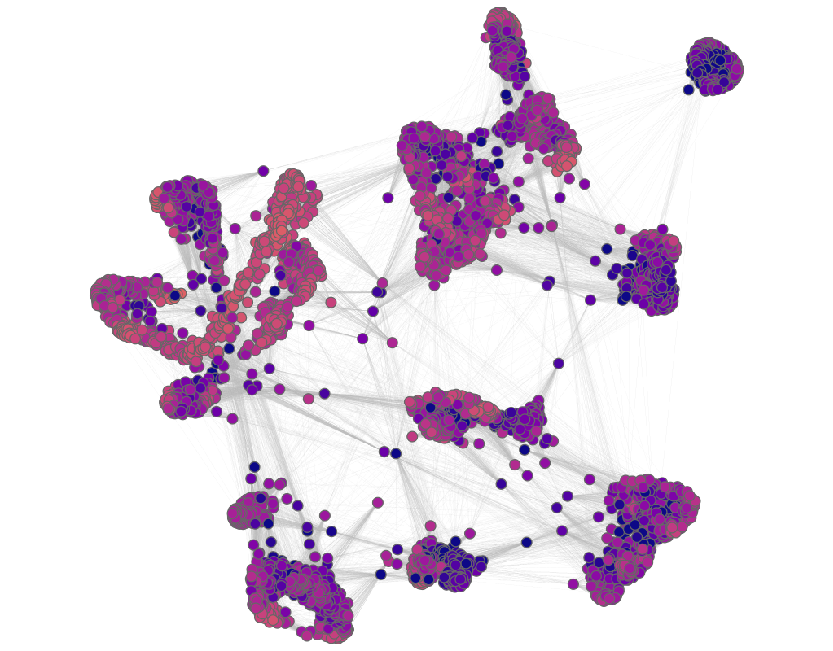}[image label]
        \node{$Q_{LCMC}$: 0.274};
    \end{tikzonimage}
    \begin{tikzonimage}[trim=15pt 0 15pt 0, clip, width=\linewidth]{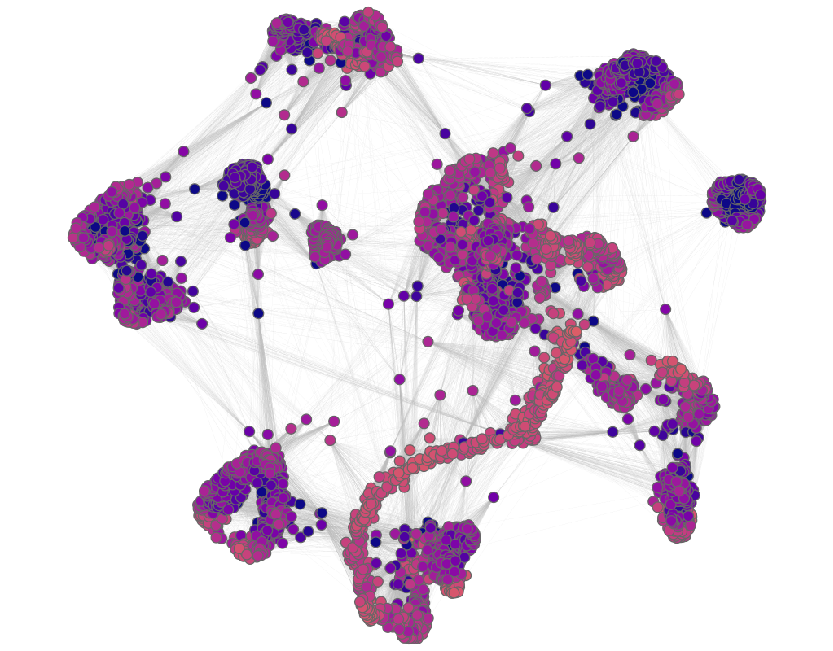}[image label]
        \node{$Q_{LCMC}$: 0.277};
    \end{tikzonimage}
  \end{minipage}
  \hspace{-1pt}
  \fcolorbox{OliveGreen}{white}{
  \begin{minipage}[t]{0.157\linewidth}
    \begin{tikzonimage}[trim=15pt 0 15pt 0, clip, width=\linewidth]{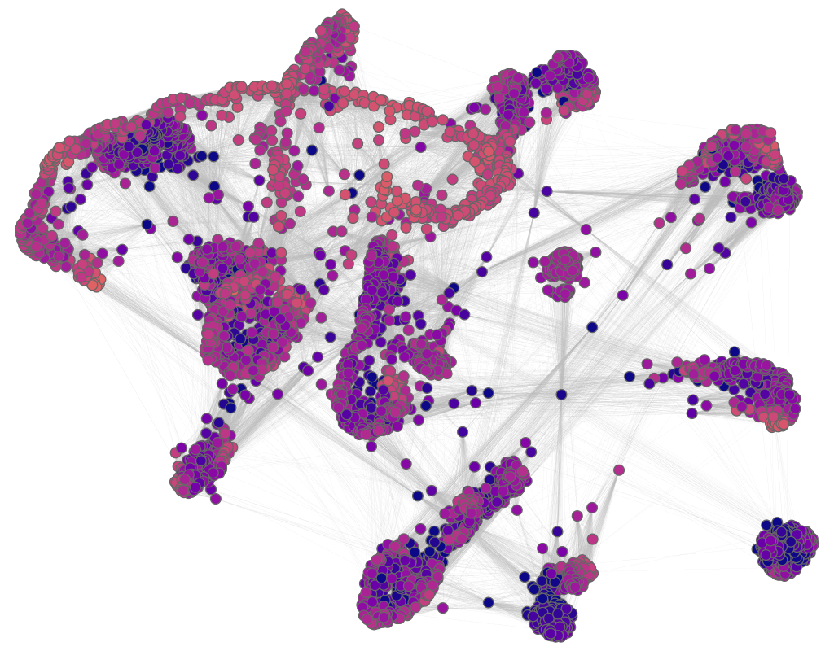}[image label]
        \node{$Q_{LCMC}$: 0.277};
    \end{tikzonimage}  
    
    \vspace{1pt}
    \begin{tikzonimage}[trim=15pt 0 15pt 0, clip, width=\linewidth]{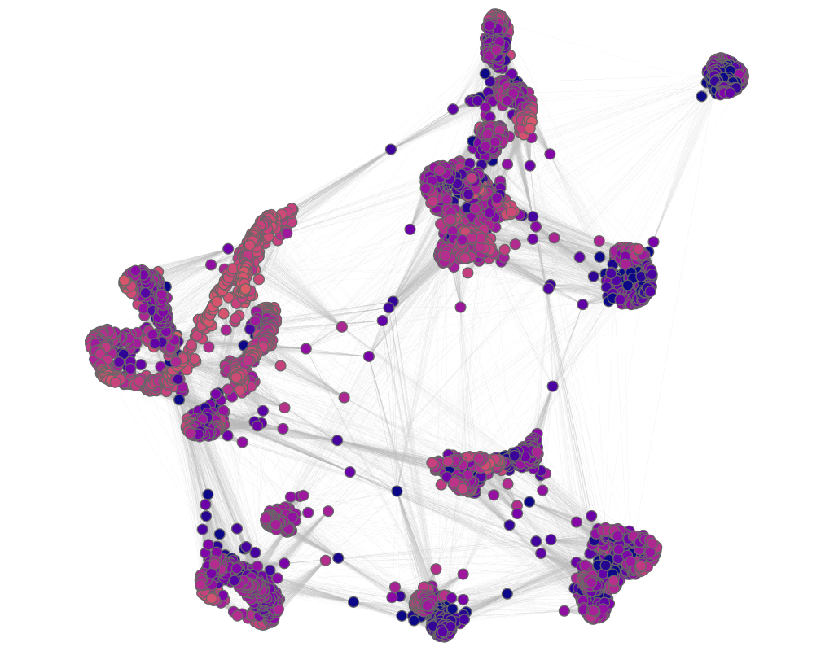}[image label]
        \node{$Q_{LCMC}$: 0.303};
    \end{tikzonimage}
    \begin{tikzonimage}[trim=15pt 0 15pt 0, clip, width=\linewidth]{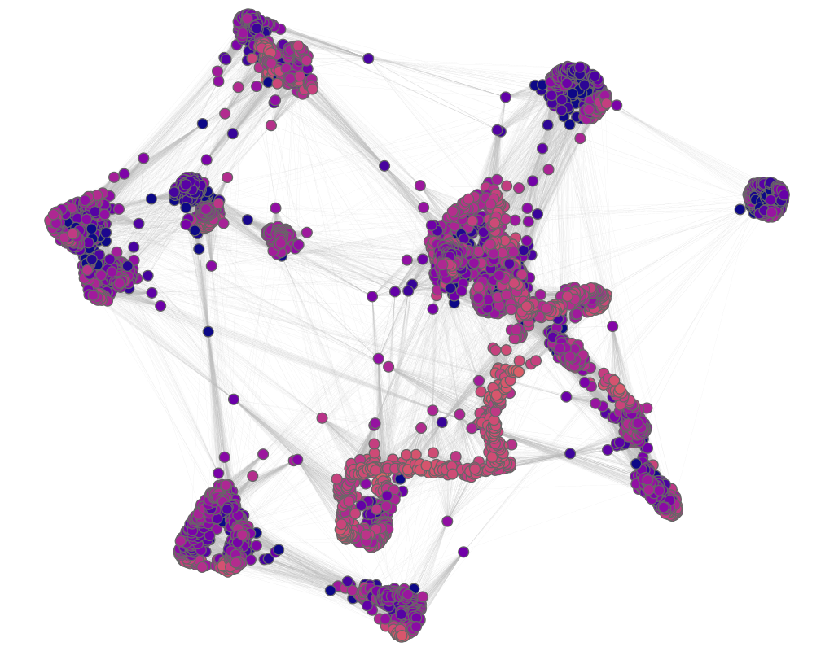}[image label]
        \node{$Q_{LCMC}$: 0.300};
    \end{tikzonimage}      
  \end{minipage} }
  \hfill
  \rule[-125pt]{.5pt}{180pt}
  \begin{minipage}[t]{0.12\linewidth}
    \centering
    \tiny
    \vspace{-50pt}
    \textit{neato}
    
    \begin{tikzonimage}[trim=15pt 0 15pt 0, clip, width=\linewidth]{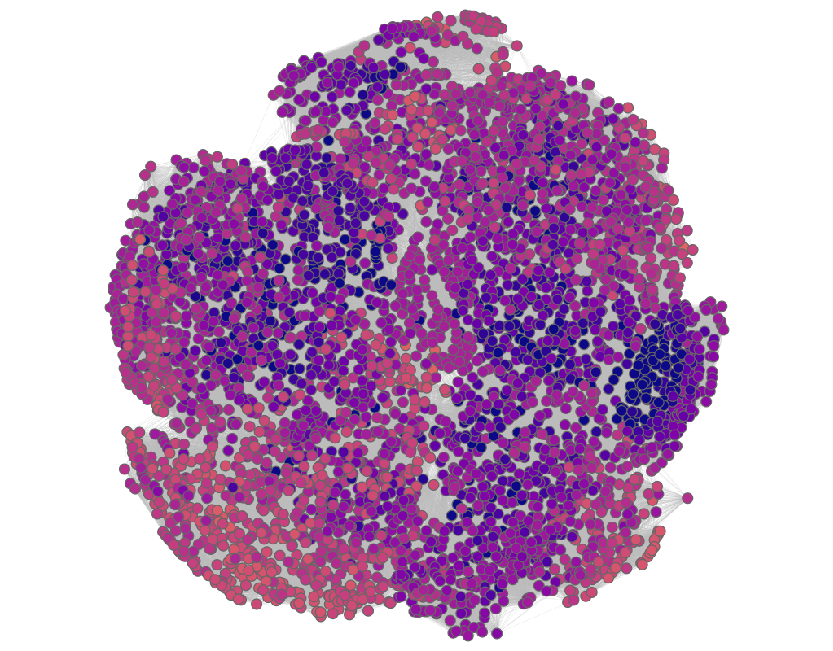}[image label]
        \node{$Q_{LCMC}$: 0.188};
    \end{tikzonimage}      
    
    \textit{fdp}

    \begin{tikzonimage}[trim=15pt 0 15pt 0, clip, width=\linewidth]{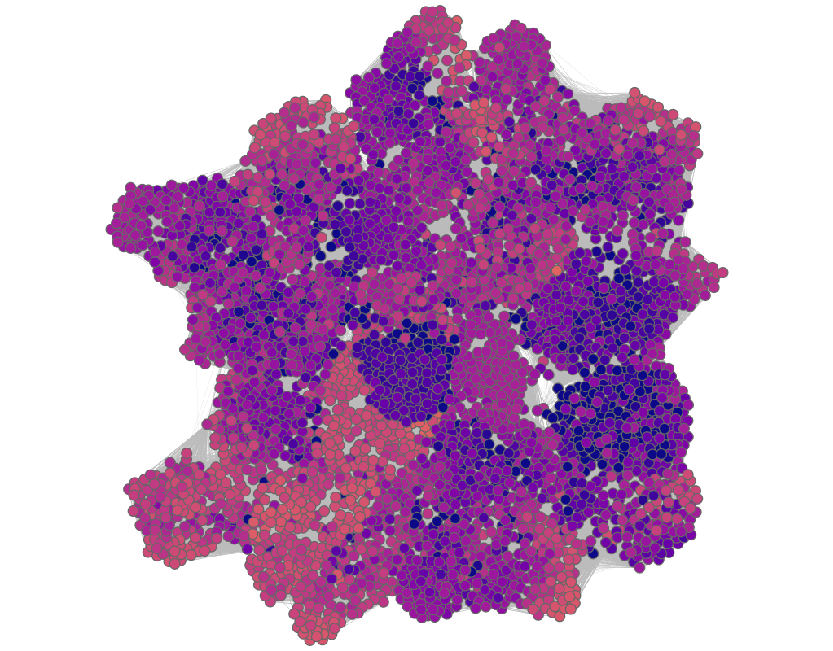}[image label]
        \node{$Q_{LCMC}$: 0.182};
    \end{tikzonimage}      
    
    \textit{sfdp}

    \begin{tikzonimage}[trim=15pt 0 15pt 0, clip, width=\linewidth]{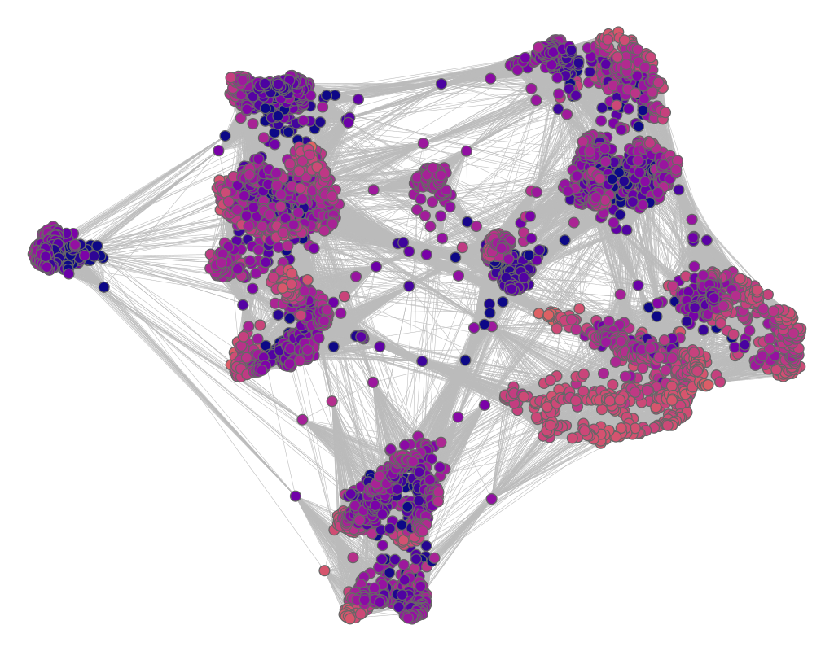}[image label]
        \node{$Q_{LCMC}$: 0.304};
    \end{tikzonimage}      
    
  \end{minipage}

  \vspace{-10pt}
    \hspace{0.04\linewidth}
  \parbox{0.157\linewidth}{\centering \footnotesize Initial~Layout (0~Iterations) }
  \parbox{0.157\linewidth}{\centering \footnotesize 5 Iterations}
  \hspace{1pt}
  \parbox{0.157\linewidth}{\centering \footnotesize 10 Iterations}
  \parbox{0.157\linewidth}{\centering \footnotesize 25 Iterations}
  \hspace{0.5pt}
  \parbox{0.157\linewidth}{\centering \footnotesize Final~Layout (300~Iterations)}
  \hfill
  \parbox{0.12\linewidth}{\centering \footnotesize Static Layout Methods}  

\caption{Our approach in using persistent homology to untangle force-directed layouts has two main functionalities. 
First, as demonstrated via the \textsc{hic\_1k\_net} dataset, we use persistent homology to formulate an initial graph layout (\textcolor{red}{column 1}) that improves both the convergence rate of the graph layout and the final layout quality.  
Second, our approach comes with interactive capabilities, as illustrated in \autoref{fig:teaser2}. 
We use the local continuity meta criterion ($Q_{LCMC}$) for layout evaluation (larger is better). 
In terms of convergence, our approach (\textcolor{blue}{column 2, rows 2 and 3}) shows the formation of major graph structures as early as $5$ iterations, whereas  the standard approach (\textcolor{blue}{column 4, row 1}) takes $25+$ iterations. 
In terms of final layout quality, the $Q_{LCMC}$ scores for the final layout (\textcolor{OliveGreen}{column 5}) show that our approach significantly exceeds the standard one. Three static layout methods, \emph{neato}, \emph{fdp}, and \emph{sfdp}, are also included for comparison (column 6), with only \emph{sfdp} (column 6, row 3) showing a comparable $Q_{LCMC}$ score.}
  \label{fig:teaser}
 }

\vgtcinsertpkg

\begin{document}

\firstsection{Introduction}

\maketitle

Force-directed layouts remain one of the most popular methods for drawing graphs. Their popularity stems from several desirable qualities: generally, they are simple to implement, they are fast for small graphs, they produce aesthetically pleasing layouts, and their iterative algorithms make progressive visualization and interaction natural. 
Nevertheless, force-directed layouts also suffer from numerous limitations, including poor initialization and over-constraint, leading to poor local minima and limited robustness to noise. 

This paper addresses two of these limitations by utilizing and highlighting important topological features of the graph.
First, force-directed layouts are strongly influenced by the initial layout of graph nodes, which is often generated randomly. 
After the initialization, successive application of the forces among nodes causes the layout to settle in a locally minimal energy state, which \textit{hopefully} shows the graph's topological structure. 
Unfortunately, the random initial layout does not consider the global topology, which potentially slows convergence and can lead to unrelated structures overlapping in the visualization. 
Second, since force-directed layouts are over-constrained systems, even without the overlap of unrelated structures, certain topological features have their shape distorted, tangled, or hidden by noise (i.e., low weight edges) making it difficult to visually identify topological features.

\begin{figure}[!t]
    \centering
    
    \begin{minipage}[b]{0.015\linewidth}
    \end{minipage}
    
    \subfloat[\textsc{bn-mouse-visual-cortex-2}]{\includegraphics[height=90pt]{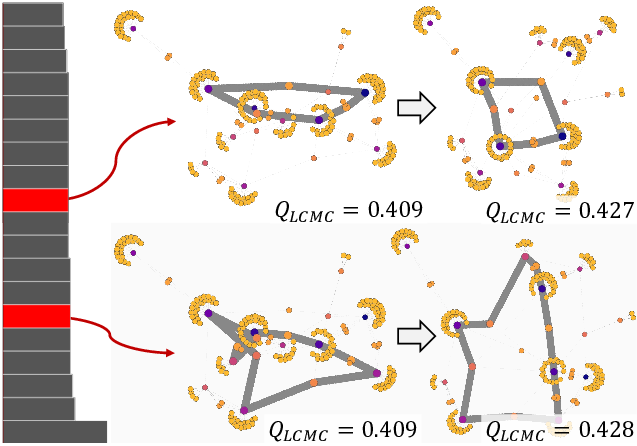}}
    \hfill
    \subfloat[\textsc{bcsstk}]{\includegraphics[height=90pt]{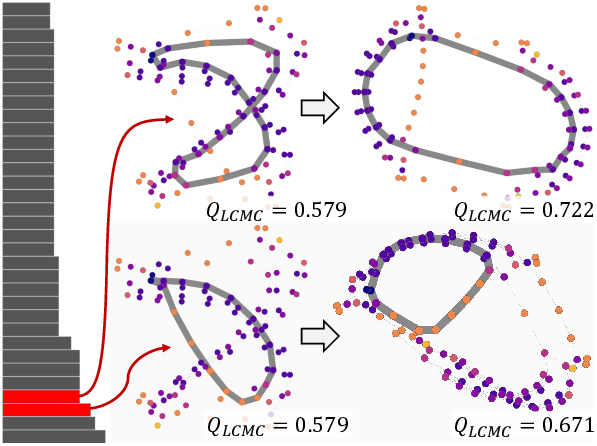}}

    \caption{Our second untangling functionality is to \textit{interactively} untangle persistent homology cycle features. Each cycle is represented by a bar in the barcode (left), which is used to highlight cycles by mouse over (middle) and apply an elliptical force by click (right). In this example, the elliptical forces are used to regularize deformed cycles (a, top and bottom), untangle twisted cycles (b, top), and disentangle a cycle covered by unrelated structures (b, bottom). The $Q_{LCMC}$ scores show that the elliptical forces also improve the overall graph layout.}
    \label{fig:teaser2}
\end{figure}

We address these limitations by using techniques from persistent homology to provide better initialization and support interactive exploration of the graph topology. 
We are not the first to consider the use of persistent homology in force-directed layouts. In their work, Suh et al.~\cite{suh2019persistent} used 0-dimensional persistent homology features (i.e., components) to enable new forms of interaction with a force-directed layout. Importantly, their work did not address the problem of initial graph layouts. Furthermore, their work did not consider 1-dimension persistent homology features, i.e., tunnels/cycles, which are an important topological features frequently present in graphs. Our work advances theirs by: (1)~utilizing persistent homology as a framework for efficiently generating good quality initial graph layouts (see \autoref{fig:teaser}); (2)~extending the algorithm for extracting topological features from graphs to efficiently extract 1-dimensional features; and (3)~providing new interactions with the 1-dimensional persistent homology features of the graph (see \autoref{fig:teaser2}).

A natural question at this point would be, why persistent homology in this context? 
First, it has a strong mathematical foundation, making the technique technically sound and  robust to noise~\cite{zomorodian2005computing}. Second, persistent homology computation, being slow for general point cloud data~\cite{otter2017roadmap}, can be made very efficient within the context we present. Finally, instead of considering only the node-link topology, persistent homology allows for a natural way of evaluating the interaction of the node-link topology with a function (i.e., a weight) defined per link. These advantages make it a valuable companion to existing graph analysis and drawing tools.

\section{Prior Work}

Graph visualization as a research area has received significant attention. 
Tamassia~\cite{tamassia2013handbook} provided a broad overview of the state of the art.
Graphs are often represented as node-link diagrams, adjacency matrices, or hybrid representations. 
Ghoniem et al.~\cite{ghoniem2004comparison}  showed that matrix-based representations are generally better than node-link diagrams for node counts greater than 20, but node-link diagrams outperform for specific tasks (e.g., pathfinding) and are aesthetically pleasing. 
Archambault et al.~\cite{archambault2010readability} assessed how graph representations affected readability and showed that only clustering could be efficiently performed on larger graphs.
A separate study by Saket et al.~\cite{saket2014node} 
concluded that node-link diagrams are superior for topology and network-related analysis compared to adjacency matrices.

\subsection{Node-Link Diagrams}
Node-link diagrams represent the data with nodes for entities and links for their pairwise relations. 
There are multiple general-purpose layout methods to position nodes with frequently used ones including force-directed, constraint-based, layered, algebraic, and multiscale layouts.

\para{Force-Directed Layouts} Force-directed (or force-based) layouts consider graphs as mechanical systems and apply forces to the nodes. In general, repulsive forces, similar to those of electrically charged particles, exist between all the vertices, and attractive forces, similar to spring-like forces, exist between connected vertices or between neighbors. 
The Eades model~\cite{eades1984heuristic} was the first to apply spring forces on the initial layout to achieve a minimal energy position. Later, Fruchterman and Reingold~\cite{fruchterman1991graph} modified the Eades model to achieve a system that distributes the vertices evenly, and has uniform edge lengths and symmetry. 
Kamada and Kawai~\cite{kamada1989algorithm} developed another variant on Eades' work. Instead of just applying attractive forces between neighboring vertices, they applied the concept of ideal distance, which is proportional to the length of the shortest path. 
Although computational costs are high for this method, speed-ups have been achieved using  heuristics~\cite{frick1994fast} and the  GPU~\cite{godiyal2008rapid}. 
Meidiana et al.~\cite{meidiana2020sublinear} presented a sublinear time model that pairs a radial tree drawing of a breadth first spanning tree with random sampling of repulsive forces. While their approach has technical similarities with our initial layout approach, there are important implementation differences and theoretical guarantees that come from our use of 0-dimensional persistent homology.

\para{Constraint-Based Layouts} Constraint-based layouts are a more sophisticated version of force-directed layouts. 
Dwyer et al.~\cite{dwyer2008exploration} archived a high-quality, topology-preserving visualization  by implementing a constraint-based layout for a detailed view and a force-directed layout for the overview. 
Archambault et al.~\cite{archambault2007topolayout} proposed the \emph{TopoLayout} algorithm that dynamically adapts the graph layout method based upon the topology detected within subgraphs.

\para{Layered Layouts} Layered layouts, in general, are used for directed graph layouts. Sugiyama et al.~\cite{sugiyama1981methods} used a 4-phase approach to layout graphs: (1)~removing cycles, (2)~assigning nodes to layers, (3)~reducing edge crossings, and (4)~assigning coordinates  to nodes. 
Bachmaier et al.~\cite{bachmaier2008cyclic} proposed to visualize directed cyclic graphs by skipping the cycle removal step.

\para{Algebraic Methods} Koren et al.~\cite{koren2002ace} developed the \textit{Algebraic Multigrid Method} (ACE) algorithm that minimizes quadratic energy.
Harel and Koren introduced \textit{High-Dimensional Embedding} (HDE)~\cite{harel2002graph} that projects a high-dimensional representation of a graph with PCA.

\para{Multiscale Layouts} Multiscale layouts start with a coarse layout and refine it in phases. 
Hachul and J{\"u}nger~\cite{hachul2004drawing} proposed \textit{Fast Multipole Multilevel Method} (FM$^3$), a force-directed method that incorporates a multiscale approach in a system to calculate repulsive forces in rapidly evolving potential fields. 
By comparing various algorithms, Hachul and J{\"u}nger~\cite{hachul2007large} showed that multiscale methods, including FM$^3$, ACE, and HDE, were significantly faster than regular force-directed layouts. They also found that FM$^3$ produced the best quality graphs in the group.

\para{Node Congestion} Node co-location is a challenging problem, particularly in multiscale layouts~\cite{von2011visual}.
Space-filling curves have been used to avoid node co-location~\cite{muelder2008rapid}. However, the approach works only for datasets with clear clustering, and for dense graphs, the visualizations are not of good quality or aesthetically pleasing.
Gansner and North~\cite{gansner1998improved} proposed to improve the force-directed layout by moving the overlapping nodes within cells of a constructed Voronoi diagram. 
By selecting good starting positions for nodes, Gajer et al.~\cite{gajer2000multi} developed a multiscale approach that improved computation time and better preserved the graph's structure. 
Adai et al.~\cite{adai2004lgl} introduced the \emph{Large Graph Layout} (LGL) algorithm, which uses a minimum spanning tree to guide the force-directed iterative layout to visualize large protein map networks. 
However, they did not consider datasets with different characteristics. 
Dunne and Shneiderman~\cite{dunne2013motif} proposed to use motifs for node and edge bundling. Their technique replaced common graph patterns of nodes and links with intuitive glyphs. They showed that the approach required less screen space and effort, while preserving the underlying relations. However, the glyphs required additional learning from users, and charts with many large glyphs added clutter to the display and increased the possibility of overlap.

\para{Edge Congestion} Node-link diagrams frequently suffer from edge crossings. Carpendale et al.~\cite{carpendale2001examining} proposed displacing edges running through the area of interest. However, certain questions were left unanswered (e.g., the amount of edge displacement to use). 
For graphs without hierarchy, Holten and Van Wijk proposed a self-organizing bundling method, where edges act as flexible springs attracting each other~\cite{holten2009force}. 
\emph{ASK-Graph}~\cite{abello2006ask} addresses the issues for highly dense graphs with node counts approaching 200k.  
Bach et al.~\cite{bach2016towards} proposed to use confluent drawings (CDs) for edge bundling based on network connectivity, which showed some promising results. 
Nevertheless, CDs worked only for sparse graphs where node counts were less than 20 and the edge density was less than 50. 
Such an approach also showed low participant confidence, indicating that CDs require significant learning and may be misleading. 
Zinsmaier et al.~\cite{6327254} proposed a level-of-detail technique that performs density-based nodes aggregation and edge accumulation.

\para{Interaction} Research into interactive visualization has been done to assist with efficient data explorations. Commonly used interaction techniques for graphs include panning and zooming~\cite{viegas2007manyeyes} and fisheye views~\cite{furnas1986generalized, wang2018structure, sarkar1994graphical} that focus on areas of interest.

\subsection{Persistent Homology} 
Persistent homology studies the topological features of data that persist across multiple scales. 
Weinberger gave a brief explanation in his work titled ``What is ... persistent homology?''~\cite{weinberger2011persistent}, while Harer and Edelsbrunner~\cite{edelsbrunner2008persistent} detailed the concept and history of persistent homology. 
Persistent homology has shown great promise to assist in the analysis of complex graphs due to the types of features it extracts and its ability to differentiate signal from noise~\cite{weinan2012landscape, petri2013networks, petri2013topological, horak2009persistent, donato2013decimation}. 
For example, Rieck et al.~\cite{rieck2017clique} used persistent homology to track the evolution of clique communities across different edge weight thresholds. 
Persistent homology has also led to notable results in the study of brain networks~\cite{cassidy2015brain, lee2012persistent, lee2011computing} and social  networks~\cite{bampasidou2014modeling, hajij2018visual, horak2009persistent}. 
Although persistent homology has been used mainly for analysis tasks within these prior methods, it has not been used to improve the visualization of graphs.

Recently, Suh et al.~\cite{suh2019persistent} used 0-dimensional persistent homology features to create a persistence barcode visualization, which was then used to manipulate a force-directed graph layout. Their framework was limited to extracting 0-dimensional topological features and utilizing those features for interactive manipulation of the graph. 
Our work extends and complements the work of Suh et al.\ by utilizing the 0-dimensional features to preprocess the initial layout of graphs, thus improving the rate of convergence and quality of layouts, and by providing an algorithm to efficiently extract, interactively highlight, and manipulate the graph layout with 1-dimensional topological features of a graph.

\section{Methods}
\label{sec:methods}

\begin{figure}[!b]
    \centering
    
    \vspace{-15pt}
    \begin{minipage}[b]{0.775\linewidth}
        \includegraphics[width=\linewidth]{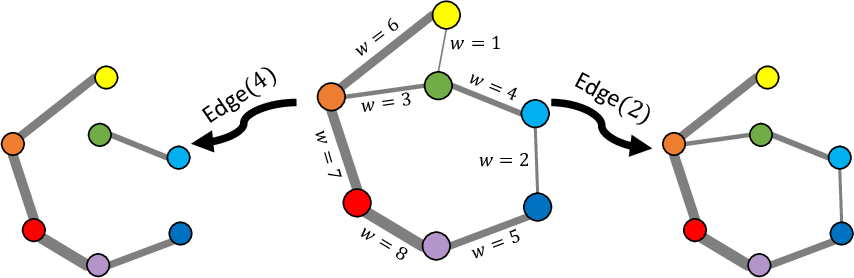}
    \end{minipage}    
    \begin{minipage}[b]{0pt}
        \hspace{-195pt}\subfloat[\label{fig:edge_complex:E4}]{}
        \vspace{45pt}
    \end{minipage}    
    \begin{minipage}[b]{0pt}
        \hspace{-125pt}\subfloat[\label{fig:edge_complex:G}]{}
        \vspace{55pt}
    \end{minipage}    
    \begin{minipage}[b]{0pt}
        \hspace{-15pt}\subfloat[\label{fig:edge_complex:E2}]{}
        \vspace{45pt}
    \end{minipage}    
    \caption{Two $\edgeComplex$ complexes computed on a  weighted graph (b), capturing different homology structures. (a)~$\edgeComplex(4)$ shows two $\Hgroup_0$ components and no $\Hgroup_1$ cycle. (c)~Meanwhile, $\edgeComplex(2)$ shows one $\Hgroup_0$ component and one $\Hgroup_1$ cycle.}
    \label{fig:edge_complex}
    \vspace{-5pt}
\end{figure}

We first describe how the persistent homology information is extracted from an input graph (\autoref{sec:ph}). 
We then describe the use of this information for building fast initial graph layouts (\autoref{sec:layout}) and for highlighting important structures in the visualization (\autoref{sec:h1}).

The input is an undirected graph $G=(V,E)$ equipped with an edge weight $w: E \to \Rspace$. 
$w$ can be any real function that quantifies the edge importance. 
Recall the Jaccard index between a pair of sets $A$ and $B$ is defined to be $J(A,B) = {{|A \cap B|}\over{|A \cup B|}}$.   
In this paper, if $w$ is not known \emph{a priori}, $w(e)$ for an edge $e$ is assigned the Jaccard index between the 1-neighborhood of its nodes, as was also done in~\cite{suh2019persistent}. 

\subsection{Persistent Homology of a Graph}
\label{sec:ph}

We describe a novel approach to extract persistent homology features from a graph, leaving the discussions of the general theory to prior works (e.g.,~\cite{edelsbrunner2008persistent}). 
Previous approaches (e.g.,~\cite{hajij2018visual}) have relied upon mapping a graph to a metric space, computing a Vietoris-Rips complex, and extracting its 0-dimensional 
and 1-dimensional 
persistent homology as topological features. 
However, these approaches may be costly, $O
\left((|V|+|E|)^3\right)$ in the worse case, making them impractical on larger graphs. 
Furthermore, persistent homology identifies a class of cycles and while identifying the existence of such a class is well defined, determining which nodes specifically contribute to the cycle in the context of graphs is ambiguous~\cite{dey2018efficient}. In the following section, we provide an alternate strategy that resolves both of these issues.

Homology deals with the topological features of a space. In particular, given a space $\Xspace$, we are interested in extracting the \textit{0-dimensional}, $\Hgroup_0(\Xspace)$, and \textit{1-dimensional},
$\Hgroup_1(\Xspace)$, homology groups of $\Xspace$, which are the connected components and tunnels/cycles of the space, respectively. 

\begin{figure}[!t]
    \centering
  
    \begin{minipage}[b]{0.925\linewidth}
        \includegraphics[width=\linewidth]{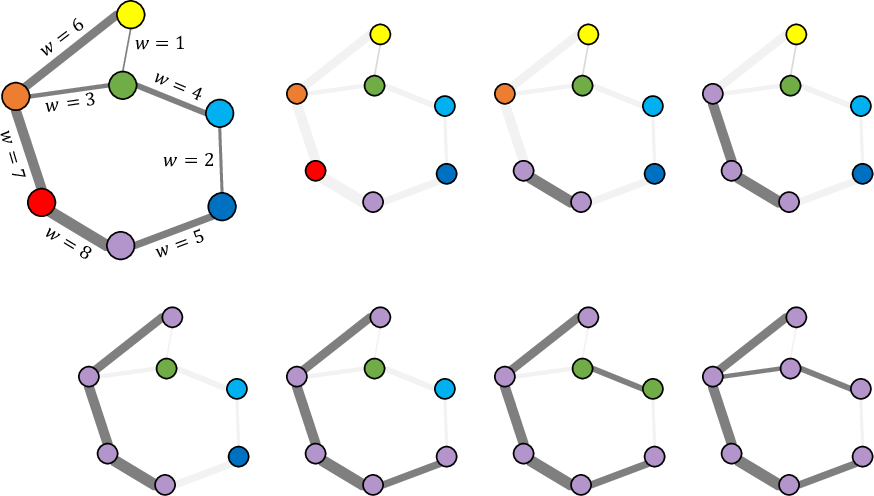}
    \end{minipage}    
    \begin{minipage}[b]{0pt}
        \hspace{-235pt}\subfloat[Input Graph]{\hspace{60pt}}
        \vspace{55pt}
    \end{minipage}    
    \begin{minipage}[b]{0pt}
        \hspace{-165pt}\subfloat[$\edgeComplex(\infty)$ $\Hgroup_0$/$\Hgroup_1$: 7/0]{\hspace{45pt}}
        \hspace{14pt}\subfloat[$\edgeComplex(8)$ $\Hgroup_0$/$\Hgroup_1$: 6/0]{\hspace{45pt}}
        \hspace{14pt}\subfloat[$\edgeComplex(7)$ $\Hgroup_0$/$\Hgroup_1$: 5/0]{\hspace{45pt}}
        \vspace{55pt}
    \end{minipage}
    \begin{minipage}[b]{0pt}
        \hspace{-225pt}\subfloat[$\edgeComplex(6)$ $\Hgroup_0$/$\Hgroup_1$: 4/0]{\hspace{45pt}}
        \hspace{14pt}\subfloat[$\edgeComplex(5)$ $\Hgroup_0$/$\Hgroup_1$: 3/0]{\hspace{45pt}}
        \hspace{14pt}\subfloat[$\edgeComplex(4)$ $\Hgroup_0$/$\Hgroup_1$: 2/0\label{fig:filtration:f}]{\hspace{45pt}}
        \hspace{14pt}\subfloat[$\edgeComplex(3)$ $\Hgroup_0$/$\Hgroup_1$: 1/0\label{fig:filtration:g}]{\hspace{45pt}}
        \vspace{-20pt}
    \end{minipage}    

    \vspace{25pt}
    \begin{minipage}[b]{0.375\linewidth}
        \subfloat[Barcode Visualization for the $\edgeComplex$ filtration\label{fig:filtration:barcode}]{\includegraphics[width=\linewidth]{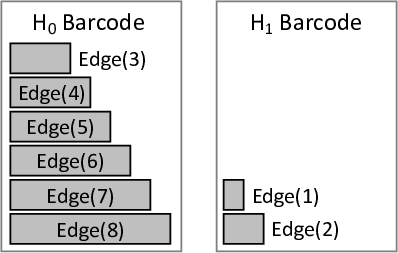}}
    \end{minipage}
    \hspace{5pt}
    \begin{minipage}[b]{0.505\linewidth}
        \centering
        \includegraphics[width=\linewidth]{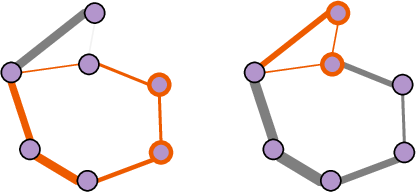}

        \vspace{-12pt}
        \subfloat[$\edgeComplex(2)$ $\Hgroup_0$/$\Hgroup_1$: 1/1\label{fig:cycles:a}]{\hspace{45pt}}
        \hspace{28pt}
        \subfloat[$\edgeComplex(1)$ $\Hgroup_0$/$\Hgroup_1$: 1/2\label{fig:cycles:b}]{\hspace{45pt}}
    \end{minipage}    
    
    \vspace{-4pt}
    \caption{(a)~Illustrating an $\edgeComplex$ filtration. (b)~The filtration begins with seven components (denoted by color), one per node, and no cycles. (c-h) As the filtration continues, the $\Hgroup_0$ components and $\Hgroup_1$ cycles are extracted at each step. (j)~When cycles form, they are extracted, shown in red. (k)~By the end, one component and two cycles remain. (i)~The resulting features are represented in a barcode visualization.}
    \label{fig:filtration}
    \vspace{-2pt}
\end{figure}

To identify the homology of a graph, we begin by describing the $\edgeComplex$ complex of a graph. 
Given a threshold $t$, for each edge $e_i$ in $G$  with a weight $w_i$, the $\edgeComplex$ complex is $\edgeComplex(t)=\{e_i \mid w_i\geq t\}$. In other words, the $\edgeComplex$ complex is the set of all edges whose weight is greater than or equal to the given threshold. 
For example, \autoref{fig:edge_complex:G} shows $\edgeComplex(4)$ and $\edgeComplex(2)$ of the graph, in \autoref{fig:edge_complex:E4} and \ref{fig:edge_complex:E2}, respectively. 

From an $\edgeComplex$ complex, its connected components ($\Hgroup_0$) and cycles ($\Hgroup_1$) can be efficiently extracted by a process that will be discussed in the forthcoming sections. 
However, extracting the homology of the graph from a single $\edgeComplex$ complex may fail to capture homology visible at different thresholds (e.g., see \autoref{fig:edge_complex}) and, therefore, requires careful selection of the threshold $t$.

Instead of selecting a single threshold $t$, we extract  $\Hgroup_0$ and $\Hgroup_1$ features of the graph across all thresholds using a multiscale notion of homology, called \emph{persistent homology}. 
Persistent homology is calculated by extracting a sequence of $\edgeComplex$ complexes, referred to as a \textit{filtration}.
We consider a finite sequence of decreasing thresholds, $\infty = t_0 \geq t_1 \geq \cdots \geq t_m=0$. A sequence of $\edgeComplex$ complexes, known as an $\edgeComplex$ filtration, is calculated and connected by inclusions,   
\[\edgeComplex(t_0) \to \edgeComplex(t_1) \to \cdots \to \edgeComplex(t_m).\] 
In other words, the $\edgeComplex$ complexes are subsets of one another, $\edgeComplex(t_i) \subseteq \edgeComplex(t_{i+1})$ for $0 \leq i \leq m-1$.  
The $\edgeComplex$ filtration can also be described as the upper-star filtration of the graph. 
$\Hgroup_0$ and $\Hgroup_1$ features are tracked across the various $\edgeComplex$ complexes in the filtration, 
\[\Hgroup(\edgeComplex(t_0)) \to \Hgroup(\edgeComplex(t_1)) \to \cdots \to \Hgroup(\edgeComplex(t_m)).\]

In the example filtration in \autoref{fig:filtration}, 
as the threshold decreases, $\Hgroup_0$ component features 
merge. 
The merging of two $\Hgroup_0$ features causes one feature to disappear in what is known as a \emph{death} event while the other feature continues to live. 
Consider the merging of the green and purple components in \autoref{fig:filtration:f}. 
In \autoref{fig:filtration:g}, the green component has died while the purple continues to live. 
In contrast, as the threshold decreases, new $\Hgroup_1$ cycle features 
appear in what are known as \emph{birth} events. Note the creation of cycles at $\edgeComplex(2)$ and $\edgeComplex(1)$ in \autoref{fig:cycles:a} and \ref{fig:cycles:b}, respectively. 
The birth and death events represent critical values that define the importance of a feature.

\subsubsection{Efficient Identification of \texorpdfstring{$\Hgroup_0$}{H0} Connected Components}

To calculate $\Hgroup_0$ features of a graph, Suh et al.~\cite{suh2019persistent} calculated the minimum spanning tree of a metric space representation of the graph by transforming edge weights into distances, e.g., $d(A,B)=1/w_{AB}$ (i.e., larger weights having smaller distances), which is inefficient on larger graphs, taking $O(|V|^2\log{|V|})$. 

Our choice of the $\edgeComplex$ filtration is a very specific one designed to capture features and increase efficiency. Using the $\edgeComplex$ filtration, the $\Hgroup_0$ information for the graph can be obtained by calculating the \textit{maximal spanning tree} of the graph, which is the spanning tree with edge weights greater than or equal to every other possible spanning tree. In calculating the maximal spanning tree of the graph, as edges are added to the tree, each edge $e_i$ represents an $\Hgroup_0$ death event at $w_i$ (i.e., merging of two connected components). 
We calculate the maximal spanning tree using Kruskal's algorithm~\cite{kruskal1956shortest}, selecting the largest weight edge instead of the smallest. The algorithm has a worst-case time complexity of $O(|E|\log{|E|})$. For non-negative weights, the resulting maximal spanning tree captures exactly the same structure as the prior minimal spanning tree of the metric space approach, only more efficiently. In addition, our maximal spanning tree approach captures meaningful features for negative and zero weight edges.

\subsubsection{Efficient Identification of \texorpdfstring{$\Hgroup_1$}{H1} Cycles}

In general persistent homology calculations, both detecting the existence of $\Hgroup_1$ features and extracting a representative cycle are computationally expensive, roughly $O(|V|^3)$~\cite{dey2018efficient}. 
The use of the $\edgeComplex$ filtration enables both detecting and extracting the $\Hgroup_1$ features much more efficiently within the limited context of graphs. To do this, we begin with an interesting observation in~\autoref{theorem:insertion}.

\begin{theorem}
\label{theorem:insertion}
    Given any spanning tree $S$ of a connected graph $G$, inserting any additional graph edge into $S$ creates a cycle.
\end{theorem}

\begin{proof} Since $S$ is a tree, it is acyclic, and any two nodes have a unique simple path between them. Therefore, if an edge is added between any two nodes in $S$, those nodes will now have two non-overlapping paths between them. Concatenating the edges of the two paths will create closed trail, i.e., a cycle, between them.
\end{proof}
 
As it turns out, this property enables efficient extraction of $\Hgroup_1$ features with the $\edgeComplex$ filtration. 
As the maximal spanning tree is calculated, an edge  $e_i$ that would be excluded from the spanning tree signifies the existence of an $\Hgroup_1$ cycle feature with a birth at $w_i$\footnote{Since 2-simplices (i.e., triangles) are not used, we do not track cycle death.}.

To extract the $\Hgroup_1$ cycle paths themselves, the unweighted shortest path is calculated between the endpoints of each edge, $e_i$, in the associated $\edgeComplex$ complex, $\edgeComplex(w_i)$. 
The shortest path is computed using Dijkstra's algorithm with worst-case time complexity $O((|V|+|E|)\log |V|)$. In practice, paths are short and generally fast to compute. Nevertheless, the shortest path needs to be calculated for every $\Hgroup_1$ cycle. 
Therefore, in practice, we defer calculating the complete cycle paths until the visualization needs the information. To further reduce the number of $\Hgroup_1$ features considered, cycles of length three are considered to be trivial and discarded\footnote{These are found by comparing the 1-neighborhood of the edge nodes.}.

While this approach will extract all $\Hgroup_1$ features of the $\edgeComplex$ complex, it will not extract all cycles of the graph.
Our $\Hgroup_1$ features are a specific type of cycle, where no chord within the cycle has a weight greater than or equal to the weight of all edges of the cycle. 
This type of cycle has a strong theoretical basis that is useful for many analysis tasks, but it may not be relevant for all such tasks. As we will show in our evaluation (see \autoref{sec.eval.h1}), these cycles are useful for many graphs. Nevertheless, adapting our approach to extracting other representative cycle types would further extend the utility of the approach.

\subsection{Using \texorpdfstring{$\Hgroup_0$}{H0} Features to Untangle Initial Graph Layouts}
\label{sec:layout}

Recent works (e.g., \cite{suh2019persistent, doraiswamy2020topomap}) have demonstrated the value of using $\Hgroup_0$ information in generating high-quality layouts of graphs and high-dimensional data, respectively. In contrast, we focus on quickly producing a good-quality layout that reflects the most important structures of the graph, as defined by persistent homology. We then utilize a D3.js's force-directed layout capabilities~\cite{bostock2011d3} to optimize the final layout.

Our algorithm works by laying out the graph using the maximal spanning tree. 
Inspired by early works on tidy tree drawing~\cite{wetherell1979tidy, reingold1981tidier, supowit1983complexity}, our approach has two main steps, with a focus on simplicity and efficiency. 
First, we generate an abstract layout of the maximal spanning tree that determines the distribution of space to nodes and subtrees of the graph. 
Then, we embed the tree into the drawing canvas using either a layered or radial scheme.

\subsubsection{Abstract Layout}

The first phase of the algorithm forms an abstract layout of the tree. 
The algorithm begins by selecting a node at random to serve as the root of the tree\footnote{We tested several other strategies, e.g., finding the most central node or the node with the most children. However, improvements were minimal, sometimes with a high additional cost over a randomly chosen node.}. The children of the selected root are then laid out hierarchically.
The algorithm recursively processes subtrees, subdividing the available space until all nodes have been visited. 
The available space is divided between children at each level based on the number of nodes in their respective subtrees. For example, in \autoref{fig:layout:abst}, the orange node is selected as the root. The leftmost subtree is allocated more space since it contains more nodes.

\subsubsection{Graph Embedding}

In the second part of the algorithm, the abstract layout is used to embed nodes into the drawing canvas using either a layered or radial layout.

\begin{figure}[!b]
    \centering
    \hfill
    \begin{minipage}[b]{0.925\linewidth}
        \includegraphics[width=\linewidth]{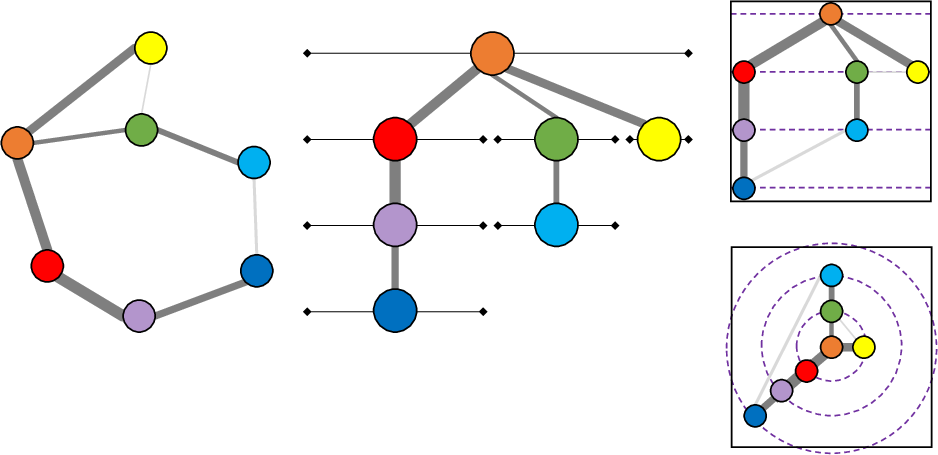}
    \end{minipage}    
    \begin{minipage}[b]{0pt}
        \hspace{-234pt}\subfloat[Graph with maximal spanning tree edges in dark grey\label{fig:layout:mst}]{\hspace{66pt}}
        \hspace{1pt}
        \hspace{5pt}\subfloat[Abstract layout calculates node depth and distributes horizontal space by subtree sizes\label{fig:layout:abst}]{\hspace{95pt}}
        \vspace{0pt}
    \end{minipage}    
    \begin{minipage}[b]{0pt}
        \hspace{-63pt}\subfloat[Layered Layout\label{fig:layout:layered}]{\hspace{60pt}}
        \vspace{52pt}
    \end{minipage}    
    \begin{minipage}[b]{0pt}
        \hspace{-64pt}\subfloat[Radial Layout\label{fig:layout:radial}]{\hspace{60pt}}
        \vspace{-10pt}
    \end{minipage}
    
    \vspace{8pt}
    \caption{An illustration of the initial layout schemes. The input graph~(a) first has an abstract layout formed in~(b) and is then mapped into a layered layout~(c) or a radial layout~(d). After the initial layout, any force-directed layout can be used.}
    \label{fig:layout}
\end{figure}

\para{Layered Layout} The first version of our layout algorithm maps the abstract layout directly to the available drawing canvas in a layered tree visualization. Specifically, the horizontal space on the canvas is mapped to the width of the tree in the abstract layout, and the vertical space is mapped to the height of the tree in the abstract layout. For example, in \autoref{fig:layout:layered}, the abstract layout from \autoref{fig:layout:abst} is mapped to the available drawing canvas.

\para{Radial Layout} The second version adopts a radial layout for the tree. 
The width of the abstract tree is mapped to an angle in the unit circle, and each layer of the tree occupies an increasing radius. For example, in \autoref{fig:layout:radial}, the abstract layout from \autoref{fig:layout:abst} is mapped to polar coordinates in the drawing canvas.

For either layered or radial layout, after the initial layout is formulated, a standard force-directed layout is applied to the entire graph.

\subsection{Interactive Untangling with Persistent Homology}
\label{sec:h1}

\subsubsection{Visualization of Persistent Homology Features}

We visualize the $\Hgroup_0$ and $\Hgroup_1$ persistent homology using a visualization based upon a persistence barcode (see \autoref{fig:teaser2}), a standard tool of persistent homology. 
For this visualization, a barcode is associated with a set of $\Hgroup_0$ or $\Hgroup_1$ features. 
For each barcode, a bar represents a single topological feature. 
Its length is proportional to the death or birth time of the associated $\Hgroup_0$ or $\Hgroup_1$ features, respectively.

\subsubsection{Interacting with \texorpdfstring{$\Hgroup_0$}{H0} Components}

Similar to the $\Hgroup_0$ interactions in~\cite{suh2019persistent}, when $\Hgroup_0$ bars are selected in the barcode, a strong attractive force is created (i.e., a spring-like force) between the nodes of the associated edge from the spanning tree. 
Our selection offers a filtering slider (see demo\footnote{Demo at \demoURL.}) to select multiple features simultaneously.

\subsubsection{Interacting with \texorpdfstring{$\Hgroup_1$}{H1} Cycles}

For interacting with $\Hgroup_1$ cycles, we offer two modalities. The first is a highlighting modality. As the user's mouse goes over the bar for a given cycle, that cycle is extracted and highlighted in the graph. \autoref{fig:teaser2} shows two examples, each highlighting two cycles.

\begin{figure}[!bt]
    \centering
    
    \begin{minipage}[b]{0.95\linewidth}
        \includegraphics[width=\linewidth]{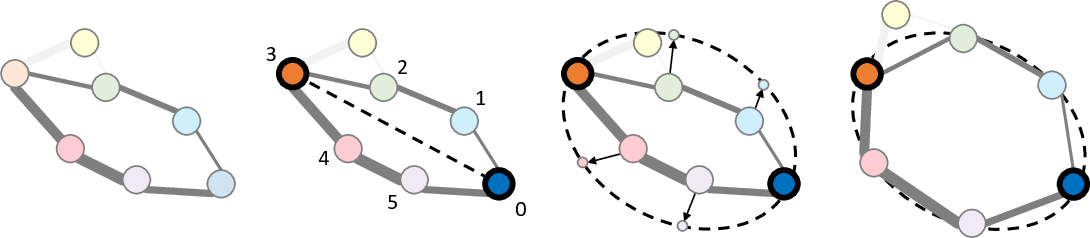}
    \end{minipage}    
    \begin{minipage}[b]{0pt}  
    \hspace{-205pt}
    \subfloat[]{\hspace{20pt}}
    \hspace{40pt}
    \subfloat[]{\hspace{20pt}}
    \hspace{40pt}
    \subfloat[]{\hspace{20pt}}
    \hspace{45pt}
    \subfloat[]{\hspace{5pt}}
    \vspace{38pt}
    \end{minipage}

    \caption{Illustration of the elliptical force applied to an $\Hgroup_1$ feature. (a)~A cycle (dark gray) is selected in the barcode. (b)~The two nodes in the cycle with the greatest euclidean distance in the visualization identify the major axis of the ellipse, whereas  the minor axis diameter is calculated by a user-specified aspect ratio. (c)~The nodes on the cycle are parameterized (i.e., ordered) and their target locations on the ellipse are identified. (d)~Forces are applied to the nodes on the cycle to move them toward their target locations.}
    \label{fig:h1_force}
\end{figure}

The second modality, triggered when a user clicks a feature in the barcode, uses information about the cycle to add a new elliptical force to the cycle nodes in the force-directed layout. 
The approach (see \autoref{fig:h1_force} for an illustration) first takes the nodes of the cycle and identifies the two nodes with the largest Euclidean distance from one another in the visualization. Those nodes serve as the major axis of the ellipse and determine the diameter of the major axis. The minor axis diameter is determined by a user-selectable aspect ratio. Once the elliptical shape is calculated, the cycle nodes are parameterized, i.e., ordered around the ellipse, to select a target location. The ordering step is quite important, as it enables powerful modifications, e.g., untangling cycles. Finally, a strong force is added to attract the nodes to their target locations. However, due to the over-constraint of force-directed layouts, this new force does not guarantee nodes will end up on the ellipse. \autoref{fig:teaser2} shows examples of imposing the elliptical shape on the cycles and examples of the forces untangling cycles in the graph.

\section{Evaluation}

To demonstrate the efficacy of our approach, we provide a two-phase evaluation. In \autoref{sec:results:initial}, we first evaluate our graph initialization approach using $\Hgroup_0$ persistent homology in terms of layout quality and rate of convergence. Second, in \autoref{sec.eval.h1}, we evaluate the layout quality of using $\Hgroup_1$ persistent homology to modify the forces of a force-directed layout. 
For all comparisons, we primarily compare to the state-of-the-practice force-directed layout provided by D3.js~\cite{bostock2011d3}, which uses a random initialization. Furthermore, in \autoref{sec:eval:other_algos}, we compare to static graph visualizations, including the \emph{neato}~\cite{kamada1989algorithm}, \emph{fdp}~\cite{eades1984heuristic}, and \emph{sfpd}~\cite{hu2005efficient} algorithms coming from Graphviz~\cite{ellson2001graphviz}.

\renewcommand{\arraystretch}{1.1}

\begin{table*}[!ht]
    \centering
    \caption{Table of \textit{dense} datasets. See \autoref{sec.eval.metrics} for details about metrics. Our discussion focuses on LCMC metrics (in blue). In these cells, bold indicates a smaller time for $T_{LCMC}$, a lower iteration count to convergence for $C_{LCMC}$, or a value 0.005 larger for $Q_{LCMC}$.}
    \label{tab:dense}
    
    \vspace{-6pt}
    \resizebox{0.975\linewidth}{!}{%
    \begin{tabular}{c!{\color{hlineColor}\vrule}c!{\color{hlineColor}\vrule}c!{\color{hlineColor}\vrule}c!{\color{hlineColor}\vrule}c|c!{\color{hlineColor}\vrule}c|c!{\color{hlineColor}\vrule}c|c!{\color{hlineColor}\vrule}c!{\color{hlineColor}\vrule}c!{\color{hlineColor}\vrule}c!{\color{hlineColor}\vrule}c!{\color{hlineColor}\vrule}c|C{2.35cm}!{\color{hlineColor}\vrule}p{6.45cm}}
    
    \multirow{2}{*}{Dataset} & \multirow{2}{*}{$|V|$} & \multirow{2}{*}{$|E|$} & Avg & \multirow{2}{*}{Layout} & \multirow{2}{*}{$T_{IT}$} & \multirow{2}{*}{$T_{AIT}$} & \multirow{2}{*}{$T_{LCMC}$} & \multirow{2}{*}{$C_{LCMC}$} & \multirow{2}{*}{$Q_{LCMC}$} & \multirow{2}{*}{$Q_{trust}$} & \multirow{2}{*}{$Q_{conv}$} & \multirow{2}{*}{$Q_{EC}$} & \multirow{2}{*}{$Q_{CA}$} & \multirow{2}{*}{$Q_{MAR}$} & \multirow{2}{*}{Source} & \multirow{2}{*}{Description} \\
     &  & & ECC & & & & & & & & & & & & \\
     \hline\hline

\multirow{2}{*}{\textsc{aves sparrow social}} & \multirow{2}{*}{31} & \multirow{2}{*}{211} & \multirow{2}{*}{3.5} & \cellcolor{white}Random & \cellcolor{white}13.5 ms & \cellcolor{white}10.9 ms & \cellcolor{white}1.19 s & \cellcolor{LCMCColorLight}108 & \cellcolor{LCMCColorLight}0.350 & \cellcolor{white}0.885 & \cellcolor{white}0.880 & \cellcolor{white}0.764 & \cellcolor{white}0.757 & \cellcolor{white}0.081 & Network & \multirow{2}{6.45cm}{Nodes represents individual free range birds, and edges are a proximity-based association index.}\\
 & & & & \cellcolor{ourMethodColor}Layered & \cellcolor{ourMethodColor}14.1 ms & \cellcolor{ourMethodColor}11.3 ms & \cellcolor{ourMethodColor}796 ms & \cellcolor{LCMCColorDark}\textbf{69} & \cellcolor{LCMCColorDark}\textbf{0.416} & \cellcolor{ourMethodColor}0.928 & \cellcolor{ourMethodColor}0.926 & \cellcolor{ourMethodColor}0.792 & \cellcolor{ourMethodColor}0.766 & \cellcolor{ourMethodColor}0.140 & Repository & \\
\arrayrulecolor{hlineColor}\hline
\multirow{2}{*}{\textsc{barbasi-albert (50,40)}} & \multirow{2}{*}{50} & \multirow{2}{*}{400} & \multirow{2}{*}{2.0} & \cellcolor{white}Random & \cellcolor{white}18.7 ms & \cellcolor{white}10.8 ms & \cellcolor{white}1.53 s & \cellcolor{LCMCColorLight}140 & \cellcolor{LCMCColorLight}0.145 & \cellcolor{white}0.657 & \cellcolor{white}0.677 & \cellcolor{white}0.546 & \cellcolor{white}0.714 & \cellcolor{white}0.077 & \multirow{2}{*}{NetworkX} & \multirow{2}{6.45cm}{A graph that satisifies Barbasi-Albert preferential attachment model \cite{barabasi1999emergence}.}\\
 & & & & \cellcolor{ourMethodColor}Radial & \cellcolor{ourMethodColor}17.8 ms & \cellcolor{ourMethodColor}11.8 ms & \cellcolor{ourMethodColor}29.6 ms & \cellcolor{LCMCColorDark}\textbf{1} & \cellcolor{LCMCColorDark}\textbf{0.186} & \cellcolor{ourMethodColor}0.691 & \cellcolor{ourMethodColor}0.702 & \cellcolor{ourMethodColor}0.575 & \cellcolor{ourMethodColor}0.725 & \cellcolor{ourMethodColor}0.029 &  & \\
\arrayrulecolor{hlineColor}\hline
\multirow{2}{*}{\textsc{bio-celegans}} & \multirow{2}{*}{453} & \multirow{2}{*}{2025} & \multirow{2}{*}{5.2} & \cellcolor{white}Random & \cellcolor{white}110 ms & \cellcolor{white}47.1 ms & \cellcolor{white}3.08 s & \cellcolor{LCMCColorLight}63 & \cellcolor{LCMCColorLight}0.219 & \cellcolor{white}0.867 & \cellcolor{white}0.761 & \cellcolor{white}0.919 & \cellcolor{white}0.719 & \cellcolor{white}0.196 & Network & \multirow{2}{6.45cm}{A metabolic network where substrates are nodes and edges are metabolic reactions between them.}\\
 & & & & \cellcolor{ourMethodColor}Layered & \cellcolor{ourMethodColor}102 ms & \cellcolor{ourMethodColor}47.3 ms & \cellcolor{ourMethodColor}1.09 s & \cellcolor{LCMCColorDark}\textbf{21} & \cellcolor{LCMCColorDark}\textbf{0.229} & \cellcolor{ourMethodColor}0.867 & \cellcolor{ourMethodColor}0.776 & \cellcolor{ourMethodColor}0.915 & \cellcolor{ourMethodColor}0.720 & \cellcolor{ourMethodColor}0.181 & Repository & \\
\arrayrulecolor{hlineColor}\hline
\multirow{2}{*}{\textsc{bn-mouse-visual-cortex-2}} & \multirow{2}{*}{193} & \multirow{2}{*}{214} & \multirow{2}{*}{6.4} & \cellcolor{white}Random & \cellcolor{white}16.6 ms & \cellcolor{white}14.0 ms & \cellcolor{white}282 ms & \cellcolor{LCMCColorLight}19 & \cellcolor{LCMCColorLight}0.409 & \cellcolor{white}0.958 & \cellcolor{white}0.960 & \cellcolor{white}0.997 & \cellcolor{white}0.746 & \cellcolor{white}0.931 & Network & \multirow{2}{6.45cm}{Mouse brain network where nodes are locations and edges are unweighted fiber tracks between them.}\\
 & & & & \cellcolor{ourMethodColor}Layered & \cellcolor{ourMethodColor}17.2 ms & \cellcolor{ourMethodColor}14.2 ms & \cellcolor{ourMethodColor}201 ms & \cellcolor{LCMCColorDark}\textbf{13} & \cellcolor{LCMCColorDark}\textbf{0.425} & \cellcolor{ourMethodColor}0.967 & \cellcolor{ourMethodColor}0.963 & \cellcolor{ourMethodColor}0.999 & \cellcolor{ourMethodColor}0.837 & \cellcolor{ourMethodColor}0.930 & Repository & \\
\arrayrulecolor{hlineColor}\hline
\multirow{2}{*}{\textsc{chordal cycle (90)}} & \multirow{2}{*}{90} & \multirow{2}{*}{180} & \multirow{2}{*}{5.0} & \cellcolor{white}Random & \cellcolor{white}11.9 ms & \cellcolor{white}13.5 ms & \cellcolor{white}1.19 s & \cellcolor{LCMCColorLight}87 & \cellcolor{LCMCColorLight}0.335 & \cellcolor{white}0.897 & \cellcolor{white}0.818 & \cellcolor{white}0.938 & \cellcolor{white}0.738 & \cellcolor{white}0.467 & \multirow{2}{*}{NetworkX} & \multirow{2}{6.45cm}{Graph where all cylces contain a chord.}\\
 & & & & \cellcolor{ourMethodColor}Radial & \cellcolor{ourMethodColor}14.6 ms & \cellcolor{ourMethodColor}13.6 ms & \cellcolor{ourMethodColor}1.02 s & \cellcolor{LCMCColorDark}\textbf{74} & \cellcolor{LCMCColorDark}\textbf{0.362} & \cellcolor{ourMethodColor}0.905 & \cellcolor{ourMethodColor}0.840 & \cellcolor{ourMethodColor}0.940 & \cellcolor{ourMethodColor}0.735 & \cellcolor{ourMethodColor}0.477 &  & \\
\arrayrulecolor{hlineColor}\hline
\multirow{2}{*}{\textsc{davis southern women }} & \multirow{2}{*}{32} & \multirow{2}{*}{89} & \multirow{2}{*}{3.7} & \cellcolor{white}Random & \cellcolor{white}6.3 ms & \cellcolor{white}14.8 ms & \cellcolor{white}1.53 s & \cellcolor{LCMCColorLight}103 & \cellcolor{LCMCColorLight}\textbf{0.382} & \cellcolor{white}0.921 & \cellcolor{white}0.903 & \cellcolor{white}0.880 & \cellcolor{white}0.736 & \cellcolor{white}0.244 & \multirow{2}{*}{NetworkX} & \multirow{2}{6.45cm}{A graph of observered attendance at 14 social events by 18 southern women in 1930's.}\\
 & & & & \cellcolor{ourMethodColor}Radial & \cellcolor{ourMethodColor}7.3 ms & \cellcolor{ourMethodColor}14.5 ms & \cellcolor{ourMethodColor}528 ms & \cellcolor{LCMCColorDark}\textbf{36} & \cellcolor{LCMCColorDark}0.371 & \cellcolor{ourMethodColor}0.917 & \cellcolor{ourMethodColor}0.903 & \cellcolor{ourMethodColor}0.879 & \cellcolor{ourMethodColor}0.749 & \cellcolor{ourMethodColor}0.180 &  & \\
\arrayrulecolor{hlineColor}\hline
\multirow{2}{*}{\textsc{dolphin social}} & \multirow{2}{*}{62} & \multirow{2}{*}{159} & \multirow{2}{*}{6.5} & \cellcolor{white}Random & \cellcolor{white}10.4 ms & \cellcolor{white}12.8 ms & \cellcolor{white}1.38 s & \cellcolor{LCMCColorLight}107 & \cellcolor{LCMCColorLight}0.474 & \cellcolor{white}0.948 & \cellcolor{white}0.947 & \cellcolor{white}0.955 & \cellcolor{white}0.748 & \cellcolor{white}0.397 & Network & \multirow{2}{6.45cm}{A social interaction network between dolphins.}\\
 & & & & \cellcolor{ourMethodColor}Radial & \cellcolor{ourMethodColor}10.4 ms & \cellcolor{ourMethodColor}13.3 ms & \cellcolor{ourMethodColor}1.02 s & \cellcolor{LCMCColorDark}\textbf{76} & \cellcolor{LCMCColorDark}\textbf{0.486} & \cellcolor{ourMethodColor}0.948 & \cellcolor{ourMethodColor}0.944 & \cellcolor{ourMethodColor}0.947 & \cellcolor{ourMethodColor}0.767 & \cellcolor{ourMethodColor}0.347 & Repository & \\
\arrayrulecolor{hlineColor}\hline
\multirow{2}{*}{\textsc{dorogovtsev-goltsev-mendes (5)}} & \multirow{2}{*}{123} & \multirow{2}{*}{243} & \multirow{2}{*}{4.3} & \cellcolor{white}Random & \cellcolor{white}14.4 ms & \cellcolor{white}12.5 ms & \cellcolor{white}1.31 s & \cellcolor{LCMCColorLight}104 & \cellcolor{LCMCColorLight}0.377 & \cellcolor{white}0.930 & \cellcolor{white}0.860 & \cellcolor{white}0.986 & \cellcolor{white}0.722 & \cellcolor{white}0.534 & \multirow{2}{*}{NetworkX} & \multirow{2}{6.45cm}{A graph that satisifies Dorogovtsev and Mendes algorithm \cite{dorogovtsev2003spectra}.}\\
 & & & & \cellcolor{ourMethodColor}Radial & \cellcolor{ourMethodColor}31.2 ms & \cellcolor{ourMethodColor}12.7 ms & \cellcolor{ourMethodColor}730 ms & \cellcolor{LCMCColorDark}\textbf{55} & \cellcolor{LCMCColorDark}\textbf{0.446} & \cellcolor{ourMethodColor}0.960 & \cellcolor{ourMethodColor}0.912 & \cellcolor{ourMethodColor}0.996 & \cellcolor{ourMethodColor}0.734 & \cellcolor{ourMethodColor}0.563 &  & \\
\arrayrulecolor{hlineColor}\hline
\multirow{2}{*}{\textsc{duplicate divergence }} & \multirow{2}{*}{50} & \multirow{2}{*}{99} & \multirow{2}{*}{4.6} & \cellcolor{white}Random & \cellcolor{white}7.7 ms & \cellcolor{white}14.1 ms & \cellcolor{white}1.29 s & \cellcolor{LCMCColorLight}91 & \cellcolor{LCMCColorLight}0.415 & \cellcolor{white}0.925 & \cellcolor{white}0.913 & \cellcolor{white}0.907 & \cellcolor{white}0.747 & \cellcolor{white}0.483 & \multirow{2}{*}{NetworkX} & \multirow{2}{6.45cm}{Starting with a small graph, nodes and some edges repeatedly duplicated.}\\
 & & & & \cellcolor{ourMethodColor}Layered & \cellcolor{ourMethodColor}8.2 ms & \cellcolor{ourMethodColor}14.1 ms & \cellcolor{ourMethodColor}318 ms & \cellcolor{LCMCColorDark}\textbf{22} & \cellcolor{LCMCColorDark}0.414 & \cellcolor{ourMethodColor}0.919 & \cellcolor{ourMethodColor}0.916 & \cellcolor{ourMethodColor}0.931 & \cellcolor{ourMethodColor}0.733 & \cellcolor{ourMethodColor}0.488 &  & \\
\arrayrulecolor{hlineColor}\hline
\multirow{2}{*}{\textsc{enron email}} & \multirow{2}{*}{143} & \multirow{2}{*}{623} & \multirow{2}{*}{6.1} & \cellcolor{white}Random & \cellcolor{white}28.1 ms & \cellcolor{white}17.4 ms & \cellcolor{white}1.02 s & \cellcolor{LCMCColorLight}57 & \cellcolor{LCMCColorLight}\textbf{0.390} & \cellcolor{white}0.921 & \cellcolor{white}0.888 & \cellcolor{white}0.948 & \cellcolor{white}0.725 & \cellcolor{white}0.234 & Network & \multirow{2}{6.45cm}{An email network collected from the Enron energy accounting scandal of 2001.}\\
 & & & & \cellcolor{ourMethodColor}Layered & \cellcolor{ourMethodColor}27.7 ms & \cellcolor{ourMethodColor}17.9 ms & \cellcolor{ourMethodColor}726 ms & \cellcolor{LCMCColorDark}\textbf{39} & \cellcolor{LCMCColorDark}0.385 & \cellcolor{ourMethodColor}0.928 & \cellcolor{ourMethodColor}0.873 & \cellcolor{ourMethodColor}0.949 & \cellcolor{ourMethodColor}0.725 & \cellcolor{ourMethodColor}0.244 & Repository & \\
\arrayrulecolor{hlineColor}\hline
\multirow{2}{*}{\textsc{hic 1k net 6}} & \multirow{2}{*}{4581} & \multirow{2}{*}{284924} & \multirow{2}{*}{4.2} & \cellcolor{white}Random & \cellcolor{white}12.4 s & \cellcolor{white}10.9 s & \cellcolor{LCMCColorLight}12.3 min & \cellcolor{LCMCColorLight}112 & \cellcolor{LCMCColorLight}0.277 & \cellcolor{white}0.991 & \cellcolor{white}0.970 & \cellcolor{white}-- & \cellcolor{white}-- & \cellcolor{white}-- & \multirow{2}{*}{BioSNAP} & \multirow{2}{6.45cm}{Nodes are gnomic regions and edges are normalized contacts between regions.}\\
 & & & & \cellcolor{ourMethodColor}Radial & \cellcolor{ourMethodColor}16.0 s & \cellcolor{ourMethodColor}10.9 s & \cellcolor{LCMCColorDark}\textbf{6.2 min} & \cellcolor{LCMCColorDark}\textbf{55} & \cellcolor{LCMCColorDark}\textbf{0.303} & \cellcolor{ourMethodColor}0.993 & \cellcolor{ourMethodColor}0.985 & \cellcolor{ourMethodColor}-- & \cellcolor{ourMethodColor}-- & \cellcolor{ourMethodColor}-- &  & \\
\arrayrulecolor{hlineColor}\hline
\multirow{2}{*}{\textsc{les miserables}} & \multirow{2}{*}{77} & \multirow{2}{*}{254} & \multirow{2}{*}{4.1} & \cellcolor{white}Random & \cellcolor{white}12.9 ms & \cellcolor{white}11.0 ms & \cellcolor{white}770 ms & \cellcolor{LCMCColorLight}69 & \cellcolor{LCMCColorLight}\textbf{0.444} & \cellcolor{white}0.930 & \cellcolor{white}0.888 & \cellcolor{white}0.948 & \cellcolor{white}0.750 & \cellcolor{white}0.407 & \multirow{2}{*}{NetworkX} & \multirow{2}{6.45cm}{The co-appearance of charecters in chapters of Les Miserables.}\\
 & & & & \cellcolor{ourMethodColor}Layered & \cellcolor{ourMethodColor}12.8 ms & \cellcolor{ourMethodColor}11.7 ms & \cellcolor{ourMethodColor}340 ms & \cellcolor{LCMCColorDark}\textbf{28} & \cellcolor{LCMCColorDark}0.424 & \cellcolor{ourMethodColor}0.920 & \cellcolor{ourMethodColor}0.873 & \cellcolor{ourMethodColor}0.945 & \cellcolor{ourMethodColor}0.759 & \cellcolor{ourMethodColor}0.417 &  & \\
\arrayrulecolor{hlineColor}\hline
\multirow{2}{*}{\textsc{movies}} & \multirow{2}{*}{101} & \multirow{2}{*}{192} & \multirow{2}{*}{7.0} & \cellcolor{white}Random & \cellcolor{white}27.0 ms & \cellcolor{white}12.8 ms & \cellcolor{white}1.14 s & \cellcolor{LCMCColorLight}87 & \cellcolor{LCMCColorLight}\textbf{0.298} & \cellcolor{white}0.862 & \cellcolor{white}0.819 & \cellcolor{white}0.932 & \cellcolor{white}0.732 & \cellcolor{white}0.483 & \multirow{2}{*}{\cite{de2018exploratory}} & \multirow{2}{6.45cm}{Collabaration between 40 Hollywood composers and 61 producers between 1964 and 1976.}\\
 & & & & \cellcolor{ourMethodColor}Radial & \cellcolor{ourMethodColor}13.1 ms & \cellcolor{ourMethodColor}12.9 ms & \cellcolor{ourMethodColor}685 ms & \cellcolor{LCMCColorDark}\textbf{52} & \cellcolor{LCMCColorDark}0.280 & \cellcolor{ourMethodColor}0.859 & \cellcolor{ourMethodColor}0.804 & \cellcolor{ourMethodColor}0.939 & \cellcolor{ourMethodColor}0.754 & \cellcolor{ourMethodColor}0.471 &  & \\
\arrayrulecolor{hlineColor}\hline
\multirow{2}{*}{\textsc{smith}} & \multirow{2}{*}{2970} & \multirow{2}{*}{97133} & \multirow{2}{*}{4.8} & \cellcolor{white}Random & \cellcolor{white}4.98 s & \cellcolor{white}3.91 s & \cellcolor{LCMCColorLight}114 s & \cellcolor{LCMCColorLight}28 & \cellcolor{LCMCColorLight}0.045 & \cellcolor{white}0.882 & \cellcolor{white}0.705 & \cellcolor{white}-- & \cellcolor{white}-- & \cellcolor{white}-- & \multirow{2}{*}{Facebook 100} & \multirow{2}{6.45cm}{Nodes are students from Smith college and edges are the friendships between them.}\\
 & & & & \cellcolor{ourMethodColor}Radial & \cellcolor{ourMethodColor}5.62 s & \cellcolor{ourMethodColor}3.73 s & \cellcolor{LCMCColorDark}\textbf{20.5 s} & \cellcolor{LCMCColorDark}\textbf{4} & \cellcolor{LCMCColorDark}0.045 & \cellcolor{ourMethodColor}0.887 & \cellcolor{ourMethodColor}0.698 & \cellcolor{ourMethodColor}-- & \cellcolor{ourMethodColor}-- & \cellcolor{ourMethodColor}-- &  & \\
\arrayrulecolor{hlineColor}\hline
\multirow{2}{*}{\textsc{train bombing}} & \multirow{2}{*}{64} & \multirow{2}{*}{243} & \multirow{2}{*}{4.6} & \cellcolor{white}Random & \cellcolor{white}31.2 ms & \cellcolor{white}11.9 ms & \cellcolor{white}1.84 s & \cellcolor{LCMCColorLight}152 & \cellcolor{LCMCColorLight}0.393 & \cellcolor{white}0.912 & \cellcolor{white}0.887 & \cellcolor{white}0.898 & \cellcolor{white}0.746 & \cellcolor{white}0.311 & \multirow{2}{*}{NetworkX} & \multirow{2}{6.45cm}{Nodes are individuals involved in 2004 madrid train bombing and edge are for prior known relations.}\\
 & & & & \cellcolor{ourMethodColor}Radial & \cellcolor{ourMethodColor}12.8 ms & \cellcolor{ourMethodColor}11.6 ms & \cellcolor{ourMethodColor}174 ms & \cellcolor{LCMCColorDark}\textbf{14} & \cellcolor{LCMCColorDark}\textbf{0.486} & \cellcolor{ourMethodColor}0.954 & \cellcolor{ourMethodColor}0.930 & \cellcolor{ourMethodColor}0.928 & \cellcolor{ourMethodColor}0.751 & \cellcolor{ourMethodColor}0.304 &  & \\
\arrayrulecolor{hlineColor}\hline
\multirow{2}{*}{\textsc{usair 97}} & \multirow{2}{*}{332} & \multirow{2}{*}{2126} & \multirow{2}{*}{5.1} & \cellcolor{white}Random & \cellcolor{white}101 ms & \cellcolor{white}41.3 ms & \cellcolor{white}2.21 s & \cellcolor{LCMCColorLight}\textbf{51} & \cellcolor{LCMCColorLight}\textbf{0.323} & \cellcolor{white}0.938 & \cellcolor{white}0.851 & \cellcolor{white}0.876 & \cellcolor{white}0.711 & \cellcolor{white}0.301 & Network & \multirow{2}{6.45cm}{A weighted graph of the air traffic between airports in the US in 1997.}\\
 & & & & \cellcolor{ourMethodColor}Layered & \cellcolor{ourMethodColor}108 ms & \cellcolor{ourMethodColor}41.2 ms & \cellcolor{ourMethodColor}3.69 s & \cellcolor{LCMCColorDark}87 & \cellcolor{LCMCColorDark}0.307 & \cellcolor{ourMethodColor}0.926 & \cellcolor{ourMethodColor}0.820 & \cellcolor{ourMethodColor}0.864 & \cellcolor{ourMethodColor}0.712 & \cellcolor{ourMethodColor}0.301 & Repository & \\
\arrayrulecolor{hlineColor}\hline

    \end{tabular}}
\end{table*}

\begin{table*}[!ht]
    \centering
    \caption{Table of \textit{sparse} datasets. See \autoref{tab:dense} for a description.}
    \label{tab:sparse}
    
    \vspace{-6pt}
    \resizebox{0.975\linewidth}{!}{%
    \begin{tabular}{c!{\color{hlineColor}\vrule}c!{\color{hlineColor}\vrule}c!{\color{hlineColor}\vrule}c!{\color{hlineColor}\vrule}c|c!{\color{hlineColor}\vrule}c|c!{\color{hlineColor}\vrule}c|c!{\color{hlineColor}\vrule}c!{\color{hlineColor}\vrule}c!{\color{hlineColor}\vrule}c!{\color{hlineColor}\vrule}c!{\color{hlineColor}\vrule}c|C{2.35cm}!{\color{hlineColor}\vrule}p{6.45cm}}
    
    \multirow{2}{*}{Dataset} & \multirow{2}{*}{$|V|$} & \multirow{2}{*}{$|E|$} & Avg & \multirow{2}{*}{Layout} & \multirow{2}{*}{$T_{IT}$} & \multirow{2}{*}{$T_{AIT}$} & \multirow{2}{*}{$T_{LCMC}$} & \multirow{2}{*}{$C_{LCMC}$} & \multirow{2}{*}{$Q_{LCMC}$} & \multirow{2}{*}{$Q_{trust}$} & \multirow{2}{*}{$Q_{conv}$} & \multirow{2}{*}{$Q_{EC}$} & \multirow{2}{*}{$Q_{CA}$} & \multirow{2}{*}{$Q_{MAR}$} & \multirow{2}{*}{Source} & \multirow{2}{*}{Description} \\
     &  & & ECC & & & & & & & & & & & & \\
     \hline\hline

\multirow{2}{*}{\textsc{airport}} & \multirow{2}{*}{2896} & \multirow{2}{*}{15641} & \multirow{2}{*}{10.2} & \cellcolor{white}Random & \cellcolor{white}908 ms & \cellcolor{white}902 ms & \cellcolor{LCMCColorLight}69.5 s & \cellcolor{LCMCColorLight}76 & \cellcolor{LCMCColorLight}0.245 & \cellcolor{white}0.948 & \cellcolor{white}0.849 & \cellcolor{white}-- & \cellcolor{white}-- & \cellcolor{white}-- & \multirow{2}{*}{Openflights.org} & \multirow{2}{6.45cm}{Airports are represented as nodes with edges counting the number of routes between airports.}\\
 & & & & \cellcolor{ourMethodColor}Layered & \cellcolor{ourMethodColor}1.1 s & \cellcolor{ourMethodColor}893 ms & \cellcolor{LCMCColorDark}\textbf{34.2 s} & \cellcolor{LCMCColorDark}\textbf{37} & \cellcolor{LCMCColorDark}\textbf{0.253} & \cellcolor{ourMethodColor}0.962 & \cellcolor{ourMethodColor}0.870 & \cellcolor{ourMethodColor}-- & \cellcolor{ourMethodColor}-- & \cellcolor{ourMethodColor}-- &  & \\
\arrayrulecolor{hlineColor}\hline
\multirow{2}{*}{\textsc{balanced tree (3,6)}} & \multirow{2}{*}{1093} & \multirow{2}{*}{1092} & \multirow{2}{*}{11.5} & \cellcolor{white}Random & \cellcolor{white}92.9 ms & \cellcolor{white}63.8 ms & \cellcolor{white}7.55 s & \cellcolor{LCMCColorLight}117 & \cellcolor{LCMCColorLight}0.221 & \cellcolor{white}0.808 & \cellcolor{white}0.676 & \cellcolor{white}0.995 & \cellcolor{white}0.702 & \cellcolor{white}0.752 & \multirow{2}{*}{NetworkX} & \multirow{2}{6.45cm}{This dataset represents a balanced tree with 3 children and a height of 6.}\\
 & & & & \cellcolor{ourMethodColor}Radial & \cellcolor{ourMethodColor}109 ms & \cellcolor{ourMethodColor}74.5 ms & \cellcolor{ourMethodColor}6.07 s & \cellcolor{LCMCColorDark}\textbf{80} & \cellcolor{LCMCColorDark}\textbf{0.374} & \cellcolor{ourMethodColor}0.982 & \cellcolor{ourMethodColor}0.940 & \cellcolor{ourMethodColor}1.000 & \cellcolor{ourMethodColor}0.671 & \cellcolor{ourMethodColor}0.779 &  & \\
\arrayrulecolor{hlineColor}\hline
\multirow{2}{*}{\textsc{barbell}} & \multirow{2}{*}{150} & \multirow{2}{*}{2501} & \multirow{2}{*}{48.3} & \cellcolor{white}Random & \cellcolor{white}107 ms & \cellcolor{white}35.7 ms & \cellcolor{white}3.04 s & \cellcolor{LCMCColorLight}82 & \cellcolor{LCMCColorLight}0.257 & \cellcolor{white}0.893 & \cellcolor{white}0.886 & \cellcolor{white}0.779 & \cellcolor{white}0.741 & \cellcolor{white}0.302 & \multirow{2}{*}{NetworkX} & \multirow{2}{6.45cm}{Two non-overlapping 50 node complete subgraphs connected by a 50 node path.}\\
 & & & & \cellcolor{ourMethodColor}Layered & \cellcolor{ourMethodColor}108 ms & \cellcolor{ourMethodColor}35.1 ms & \cellcolor{ourMethodColor}529 ms & \cellcolor{LCMCColorDark}\textbf{12} & \cellcolor{LCMCColorDark}\textbf{0.350} & \cellcolor{ourMethodColor}0.913 & \cellcolor{ourMethodColor}0.914 & \cellcolor{ourMethodColor}0.780 & \cellcolor{ourMethodColor}0.741 & \cellcolor{ourMethodColor}0.234 &  & \\
\arrayrulecolor{hlineColor}\hline
\multirow{2}{*}{\textsc{bcsstk}} & \multirow{2}{*}{110} & \multirow{2}{*}{254} & \multirow{2}{*}{13.4} & \cellcolor{white}Random & \cellcolor{white}32.6 ms & \cellcolor{white}12.0 ms & \cellcolor{white}1.41 s & \cellcolor{LCMCColorLight}115 & \cellcolor{LCMCColorLight}\textbf{0.582} & \cellcolor{white}0.972 & \cellcolor{white}0.897 & \cellcolor{white}0.984 & \cellcolor{white}0.685 & \cellcolor{white}0.318 & UF Sparse & \multirow{2}{6.45cm}{The stiffness matrix used in structural simulation.}\\
 & & & & \cellcolor{ourMethodColor}Radial & \cellcolor{ourMethodColor}16.8 ms & \cellcolor{ourMethodColor}11.9 ms & \cellcolor{ourMethodColor}1.27 s & \cellcolor{LCMCColorDark}\textbf{106} & \cellcolor{LCMCColorDark}0.575 & \cellcolor{ourMethodColor}0.973 & \cellcolor{ourMethodColor}0.900 & \cellcolor{ourMethodColor}0.990 & \cellcolor{ourMethodColor}0.705 & \cellcolor{ourMethodColor}0.386 & Matrix Collection & \\
\arrayrulecolor{hlineColor}\hline
\multirow{2}{*}{\textsc{bio-diseasome}} & \multirow{2}{*}{516} & \multirow{2}{*}{1188} & \multirow{2}{*}{11.6} & \cellcolor{white}Random & \cellcolor{white}76.1 ms & \cellcolor{white}42.4 ms & \cellcolor{white}3.64 s & \cellcolor{LCMCColorLight}84 & \cellcolor{LCMCColorLight}0.407 & \cellcolor{white}0.909 & \cellcolor{white}0.845 & \cellcolor{white}0.991 & \cellcolor{white}0.738 & \cellcolor{white}0.440 & Network & \multirow{2}{6.45cm}{Graph of links for scientifically known disorder-gene association \cite{goh2007human}.}\\
 & & & & \cellcolor{ourMethodColor}Radial & \cellcolor{ourMethodColor}82.7 ms & \cellcolor{ourMethodColor}42.8 ms & \cellcolor{ourMethodColor}1.24 s & \cellcolor{LCMCColorDark}\textbf{27} & \cellcolor{LCMCColorDark}\textbf{0.478} & \cellcolor{ourMethodColor}0.964 & \cellcolor{ourMethodColor}0.957 & \cellcolor{ourMethodColor}0.992 & \cellcolor{ourMethodColor}0.744 & \cellcolor{ourMethodColor}0.436 & Repository & \\
\arrayrulecolor{hlineColor}\hline
\multirow{2}{*}{\textsc{circular ladder graph (100)}} & \multirow{2}{*}{200} & \multirow{2}{*}{300} & \multirow{2}{*}{51.0} & \cellcolor{white}Random & \cellcolor{white}33.7 ms & \cellcolor{white}16.6 ms & \cellcolor{white}1.77 s & \cellcolor{LCMCColorLight}105 & \cellcolor{LCMCColorLight}0.455 & \cellcolor{white}0.956 & \cellcolor{white}0.805 & \cellcolor{white}0.995 & \cellcolor{white}0.750 & \cellcolor{white}0.426 & \multirow{2}{*}{NetworkX} & \multirow{2}{6.45cm}{Pairs of nodes are connected in a ladder like pattern and the ladder forms a large cycle.}\\
 & & & & \cellcolor{ourMethodColor}Radial & \cellcolor{ourMethodColor}20.2 ms & \cellcolor{ourMethodColor}16.4 ms & \cellcolor{ourMethodColor}1.11 s & \cellcolor{LCMCColorDark}\textbf{66} & \cellcolor{LCMCColorDark}\textbf{0.814} & \cellcolor{ourMethodColor}0.999 & \cellcolor{ourMethodColor}0.998 & \cellcolor{ourMethodColor}1.000 & \cellcolor{ourMethodColor}0.427 & \cellcolor{ourMethodColor}0.557 &  & \\
\arrayrulecolor{hlineColor}\hline
\multirow{2}{*}{\textsc{connected cavemen (10,20)}} & \multirow{2}{*}{200} & \multirow{2}{*}{1900} & \multirow{2}{*}{11.0} & \cellcolor{white}Random & \cellcolor{white}85.3 ms & \cellcolor{white}31.9 ms & \cellcolor{white}1.33 s & \cellcolor{LCMCColorLight}39 & \cellcolor{LCMCColorLight}0.493 & \cellcolor{white}0.985 & \cellcolor{white}0.986 & \cellcolor{white}0.961 & \cellcolor{white}0.766 & \cellcolor{white}0.050 & \multirow{2}{*}{NetworkX} & \multirow{2}{6.45cm}{A graph of 10 cliques of 20 nodes each.}\\
 & & & & \cellcolor{ourMethodColor}Radial & \cellcolor{ourMethodColor}93.6 ms & \cellcolor{ourMethodColor}32.9 ms & \cellcolor{ourMethodColor}915 ms & \cellcolor{LCMCColorDark}\textbf{25} & \cellcolor{LCMCColorDark}0.496 & \cellcolor{ourMethodColor}0.986 & \cellcolor{ourMethodColor}0.985 & \cellcolor{ourMethodColor}0.961 & \cellcolor{ourMethodColor}0.768 & \cellcolor{ourMethodColor}0.045 &  & \\
\arrayrulecolor{hlineColor}\hline
\multirow{2}{*}{\textsc{engymes-g123}} & \multirow{2}{*}{90} & \multirow{2}{*}{127} & \multirow{2}{*}{10.3} & \cellcolor{white}Random & \cellcolor{white}13.9 ms & \cellcolor{white}13.9 ms & \cellcolor{white}1.7 s & \cellcolor{LCMCColorLight}121 & \cellcolor{LCMCColorLight}0.374 & \cellcolor{white}0.891 & \cellcolor{white}0.808 & \cellcolor{white}0.979 & \cellcolor{white}0.788 & \cellcolor{white}0.611 & Network & \multirow{2}{6.45cm}{Graph of cheminformatics.}\\
 & & & & \cellcolor{ourMethodColor}Layered & \cellcolor{ourMethodColor}15.5 ms & \cellcolor{ourMethodColor}14.3 ms & \cellcolor{ourMethodColor}1.2 s & \cellcolor{LCMCColorDark}\textbf{83} & \cellcolor{LCMCColorDark}\textbf{0.436} & \cellcolor{ourMethodColor}0.937 & \cellcolor{ourMethodColor}0.792 & \cellcolor{ourMethodColor}0.986 & \cellcolor{ourMethodColor}0.733 & \cellcolor{ourMethodColor}0.592 & Repository & \\
\arrayrulecolor{hlineColor}\hline
\multirow{2}{*}{\textsc{ladder}} & \multirow{2}{*}{20} & \multirow{2}{*}{28} & \multirow{2}{*}{8.0} & \cellcolor{white}Random & \cellcolor{white}3.6 ms & \cellcolor{white}15.3 ms & \cellcolor{white}1.71 s & \cellcolor{LCMCColorLight}111 & \cellcolor{LCMCColorLight}0.377 & \cellcolor{white}0.932 & \cellcolor{white}0.929 & \cellcolor{white}0.975 & \cellcolor{white}0.834 & \cellcolor{white}0.507 & \multirow{2}{*}{NetworkX} & \multirow{2}{6.45cm}{Pairs of nodes connected in a ladder like pattern.}\\
 & & & & \cellcolor{ourMethodColor}Radial & \cellcolor{ourMethodColor}2.5 ms & \cellcolor{ourMethodColor}15.1 ms & \cellcolor{ourMethodColor}1.21 s & \cellcolor{LCMCColorDark}\textbf{80} & \cellcolor{LCMCColorDark}\textbf{0.402} & \cellcolor{ourMethodColor}0.950 & \cellcolor{ourMethodColor}0.947 & \cellcolor{ourMethodColor}1.000 & \cellcolor{ourMethodColor}1.000 & \cellcolor{ourMethodColor}0.602 &  & \\
\arrayrulecolor{hlineColor}\hline
\multirow{2}{*}{\textsc{lobster}} & \multirow{2}{*}{300} & \multirow{2}{*}{299} & \multirow{2}{*}{77.5} & \cellcolor{white}Random & \cellcolor{white}21.2 ms & \cellcolor{white}20.9 ms & \cellcolor{white}3.79 s & \cellcolor{LCMCColorLight}180 & \cellcolor{LCMCColorLight}0.210 & \cellcolor{white}0.916 & \cellcolor{white}0.869 & \cellcolor{white}0.992 & \cellcolor{white}0.723 & \cellcolor{white}0.857 & \multirow{2}{*}{NetworkX} & \multirow{2}{6.45cm}{A tree that forms a caterpillar graph with the removal of a leaf \cite{golomb1996polyominoes}.}\\
 & & & & \cellcolor{ourMethodColor}Layered & \cellcolor{ourMethodColor}21.1 ms & \cellcolor{ourMethodColor}19.9 ms & \cellcolor{ourMethodColor}1.12 s & \cellcolor{LCMCColorDark}\textbf{55} & \cellcolor{LCMCColorDark}\textbf{0.688} & \cellcolor{ourMethodColor}0.996 & \cellcolor{ourMethodColor}0.997 & \cellcolor{ourMethodColor}1.000 & \cellcolor{ourMethodColor}1.000 & \cellcolor{ourMethodColor}0.807 &  & \\
\arrayrulecolor{hlineColor}\hline
\multirow{2}{*}{\textsc{lollipop (10,50)}} & \multirow{2}{*}{60} & \multirow{2}{*}{95} & \multirow{2}{*}{40.2} & \cellcolor{white}Random & \cellcolor{white}6.9 ms & \cellcolor{white}14.3 ms & \cellcolor{white}1.37 s & \cellcolor{LCMCColorLight}95 & \cellcolor{LCMCColorLight}0.444 & \cellcolor{white}0.894 & \cellcolor{white}0.833 & \cellcolor{white}0.909 & \cellcolor{white}0.803 & \cellcolor{white}0.770 & \multirow{2}{*}{NetworkX} & \multirow{2}{6.45cm}{The shape of a lollipop with a clique of 10 nodes connected to a thread of 50 nodes.}\\
 & & & & \cellcolor{ourMethodColor}Layered & \cellcolor{ourMethodColor}7.4 ms & \cellcolor{ourMethodColor}14.4 ms & \cellcolor{ourMethodColor}482 ms & \cellcolor{LCMCColorDark}\textbf{33} & \cellcolor{LCMCColorDark}\textbf{0.755} & \cellcolor{ourMethodColor}0.996 & \cellcolor{ourMethodColor}0.995 & \cellcolor{ourMethodColor}0.921 & \cellcolor{ourMethodColor}0.813 & \cellcolor{ourMethodColor}0.551 &  & \\
\arrayrulecolor{hlineColor}\hline
\multirow{2}{*}{\textsc{map of science}} & \multirow{2}{*}{554} & \multirow{2}{*}{2276} & \multirow{2}{*}{12.6} & \cellcolor{white}Random & \cellcolor{white}129 ms & \cellcolor{white}56.4 ms & \cellcolor{white}5.77 s & \cellcolor{LCMCColorLight}100 & \cellcolor{LCMCColorLight}0.361 & \cellcolor{white}0.956 & \cellcolor{white}0.930 & \cellcolor{white}0.986 & \cellcolor{white}0.666 & \cellcolor{white}0.135 & \multirow{2}{*}{\cite{borner2012design}} & \multirow{2}{6.45cm}{Graph of science sub-disciplines and cross-disciplinary co-authorships.}\\
 & & & & \cellcolor{ourMethodColor}Radial & \cellcolor{ourMethodColor}153 ms & \cellcolor{ourMethodColor}58.6 ms & \cellcolor{ourMethodColor}4.14 s & \cellcolor{LCMCColorDark}\textbf{68} & \cellcolor{LCMCColorDark}\textbf{0.402} & \cellcolor{ourMethodColor}0.977 & \cellcolor{ourMethodColor}0.957 & \cellcolor{ourMethodColor}0.990 & \cellcolor{ourMethodColor}0.687 & \cellcolor{ourMethodColor}0.145 &  & \\
\arrayrulecolor{hlineColor}\hline
\multirow{2}{*}{\textsc{random geometric (400,0,1)}} & \multirow{2}{*}{400} & \multirow{2}{*}{2263} & \multirow{2}{*}{13.5} & \cellcolor{white}Random & \cellcolor{white}107 ms & \cellcolor{white}51.8 ms & \cellcolor{white}4.56 s & \cellcolor{LCMCColorLight}86 & \cellcolor{LCMCColorLight}0.546 & \cellcolor{white}0.983 & \cellcolor{white}0.925 & \cellcolor{white}0.990 & \cellcolor{white}0.665 & \cellcolor{white}0.071 & \multirow{2}{*}{NetworkX} & \multirow{2}{6.45cm}{Nodes randomly placed in a cube and connectted if their distance is less than 0.1.}\\
 & & & & \cellcolor{ourMethodColor}Radial & \cellcolor{ourMethodColor}113 ms & \cellcolor{ourMethodColor}53.2 ms & \cellcolor{ourMethodColor}2.83 s & \cellcolor{LCMCColorDark}\textbf{51} & \cellcolor{LCMCColorDark}\textbf{0.622} & \cellcolor{ourMethodColor}0.987 & \cellcolor{ourMethodColor}0.963 & \cellcolor{ourMethodColor}0.993 & \cellcolor{ourMethodColor}0.660 & \cellcolor{ourMethodColor}0.076 &  & \\
\arrayrulecolor{hlineColor}\hline
\multirow{2}{*}{\textsc{retweet}} & \multirow{2}{*}{96} & \multirow{2}{*}{117} & \multirow{2}{*}{7.3} & \cellcolor{white}Random & \cellcolor{white}9.7 ms & \cellcolor{white}14.0 ms & \cellcolor{white}1.44 s & \cellcolor{LCMCColorLight}102 & \cellcolor{LCMCColorLight}0.440 & \cellcolor{white}0.936 & \cellcolor{white}0.905 & \cellcolor{white}0.985 & \cellcolor{white}0.678 & \cellcolor{white}0.781 & Network & \multirow{2}{6.45cm}{Network of twitter users as nodes and retweets as edges.}\\
 & & & & \cellcolor{ourMethodColor}Radial & \cellcolor{ourMethodColor}9.8 ms & \cellcolor{ourMethodColor}14.0 ms & \cellcolor{ourMethodColor}977 ms & \cellcolor{LCMCColorDark}\textbf{69} & \cellcolor{LCMCColorDark}\textbf{0.466} & \cellcolor{ourMethodColor}0.935 & \cellcolor{ourMethodColor}0.900 & \cellcolor{ourMethodColor}0.989 & \cellcolor{ourMethodColor}0.736 & \cellcolor{ourMethodColor}0.778 & Repository & \\
\arrayrulecolor{hlineColor}\hline
\multirow{2}{*}{\textsc{science collaboration network}} & \multirow{2}{*}{379} & \multirow{2}{*}{914} & \multirow{2}{*}{12.1} & \cellcolor{white}Random & \cellcolor{white}61.6 ms & \cellcolor{white}31.1 ms & \cellcolor{white}2.52 s & \cellcolor{LCMCColorLight}79 & \cellcolor{LCMCColorLight}0.433 & \cellcolor{white}0.947 & \cellcolor{white}0.887 & \cellcolor{white}0.992 & \cellcolor{white}0.741 & \cellcolor{white}0.352 & \multirow{2}{*}{\cite{newman2001structure}} & \multirow{2}{6.45cm}{Nodes are network theory publishing scientists and edges are collaborations between them.}\\
 & & & & \cellcolor{ourMethodColor}Layered & \cellcolor{ourMethodColor}50.6 ms & \cellcolor{ourMethodColor}36.2 ms & \cellcolor{ourMethodColor}2.4 s & \cellcolor{LCMCColorDark}\textbf{65} & \cellcolor{LCMCColorDark}\textbf{0.500} & \cellcolor{ourMethodColor}0.972 & \cellcolor{ourMethodColor}0.960 & \cellcolor{ourMethodColor}0.994 & \cellcolor{ourMethodColor}0.749 & \cellcolor{ourMethodColor}0.356 &  & \\
\arrayrulecolor{hlineColor}\hline
\multirow{2}{*}{\textsc{watts strogatz (100,5,0,0.5)}} & \multirow{2}{*}{100} & \multirow{2}{*}{200} & \multirow{2}{*}{12.3} & \cellcolor{white}Random & \cellcolor{white}24.7 ms & \cellcolor{white}13.1 ms & \cellcolor{white}1.75 s & \cellcolor{LCMCColorLight}131 & \cellcolor{LCMCColorLight}0.479 & \cellcolor{white}0.927 & \cellcolor{white}0.882 & \cellcolor{white}0.989 & \cellcolor{white}0.687 & \cellcolor{white}0.199 & \multirow{2}{*}{NetworkX} & \multirow{2}{6.45cm}{A small world graph that satisifies the Watts-Strogatz model \cite{watts1998collective}.}\\
 & & & & \cellcolor{ourMethodColor}Layered & \cellcolor{ourMethodColor}13.9 ms & \cellcolor{ourMethodColor}13.8 ms & \cellcolor{ourMethodColor}470 ms & \cellcolor{LCMCColorDark}\textbf{33} & \cellcolor{LCMCColorDark}\textbf{0.611} & \cellcolor{ourMethodColor}0.976 & \cellcolor{ourMethodColor}0.941 & \cellcolor{ourMethodColor}0.991 & \cellcolor{ourMethodColor}0.646 & \cellcolor{ourMethodColor}0.196 &  & \\
\arrayrulecolor{hlineColor}\hline

    \end{tabular}}
\end{table*}

\subsection{Implementation}

We have implemented our approach in JavaScript and D3.js v5 using the base implementation of D3.js force-directed layout with all standard settings. To initialize the layout, our code provides xy-coordinates to all nodes before the D3.js force-directed layout simulation takes control of the data. Modifications to the graph forces are done by adding new forces to the D3.js layout simulation. All experiments use the default D3.js stopping criteria for computation.
A demo version of our approach 
is at \demoURL, 
and our source code is available at \sourceURL.

\subsection{Datasets}
\label{sec.eval.data}

We have tested 32 datasets that include a mix of synthetic and real-world datasets, acquired from sources including the Network Repository~\cite{rossi2015network}, NetworkX~\cite{hagberg2008exploring}, BioSNAP~\cite{biosnapnets}, and the UF Sparse Matrix Collection~\cite{davis2011university}. The graphs are evenly divided into 16 dense and 16 sparse graphs, based upon their average node eccentricity (ECC)\footnote{Eccentricity is the maximum shortest path distance from a given node.}. A summary of graphs found in the paper can be seen in \autoref{tab:dense} and~\autoref{tab:sparse}. Further, as a practical matter, interactivity of graph visualizations begins to degrade at $\sim1000$ nodes in D3.js. Therefore, we differentiate larger graphs as those where $|N|>1000$. Graphs not in the paper can be found in a comprehensive table of results included in our supplemental materials and in our demo.

All graphs are colored using the D3.js plasma color map (0~\scalebox{-1}[1]{\includegraphics[width=40pt, height=5pt]{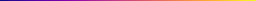}}~1) of their normalized node valence. The only exception is the \textsc{map of science} dataset (see \autoref{fig:sparse:map}), which is colored using a categorical color map.

\subsection{Evaluation Metrics}
\label{sec.eval.metrics}

Layout algorithms are often optimized considering aesthetic criteria. 
Purchase~\cite{purchase1997aesthetic} worked on various aesthetic criteria of importance and priority and showed that minimizing the number of edge crossings serves as critical aesthetic quality. Beck et al.~\cite{beck2009towards} defined several aesthetic criteria that ease designing, comparing, and evaluating different dynamic visualizations, including general aesthetic criteria, dynamic aesthetic criteria, and aesthetic scalability criteria. We use several criteria, including time ($T_*$), convergence ($C_*$), and layout quality ($Q_*$).

Our main goal is to measure whether global structures overlap with one another. To identify those overlaps, the primary measure we consider is that of co-ranking. Co-ranking compares the $k$-neighborhoods of a high-dimensional space, in our case defined by the unweighted shortest path distance in the graph, with a low-dimensional embedding, the Euclidean distance between nodes on the image. We use several meta-criteria on the co-ranking.

\begin{description}[noitemsep,itemsep=2pt,leftmargin=7pt,labelindent=10pt,itemindent=0pt]
\item[Local Continuity Meta Criterion] ($Q_{LCMC}$/$C_{LCMC}$) measures the ranked order overlap of $k$-neighborhoods for a range $[1,k]$ and averages them~\cite{chen2009local}. To ease comparisons, we fix $k=20$. $Q_{LCMC}$ is normalized such that $Q_{LCMC}\in[-1,1]$, where larger is better and negative values imply opposite ordering. We also utilize the convergence $C_{LCMC}$, which is the iteration number when $Q_{LCMC}$ is within 0.01 of the final value (after 300 iterations of force calculations). 
\item[Trustworthiness] ($Q_{trust}$) and \textbf{Continuity} ($Q_{cont}$) co-ranking meta criteria~\cite{venna2006local}, whose conclusions parallel LCMC, are also provided.   
\end{description}

We next consider three measures of performance that quantify the processing time.
\begin{description}[noitemsep,itemsep=2pt,leftmargin=7pt,labelindent=10pt,itemindent=0pt]
    \item[Initialization Time] ($T_{IT}$) is the time taken to initialize the force-directed layout simulation. For our approach, the timing includes the overhead to calculate the spanning tree and position nodes.
    \item[Average Iteration Time] ($T_{AIT}$)  is the average time required to calculate one iteration of the force-directed layout.
    \item[Total Time to LCMC Convergence] ($T_{LCMC}$) is the total time ($T_{IT}+T_{AIT}*C_{LCMC}$) required to reach the LCMC convergence criteria.    
\end{description}

Finally, we produce a set of well-established graph readability metrics~\cite{dunne2015readability,gove2018pays}. Our evaluation does not discuss them, but they are included for completeness.
\begin{description}[noitemsep,itemsep=2pt,leftmargin=7pt,labelindent=10pt,itemindent=0pt]
    \item[Edge Crossings] ($Q_{EC}$) measures the ratio of non-intersecting edges to total possible intersections. Graphs are generally considered more readable with fewer crossings. $Q_{EC}$ is normalized, such that $Q_{EC}\in[0,1]$, where larger is better.
    \item[Crossing Angle] ($Q_{CA}$) is the average deviation from the ideal crossing angle. If edges cross, it is preferable they cross at an ideal crossing angle of 70 degrees that makes their individual paths most visible. $Q_{CA}$ is normalized, such that $Q_{CA}\in[0,1]$, where larger is better.
    \item[Minimum Angular Resolution]  ($Q_{MAR}$) measures the average deviation of adjacent edge angles from the ideal angle ($360^{\circ}/degree(v_i)$ for any $v_i \in V$). For nodes with multiple edges, it is preferable to their egress be distributed around the node as much as possible. $Q_{MAR}$ is normalized, such that $Q_{MAR}\in[0,1]$, where larger is better.
\end{description}

\section{Results}

We evaluate our method's ability to untangle initial graph layouts, followed by untangling cycle structures.

\subsection{Untangling Initial Graph Layouts}
\label{sec:results:initial}

We evaluate our initial graph layout approach in terms of graph quality, convergence, and time. For the experiments, we initialized the graphs with either the standard D3.js random layout or our approach and let them run until D3.js stopped force calculations (after 300 iterations using the default settings).

\subsubsection{Layout Quality}

\autoref{tab:dense} and~\ref{tab:sparse} show the results for all quality metrics from \autoref{sec.eval.metrics}, except for the readability measures for the three largest graphs,  which were skipped due to very high computational costs. Although all metrics are available, we discuss only $Q_{LCMC}$.

\para{Layered vs.\ Radial} The results in \autoref{tab:dense} and~\ref{tab:sparse}  show only the layout method, layered or radial, which produced higher $Q_{LCMC}$. In many cases, the result between both methods is effectively identical. The results show that neither method is universally better and seemed to be graph dependent. Nevertheless, the results for both layout methods are available in the supplemental material.

\begin{figure*}[!ht]

\subfloat[\textsc{bio-celegans}\label{fig:dense:bio}]{
\begin{minipage}[t]{0.315\linewidth}
\rotatebox{90}{\tiny \hspace{13pt} Random}
\hspace{-8pt}
\begin{tikzonimage}[width=0.325\linewidth]{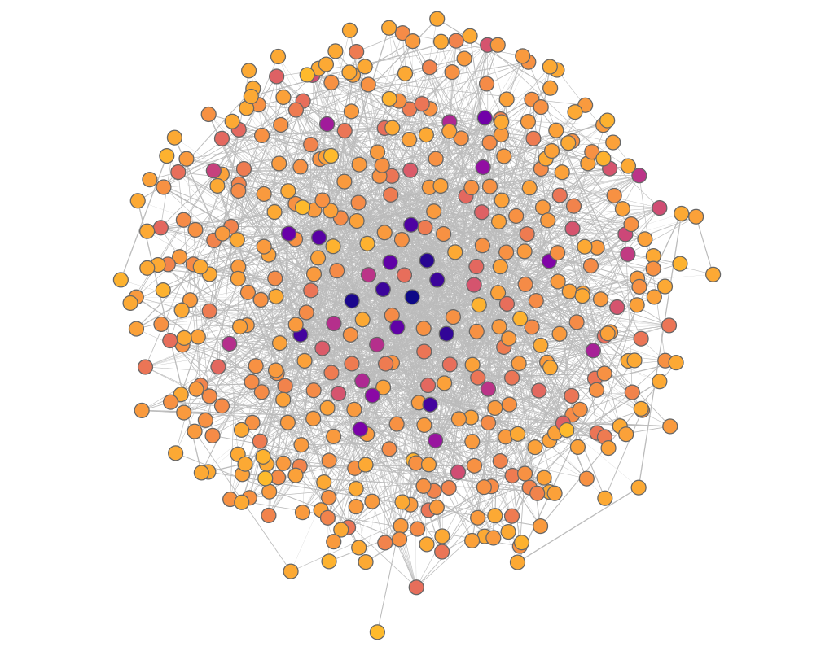}[image label]\node{$Q_{LCMC}$: 0.002};\end{tikzonimage}
\hspace{-2pt}
\begin{tikzonimage}[width=0.325\linewidth]{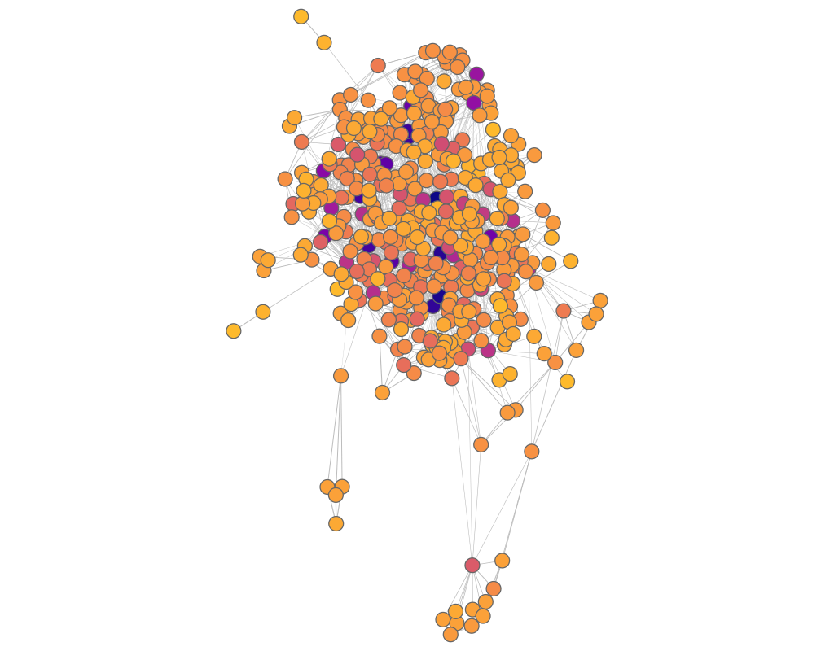}[image label]\node{$Q_{LCMC}$: 0.219};\end{tikzonimage}
\hspace{-1pt}
\rule[-45pt]{.5pt}{85pt}
\rotatebox{90}{\tiny \hspace{14pt} neato}
\hspace{-9pt}
\begin{tikzonimage}[width=0.28\linewidth]{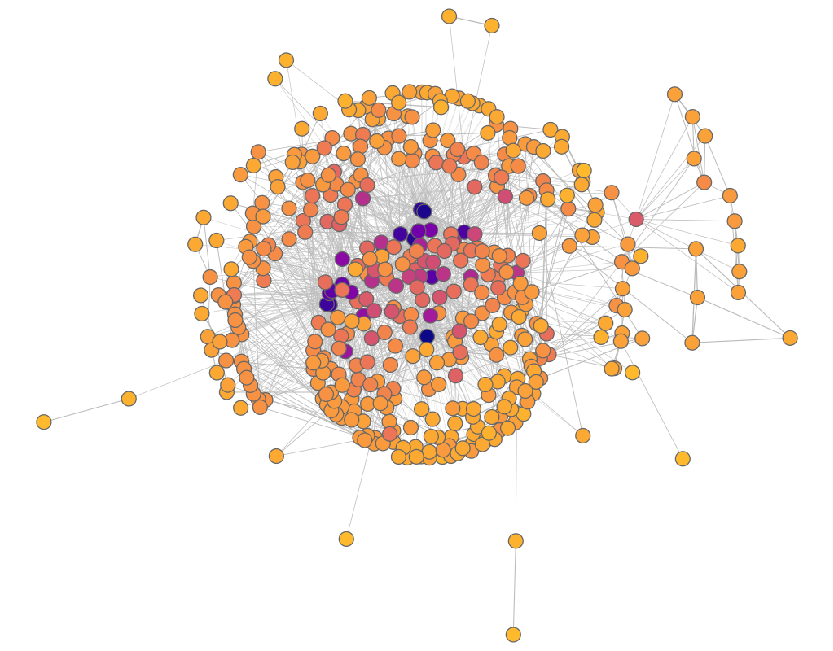}[image label]\node{$Q_{LCMC}$: 0.242};\end{tikzonimage}

\vspace{-45pt}
\rotatebox{90}{\tiny \hspace{13pt} Layered}
\hspace{-8pt}
\begin{tikzonimage}[width=0.325\linewidth]{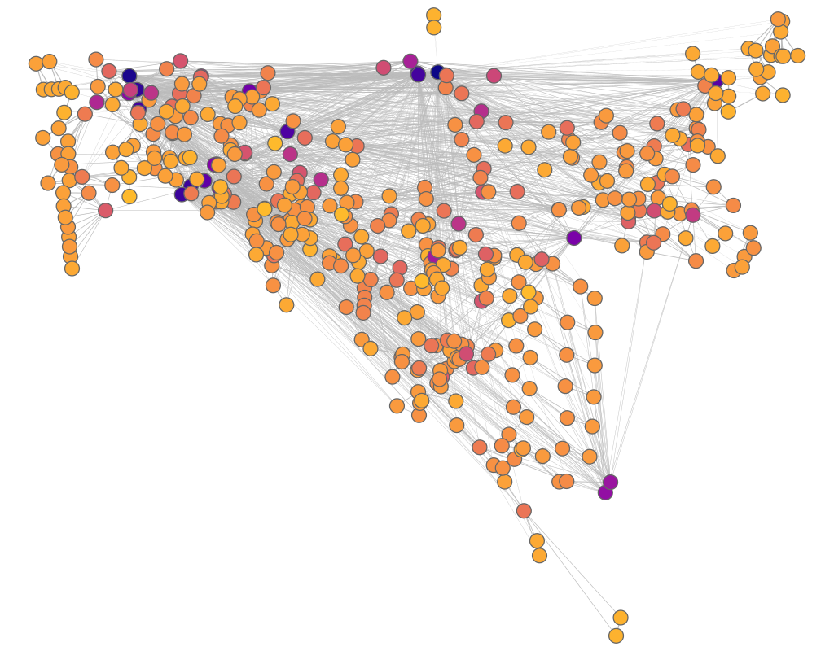}[image label]\node{$Q_{LCMC}$: 0.127};\end{tikzonimage}
\hspace{-2pt}
\begin{tikzonimage}[width=0.325\linewidth]{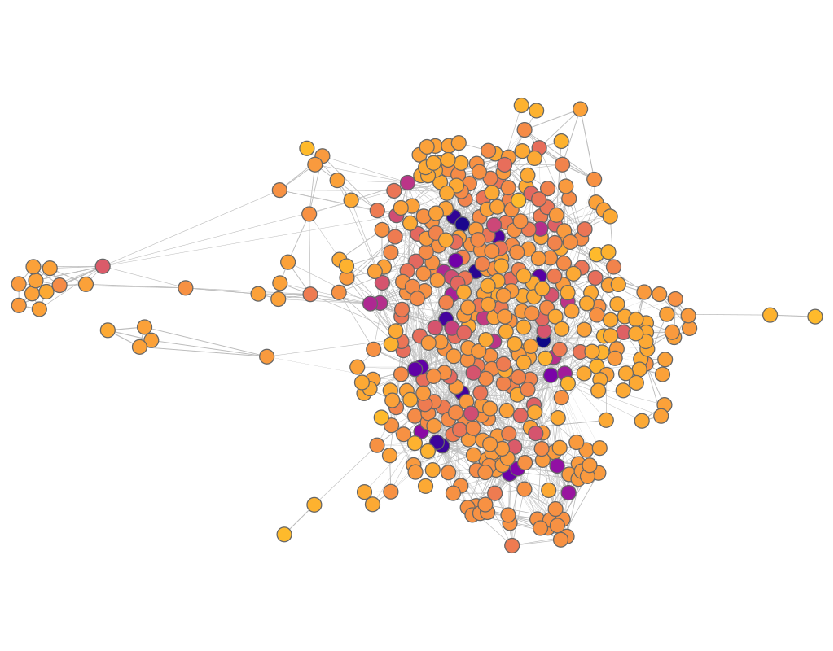}[image label]\node{$Q_{LCMC}$: 0.229};\end{tikzonimage}
\hspace{2.5pt}\rotatebox{90}{\tiny \hspace{15pt} sfdp}
\hspace{-9pt}
\begin{tikzonimage}[width=0.28\linewidth]{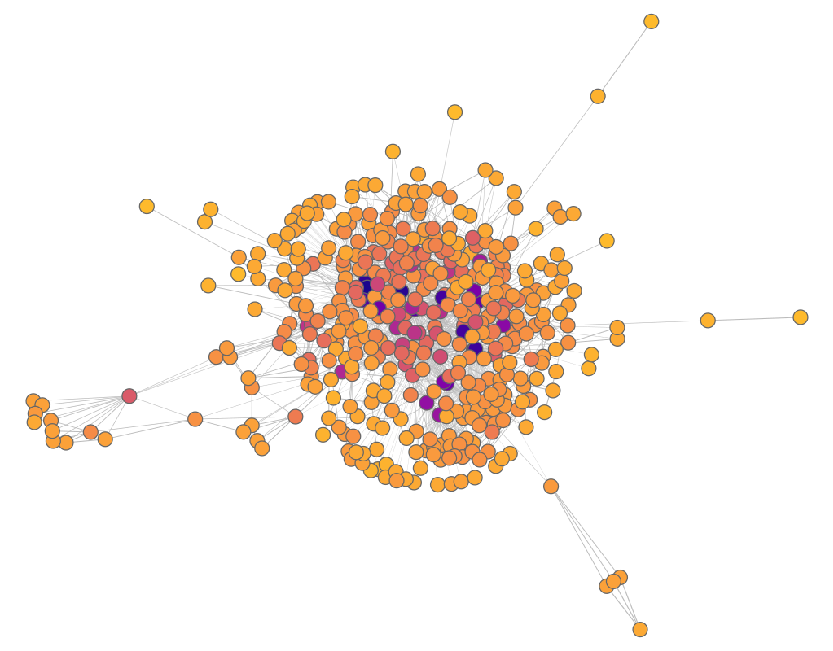}[image label]\node{$Q_{LCMC}$: 0.225};\end{tikzonimage}

{\footnotesize \hspace{23pt} Initial \hspace{33pt} Final}
\end{minipage}
}
\hfill
\subfloat[\textsc{bn-mouse-visual-cortex-2}\label{fig:dense:mouse}]{
\begin{minipage}[t]{0.315\linewidth}
\rotatebox{90}{\tiny \hspace{13pt} Random} \hspace{-8pt}
\begin{tikzonimage}[width=0.325\linewidth]{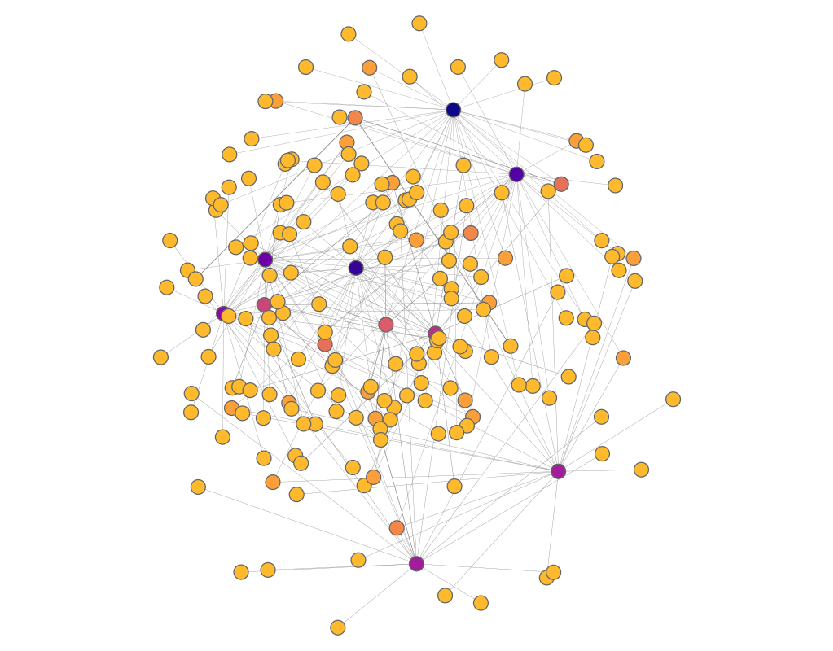}[image label]\node{$Q_{LCMC}$: 0.035};\end{tikzonimage}
\begin{tikzonimage}[width=0.325\linewidth]{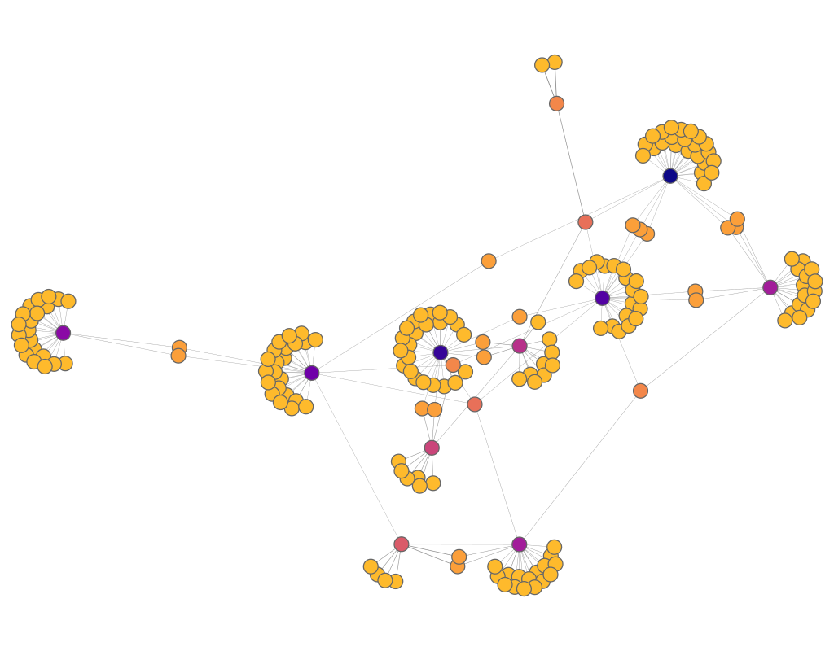}[image label]\node{$Q_{LCMC}$: 0.409};\end{tikzonimage}
\hspace{-1pt} \rule[-45pt]{.5pt}{85pt} \rotatebox{90}{\tiny \hspace{14pt} neato} \hspace{-9pt}
\begin{tikzonimage}[width=0.28\linewidth]{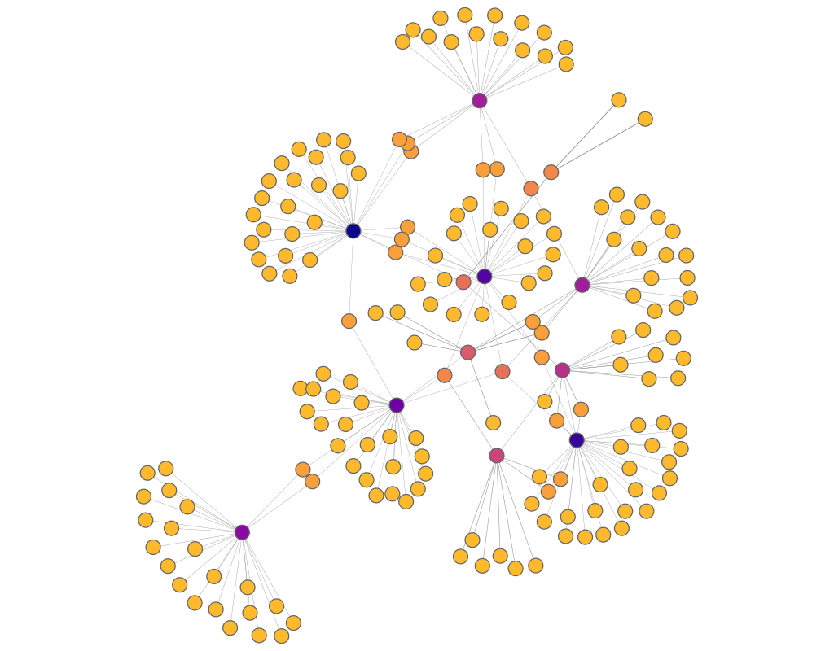}[image label]\node{$Q_{LCMC}$: 0.372};\end{tikzonimage}

\vspace{-45pt} \rotatebox{90}{\tiny \hspace{13pt} Layered} \hspace{-8pt}
\begin{tikzonimage}[width=0.325\linewidth]{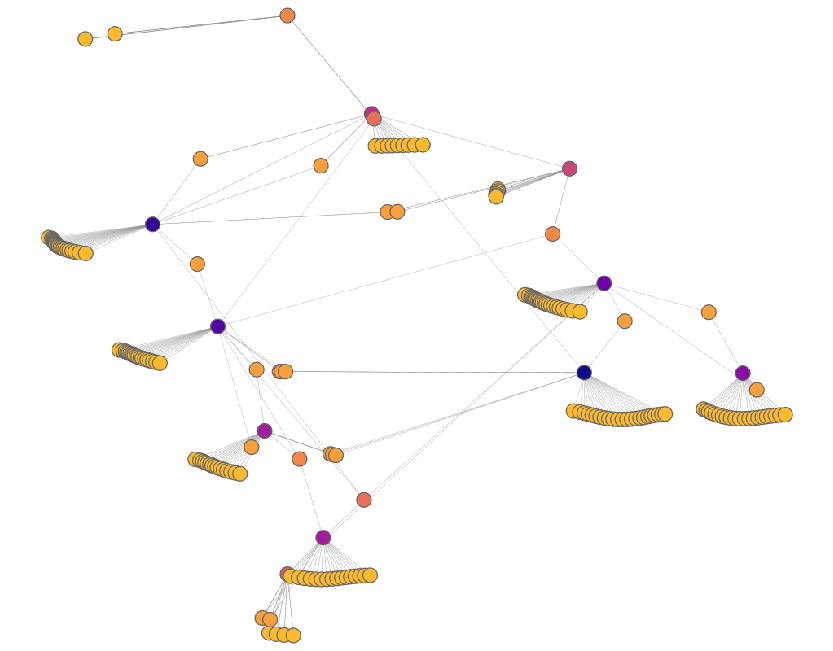}[image label]\node{$Q_{LCMC}$: 0.433};\end{tikzonimage}
\begin{tikzonimage}[width=0.325\linewidth]{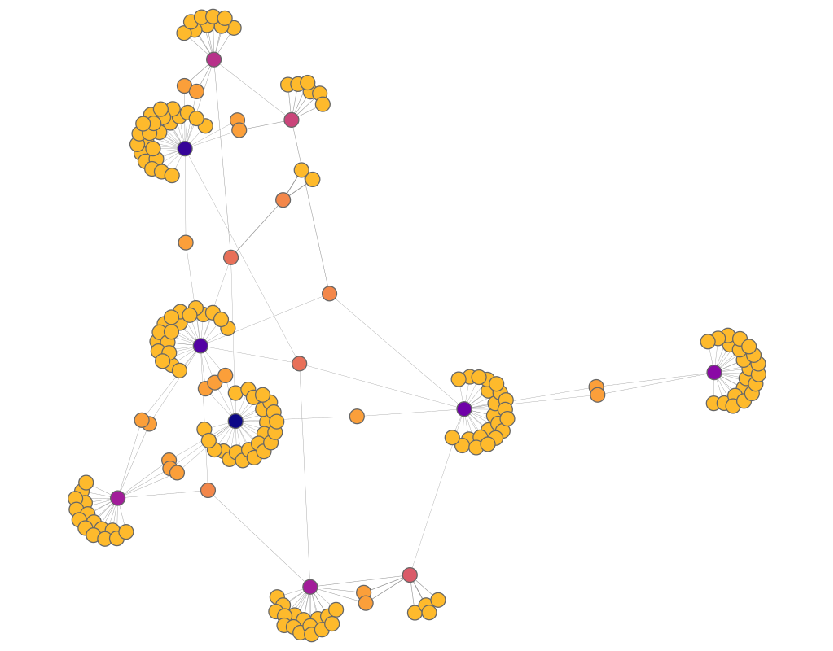}[image label]\node{$Q_{LCMC}$: 0.425};\end{tikzonimage}
\hspace{2.5pt}\rotatebox{90}{\tiny \hspace{14pt} sfdp} \hspace{-9pt}
\begin{tikzonimage}[width=0.28\linewidth]{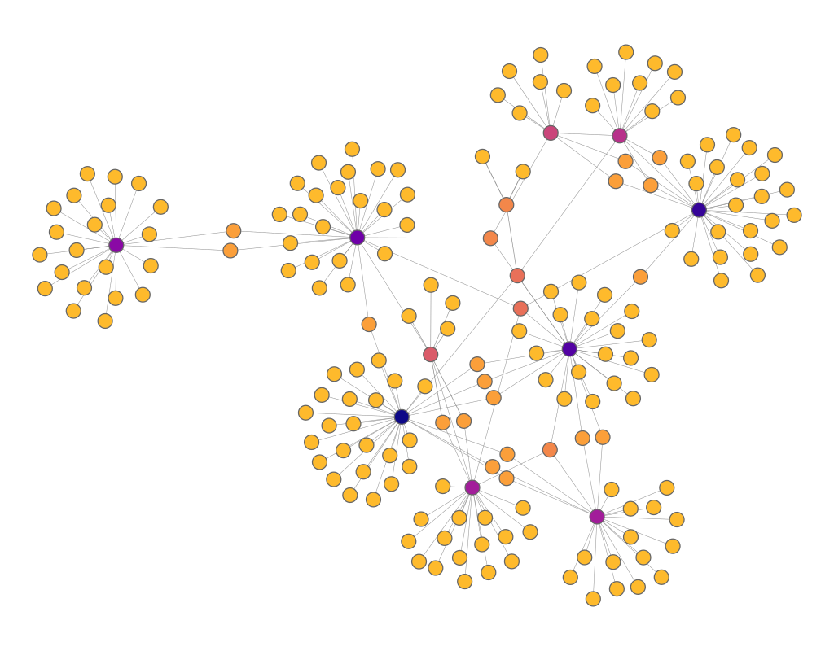}[image label]\node{$Q_{LCMC}$: 0.358};\end{tikzonimage}

{\footnotesize \hspace{23pt} Initial \hspace{33pt} Final}
\end{minipage}
}
\hfill
\subfloat[\textsc{dorogovtsev-goltsev-mendes}\label{fig:dense:dorogovtsev}]{
\begin{minipage}[t]{0.315\linewidth}
\rotatebox{90}{\tiny \hspace{13pt} Random} \hspace{-8pt}
\begin{tikzonimage}[width=0.325\linewidth]{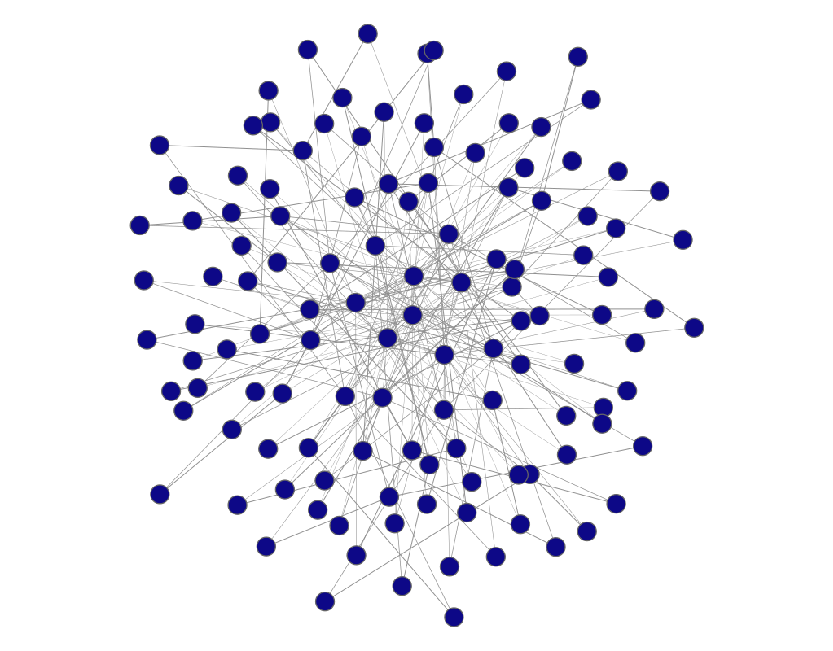}[image label]\node{$Q_{LCMC}$: 0.012};\end{tikzonimage}
\begin{tikzonimage}[width=0.325\linewidth]{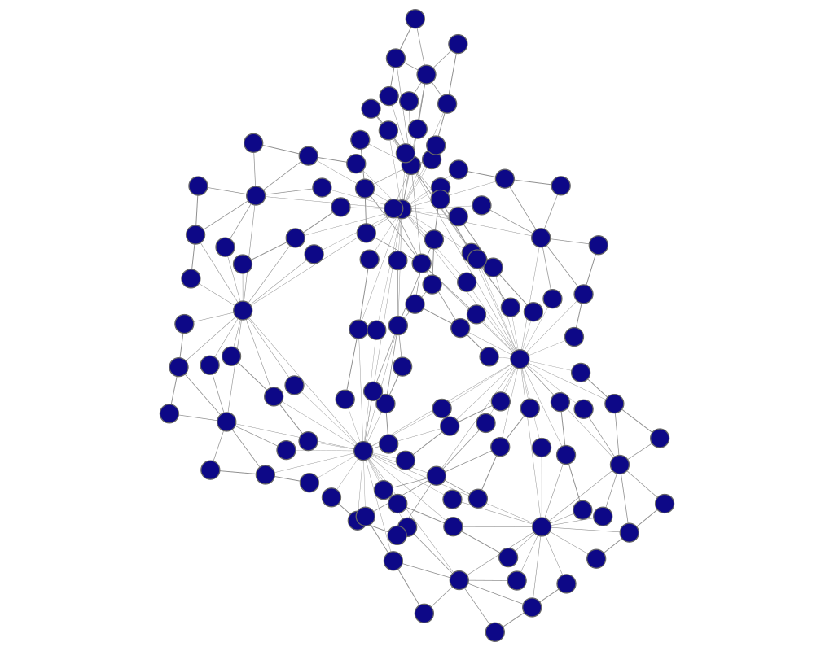}[image label]\node{$Q_{LCMC}$: 0.377};\end{tikzonimage}
\hspace{-1pt} \rule[-45pt]{.5pt}{85pt} \rotatebox{90}{\tiny \hspace{14pt} neato} \hspace{-9pt}
\begin{tikzonimage}[width=0.28\linewidth]{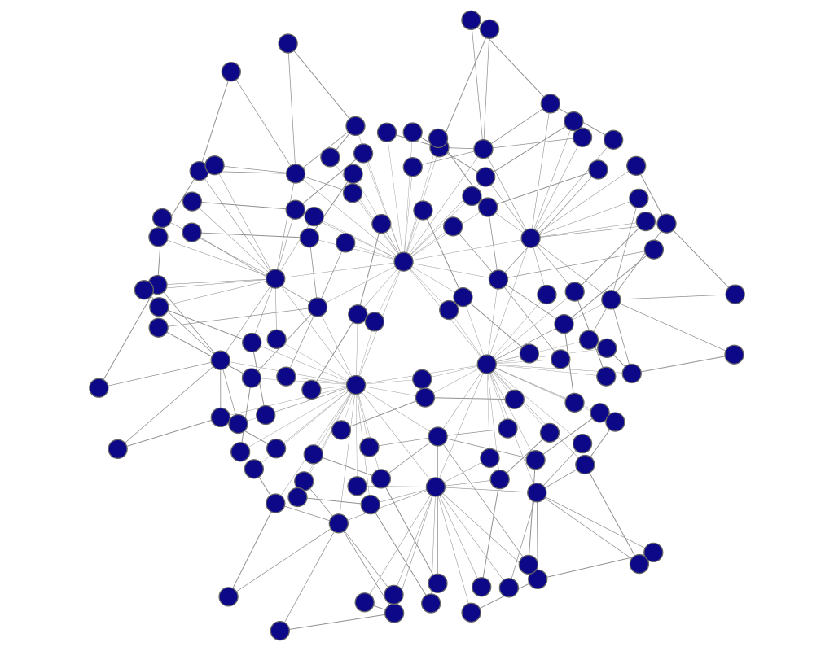}[image label]\node{$Q_{LCMC}$: 0.485};\end{tikzonimage}

\vspace{-45pt} \rotatebox{90}{\tiny \hspace{13pt} Layered} \hspace{-8pt}
\begin{tikzonimage}[width=0.325\linewidth]{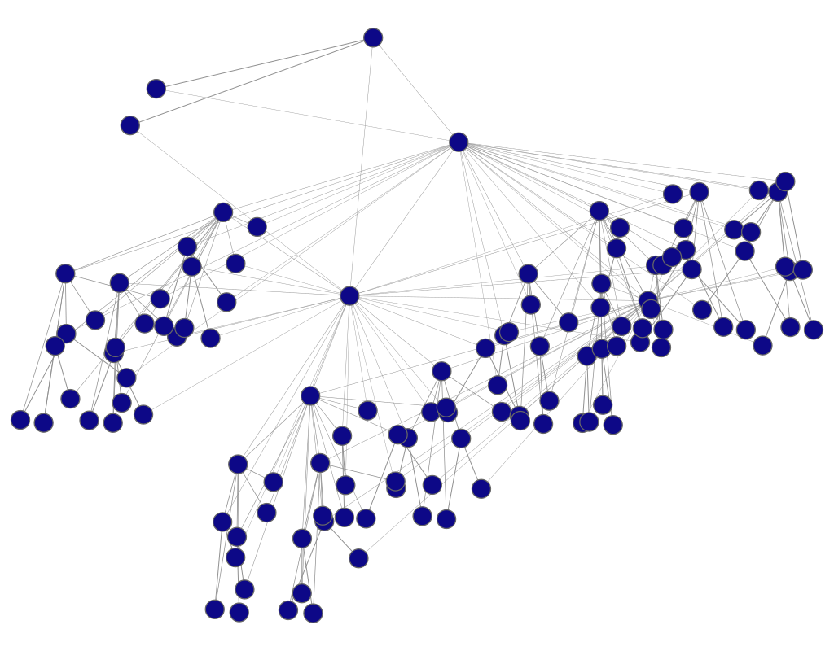}[image label]\node{$Q_{LCMC}$: 0.211};\end{tikzonimage}
\begin{tikzonimage}[width=0.325\linewidth]{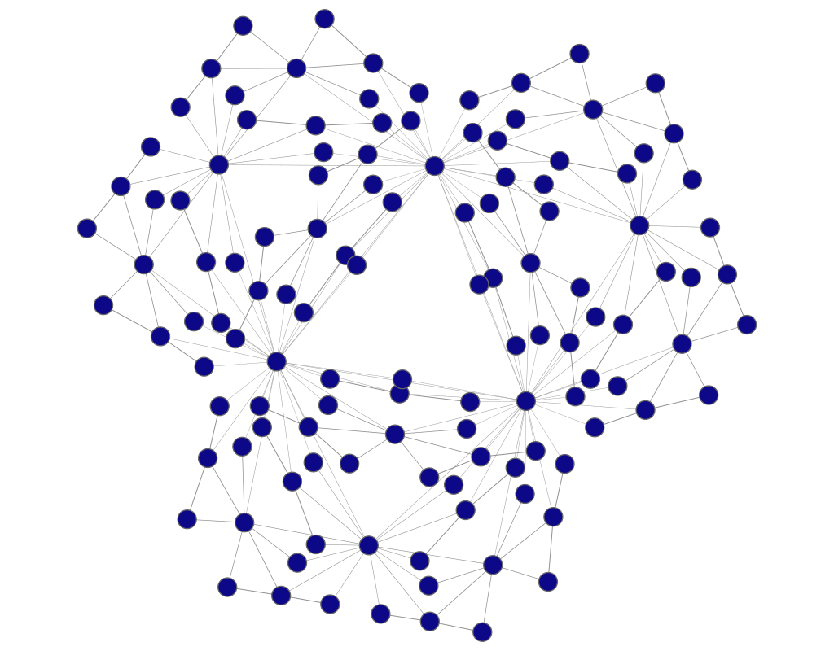}[image label]\node{$Q_{LCMC}$: 0.446};\end{tikzonimage}
\hspace{2.5pt}\rotatebox{90}{\tiny \hspace{14pt} sfdp} \hspace{-9pt}
\begin{tikzonimage}[width=0.28\linewidth]{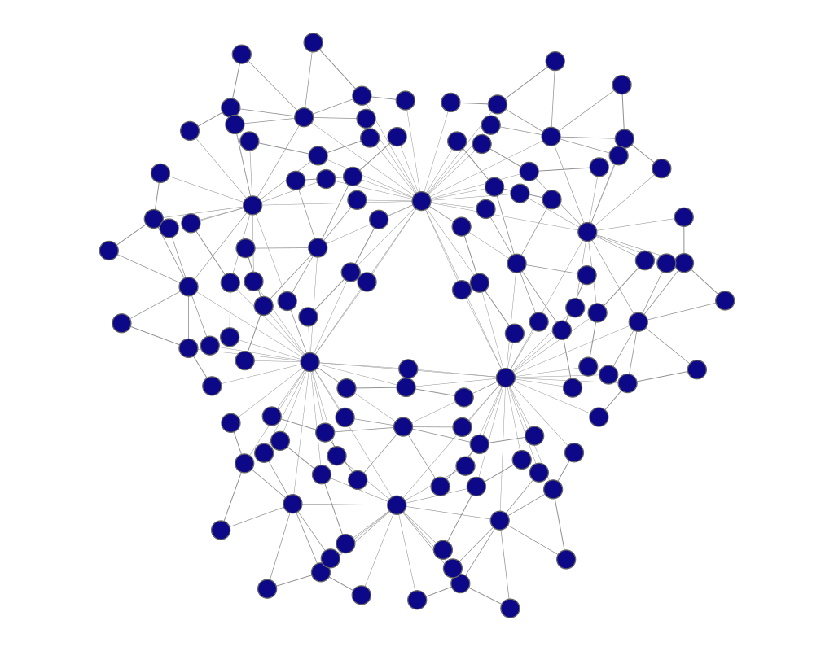}[image label]\node{$Q_{LCMC}$: 0.464};\end{tikzonimage}

{\footnotesize \hspace{23pt} Initial \hspace{33pt} Final}
\end{minipage}
}

\vspace{-3pt}
\subfloat[\textsc{duplicate divergence}\label{fig:dense:divergence}]{
\begin{minipage}[t]{0.315\linewidth}
\rotatebox{90}{\tiny \hspace{13pt} Random} \hspace{-8pt}
\begin{tikzonimage}[width=0.325\linewidth]{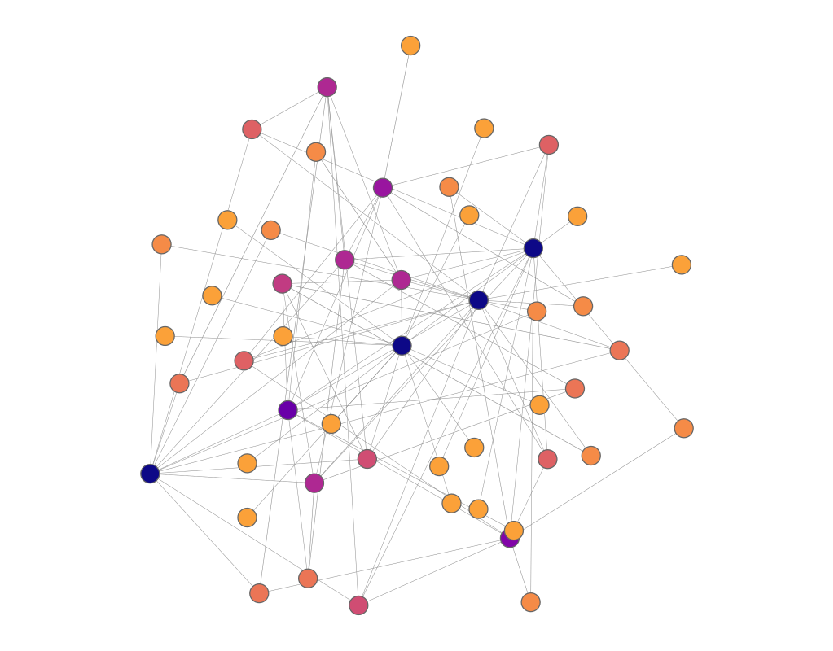}[image label]\node{$Q_{LCMC}$: 0.017};\end{tikzonimage}
\begin{tikzonimage}[width=0.325\linewidth]{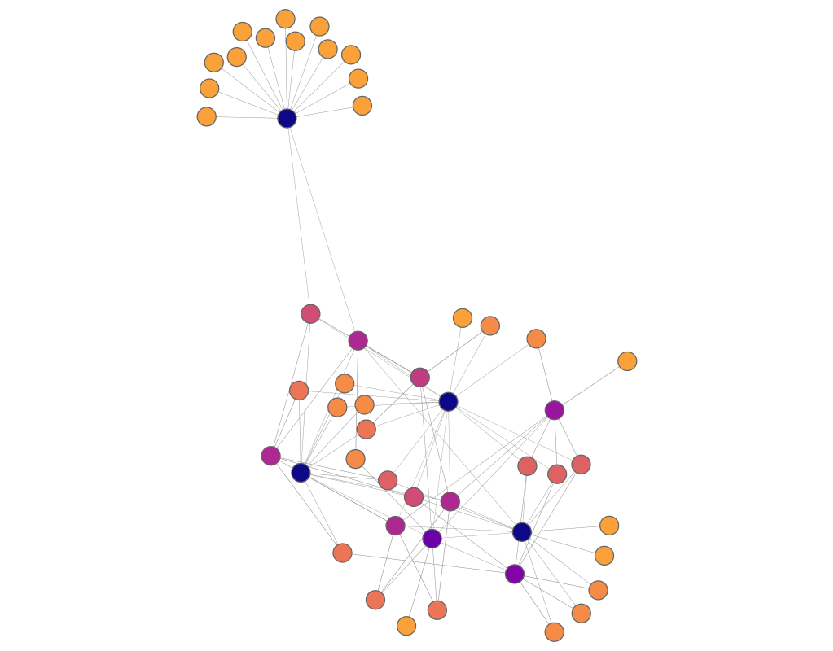}[image label]\node{$Q_{LCMC}$: 0.415};\end{tikzonimage}
\hspace{-1pt} \rule[-45pt]{.5pt}{85pt} \rotatebox{90}{\tiny \hspace{14pt} neato} \hspace{-9pt}
\begin{tikzonimage}[width=0.28\linewidth]{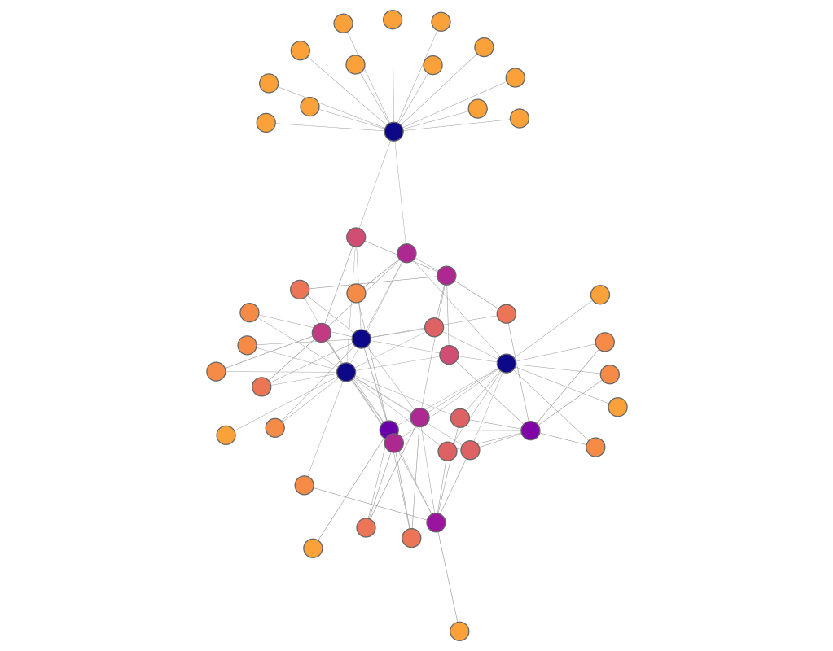}[image label]\node{$Q_{LCMC}$: 0.453};\end{tikzonimage}

\vspace{-45pt} \rotatebox{90}{\tiny \hspace{13pt} Layered} \hspace{-8pt}
\begin{tikzonimage}[width=0.325\linewidth]{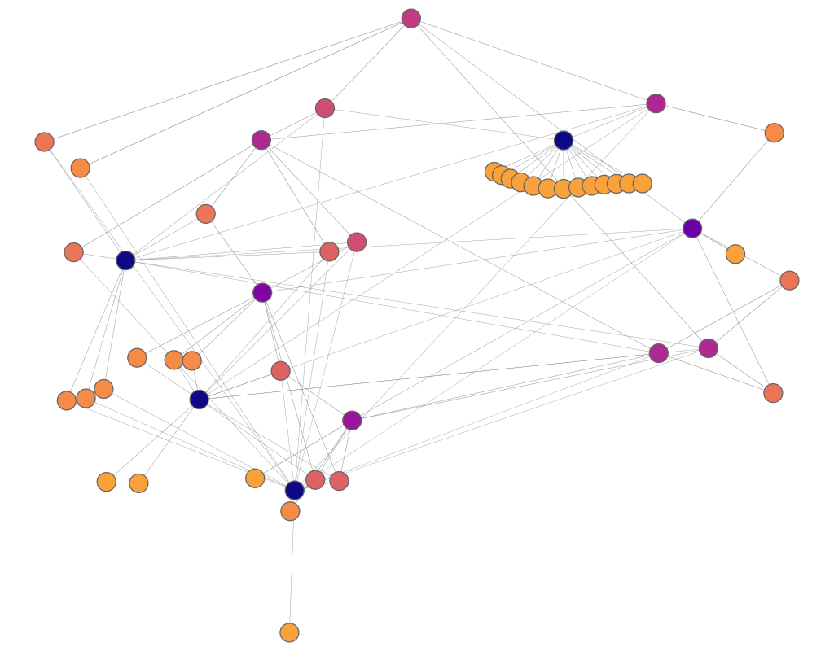}[image label]\node{$Q_{LCMC}$: 0.265};\end{tikzonimage}
\hspace{-2pt}
\begin{tikzonimage}[width=0.325\linewidth]{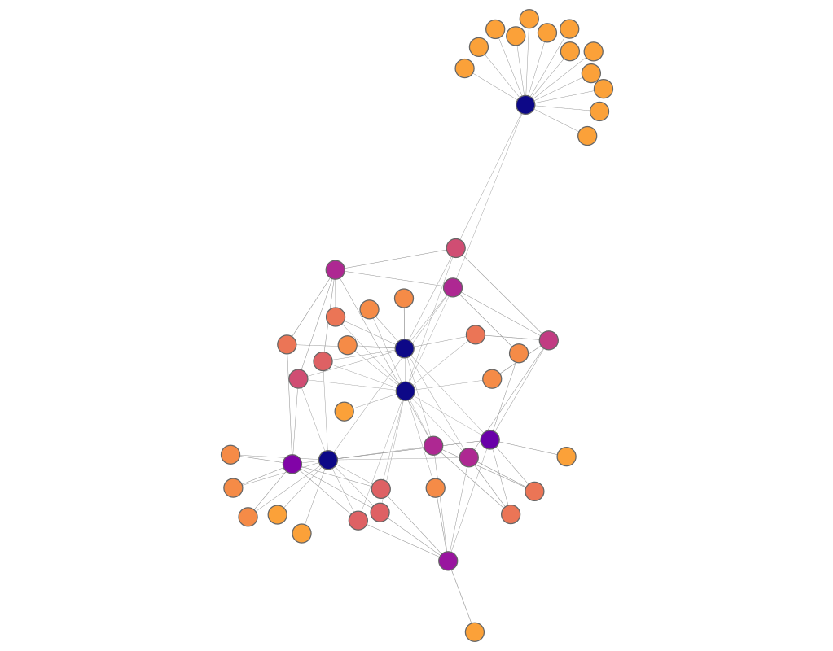}[image label]\node{$Q_{LCMC}$: 0.414};\end{tikzonimage}
\hspace{2.5pt}\rotatebox{90}{\tiny \hspace{14pt} sfdp} \hspace{-9pt}
\begin{tikzonimage}[width=0.28\linewidth]{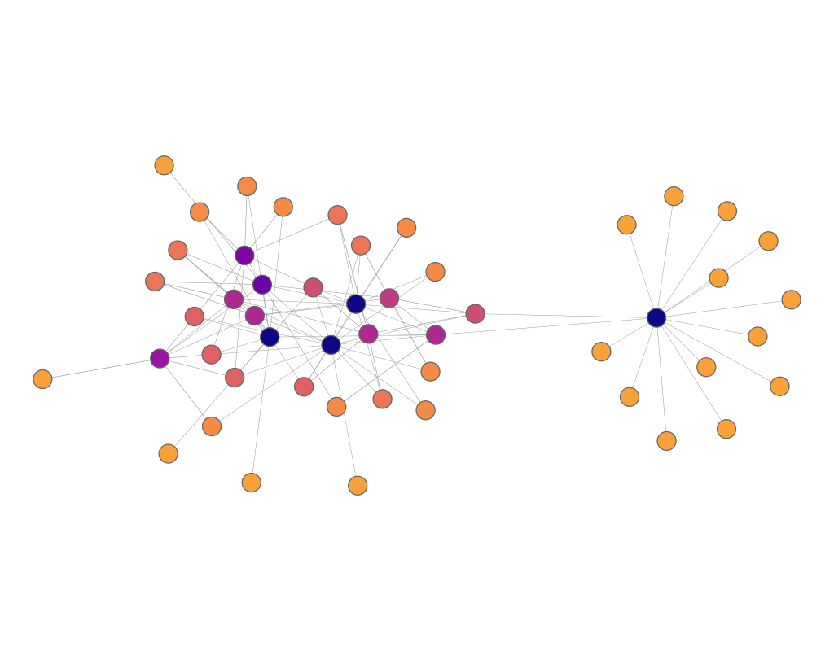}[image label]\node{$Q_{LCMC}$: 0.355};\end{tikzonimage}

{\footnotesize \hspace{23pt} Initial \hspace{33pt} Final}
\end{minipage}
}
\hfill
\subfloat[\textsc{enron email}\label{fig:dense:enron}]{
\begin{minipage}[t]{0.315\linewidth}
\rotatebox{90}{\tiny \hspace{13pt} Random} \hspace{-8pt}
\begin{tikzonimage}[width=0.325\linewidth]{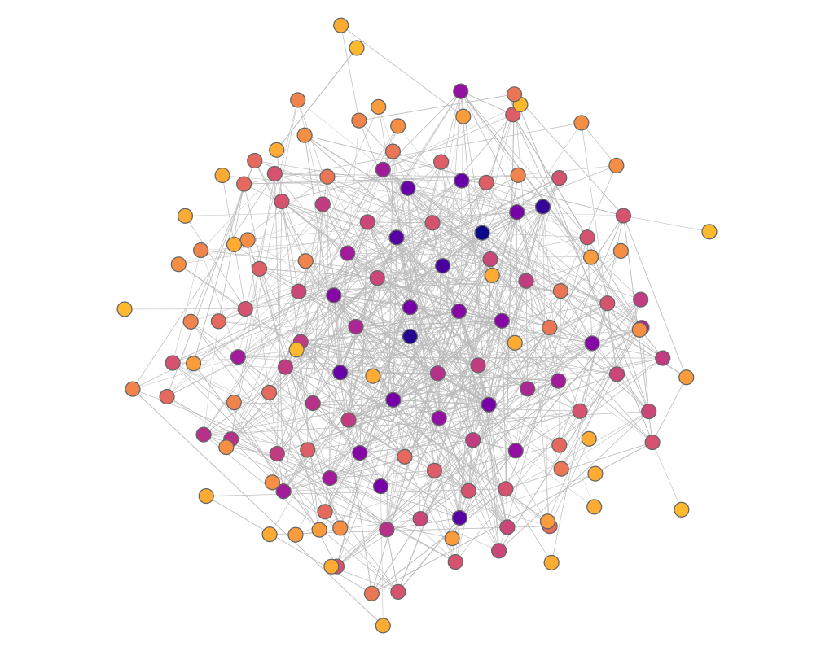}[image label]\node{$Q_{LCMC}$: 0.026};\end{tikzonimage}
\hspace{-2pt}
\begin{tikzonimage}[width=0.325\linewidth]{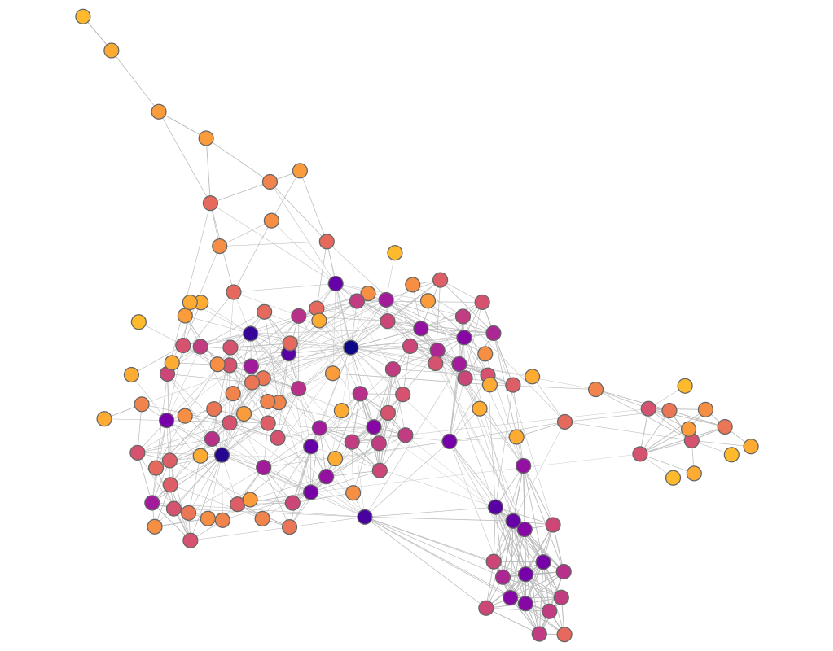}[image label]\node{$Q_{LCMC}$: 0.390};\end{tikzonimage}
\hspace{-1pt} \rule[-45pt]{.5pt}{85pt} \rotatebox{90}{\tiny \hspace{14pt} neato} \hspace{-9pt}
\begin{tikzonimage}[width=0.28\linewidth]{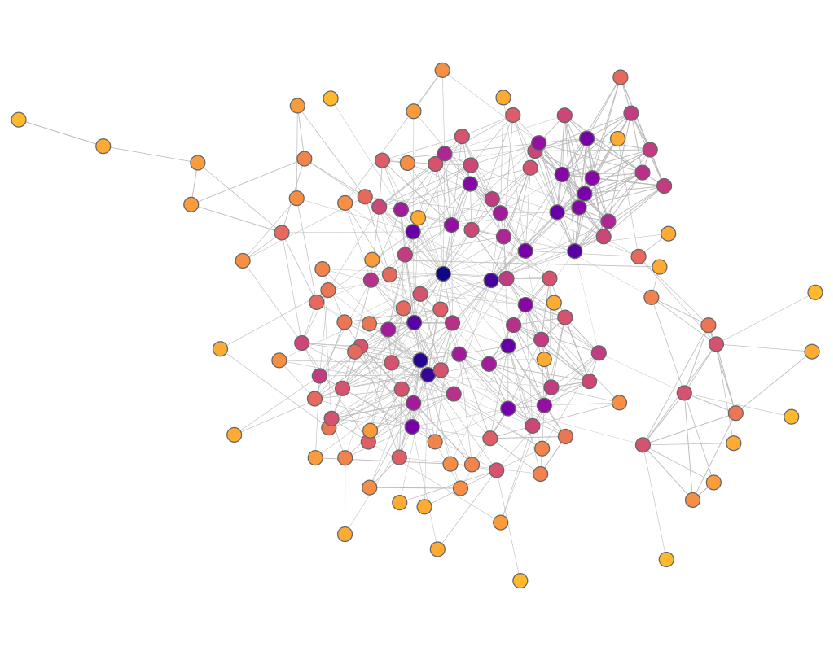}[image label]\node{$Q_{LCMC}$: 0.419};\end{tikzonimage}

\vspace{-45pt} \rotatebox{90}{\tiny \hspace{13pt} Layered} \hspace{-8pt}
\begin{tikzonimage}[width=0.325\linewidth]{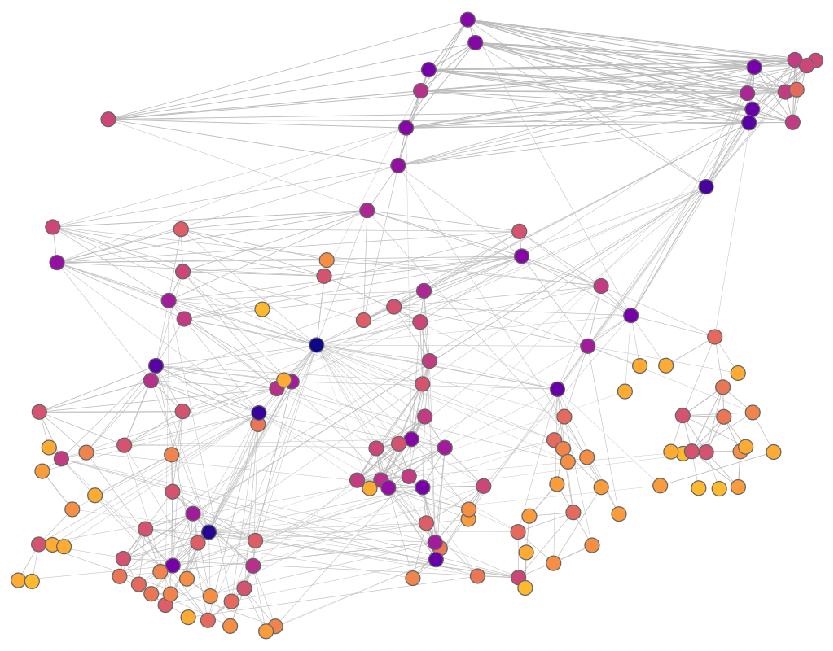}[image label]\node{$Q_{LCMC}$: 0.317};\end{tikzonimage}
\hspace{-2pt}
\begin{tikzonimage}[width=0.325\linewidth]{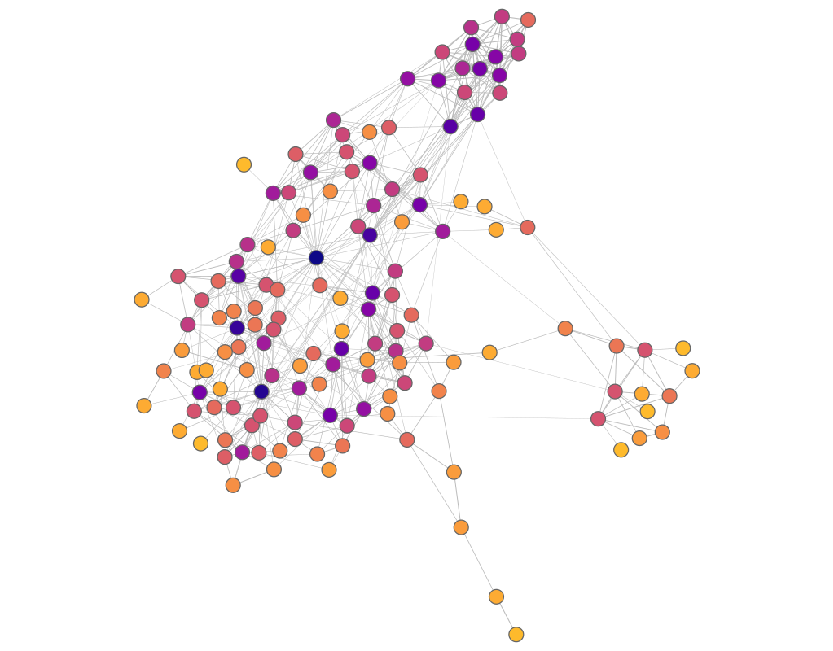}[image label]\node{$Q_{LCMC}$: 0.385};\end{tikzonimage}
\hspace{2.5pt}\rotatebox{90}{\tiny \hspace{14pt} sfdp} \hspace{-9pt}
\begin{tikzonimage}[width=0.28\linewidth]{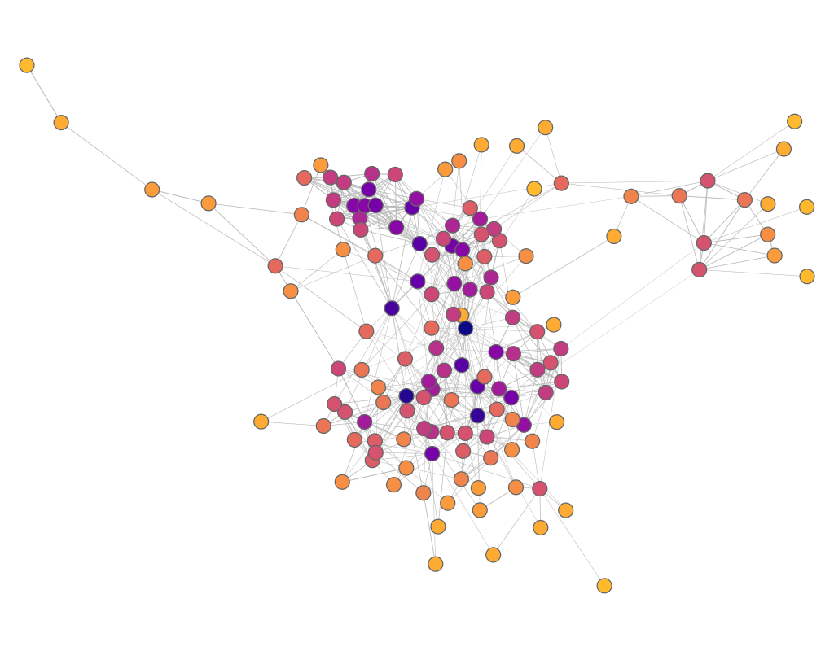}[image label]\node{$Q_{LCMC}$: 0.410};\end{tikzonimage}

{\footnotesize \hspace{23pt} Initial \hspace{33pt} Final}
\end{minipage}
}
\hfill
\subfloat[\textsc{usair 97}\label{fig:dense:usair}]{
\begin{minipage}[t]{0.315\linewidth}
\rotatebox{90}{\tiny \hspace{13pt} Random} \hspace{-8pt}
\begin{tikzonimage}[width=0.325\linewidth]{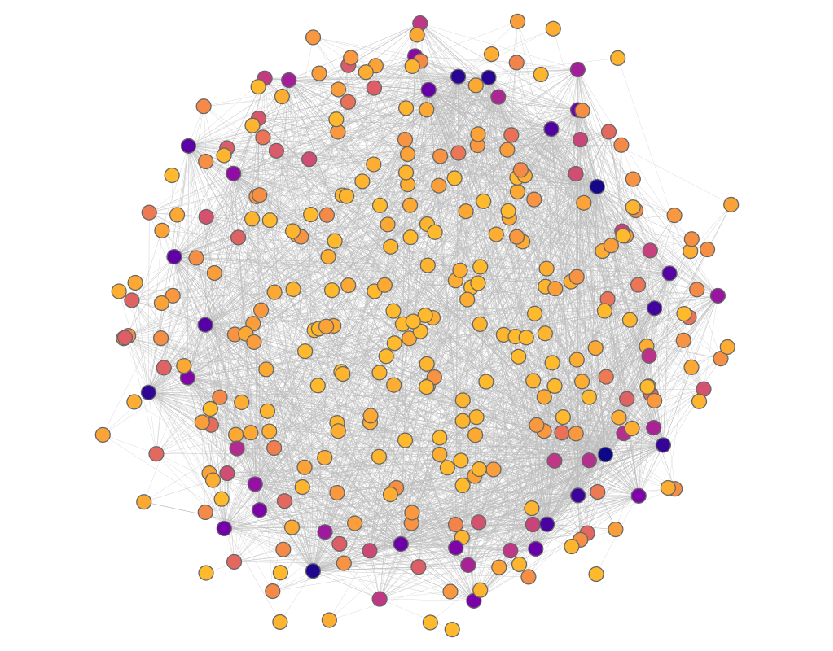}[image label]\node{$Q_{LCMC}$: 0.011};\end{tikzonimage}
\hspace{-2pt}
\begin{tikzonimage}[width=0.325\linewidth]{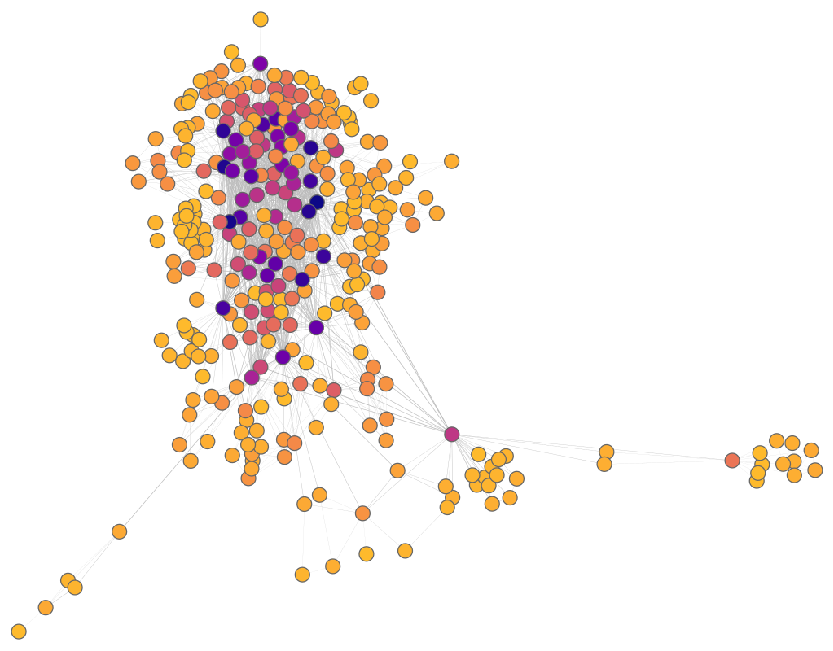}[image label]\node{$Q_{LCMC}$: 0.323};\end{tikzonimage}
\hspace{-1pt} \rule[-45pt]{.5pt}{85pt} \rotatebox{90}{\tiny \hspace{14pt} neato} \hspace{-9pt}
\begin{tikzonimage}[width=0.28\linewidth]{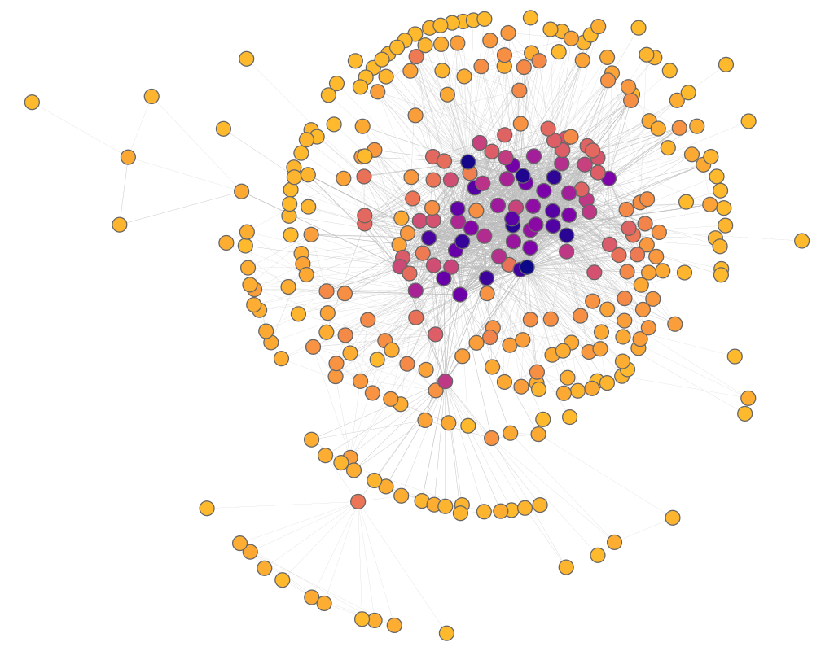}[image label]\node{$Q_{LCMC}$: 0.398};\end{tikzonimage}

\vspace{-45pt} \rotatebox{90}{\tiny \hspace{13pt} Layered} \hspace{-8pt}
\begin{tikzonimage}[width=0.325\linewidth]{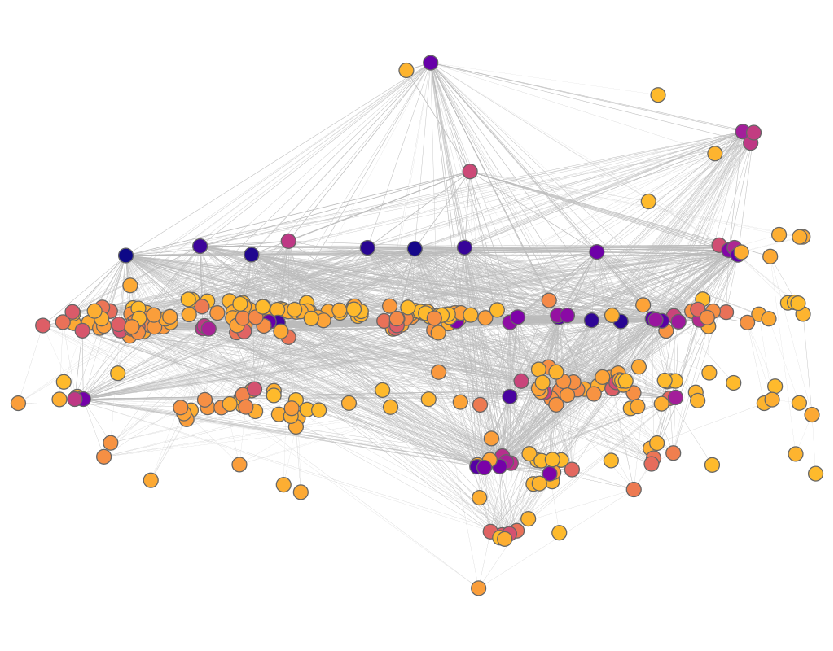}[image label]\node{$Q_{LCMC}$: 0.171};\end{tikzonimage}
\hspace{-2pt}
\begin{tikzonimage}[width=0.325\linewidth]{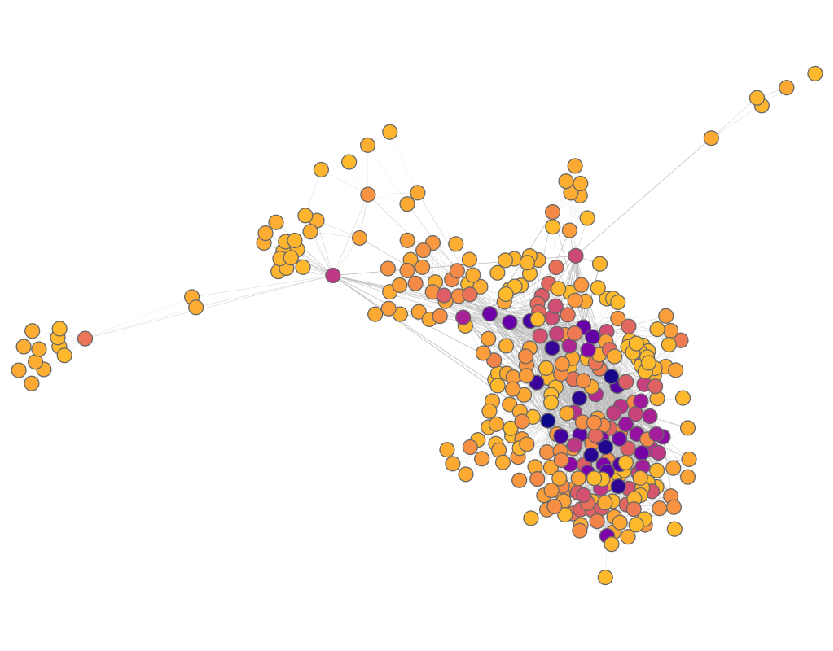}[image label]\node{$Q_{LCMC}$: 0.307};\end{tikzonimage}
\hspace{2.5pt}\rotatebox{90}{\tiny \hspace{14pt} sfdp} \hspace{-9pt}
\begin{tikzonimage}[width=0.28\linewidth]{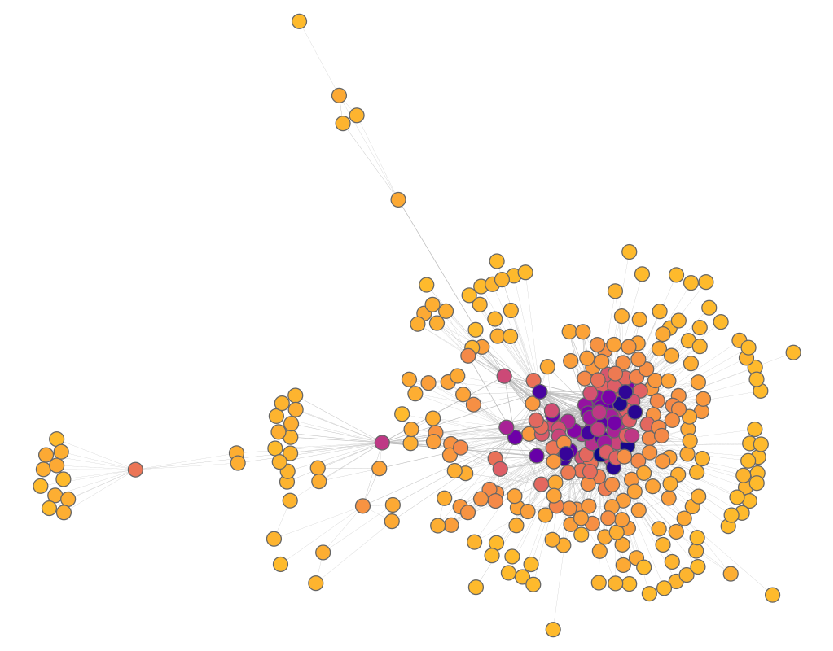}[image label]\node{$Q_{LCMC}$: 0.375};\end{tikzonimage}

{\footnotesize \hspace{23pt} Initial \hspace{33pt} Final}
\end{minipage}
}

    \vspace{-3pt}
    \caption{Example of \textit{dense} graph initial layouts (left), final layouts (middle), and the result of running neato and sfdp (right). Graphs using our technique (bottom) converge more quickly and show similar, occasionally better, results compared to the standard random approach (top).}
    \label{fig.dense}
    \vspace{-3pt}
\end{figure*}

\para{Dense Graphs} With dense graphs (see \autoref{tab:dense} and \autoref{fig.dense}), our approach generally produced higher $Q_{LCMC}$ scores. However, for several graphs, our scores were similar or slightly lower, e.g., \autoref{fig:dense:enron} and~\ref{fig:dense:usair}. In these cases, the results are still quite similar visually.
In general though for dense graphs, it seems that no matter the initial position of nodes, the layout will end in a more or less similar configuration. Importantly, even when our score is lower, our approach converges much more quickly, e.g., for the \textsc{hic 1k net~6} dataset (see \autoref{fig:large_graphs.smith}), our method produced a similar $Q_{LCMC}$ in about one fifth the time of the random layout. This property is discussed more in \autoref{sec:eval:h0:conv}.

\para{Sparse Graphs} With sparse graphs (see \autoref{tab:sparse} and \autoref{fig.sparse}), the story is a bit different, as these graphs are not so overconstrained. Using random initial layouts, quite often, their topological structures are overlapping or hidden altogether. On the other hand, our approach untangles these topological structures, leading to better final graph layouts, e.g., with \textsc{engymes-g123} dataset (see \autoref{fig:sparse:engymes}) our approach ($Q_{LCMC}=0.436$) produces higher co-ranking scores than the random layout ($Q_{LCMC}=0.374$), and the cycle structures of the graph are more clearly visible. There was one sparse case, \textsc{bcsstk}, where our approach performed slightly worse than random because it was unable to untangle the long cycle in the graph (see supplement). However, the interactive functionality discussed in \autoref{sec.eval.h1} resolved that issue.

\para{Larger Graphs} Improving layouts is particularly important for larger graphs (i.e., graphs of 1000 nodes, see \autoref{sec.eval.data}). \autoref{fig:teaser} and \ref{fig:large_graphs} show examples of larger graphs. Our approach shows better clustering structures for the \textsc{hic 1k net} and \textsc{airport} datasets, supported by the improved $Q_{LCMC}$ scores. \textsc{smith}, on the other hand, being a dense graph, shows similar clustering and identical $Q_{LCMC}$ scores. 

Overall, we observe that although our approach could improve the layout quality of dense graphs, sparse graphs almost always benefited from utilizing our approach.

\begin{figure}[!b]
    \centering
    \vspace{-5pt}
    
    \begin{minipage}[t]{0.02\linewidth}
    \rotatebox{90}{\tiny \hspace{65pt}$Q_{LCMC}$}
    \end{minipage}
    \begin{minipage}[t]{0.3\linewidth}
    \includegraphics[width=\linewidth]{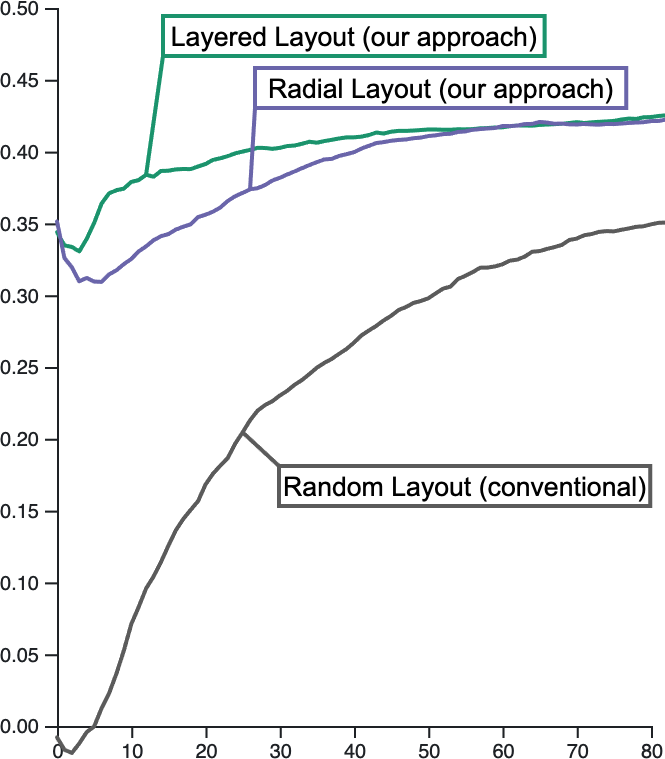}
    
    \vspace{-5pt}\hspace{55pt}
    {\tiny iterations}
    
    \vspace{-12pt}
    \subfloat[\textsc{engymes-g123}]{\hspace{\linewidth}}
    
    \end{minipage}
    \hfill
    \begin{minipage}[t]{0.3\linewidth}
    
    \includegraphics[width=\linewidth]{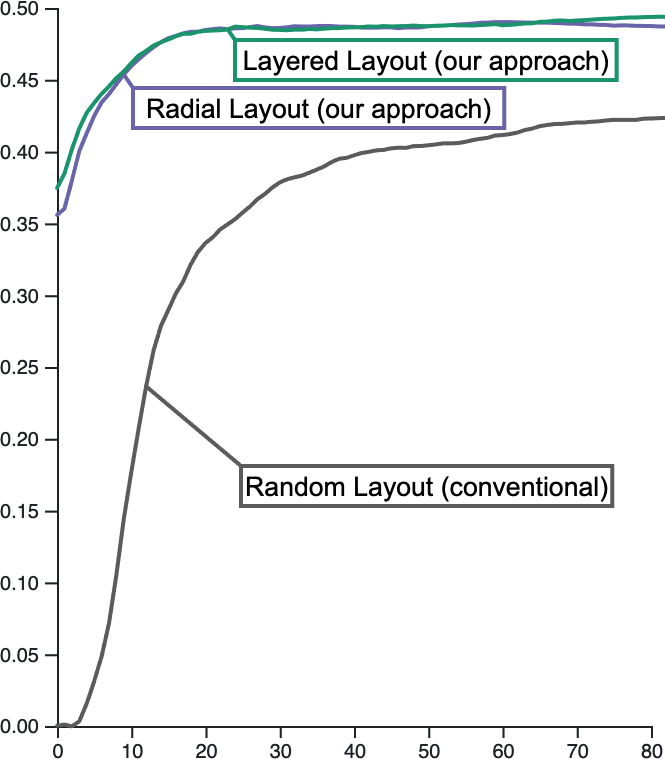}
    
    \vspace{-5pt}\hspace{55pt}
    {\tiny iterations}
    
    \vspace{-12pt}
    \subfloat[\textsc{science collab.}]{\hspace{\linewidth}}
    
    \end{minipage}
    \hfill
    \begin{minipage}[t]{0.3\linewidth}    
    
    \includegraphics[width=\linewidth]{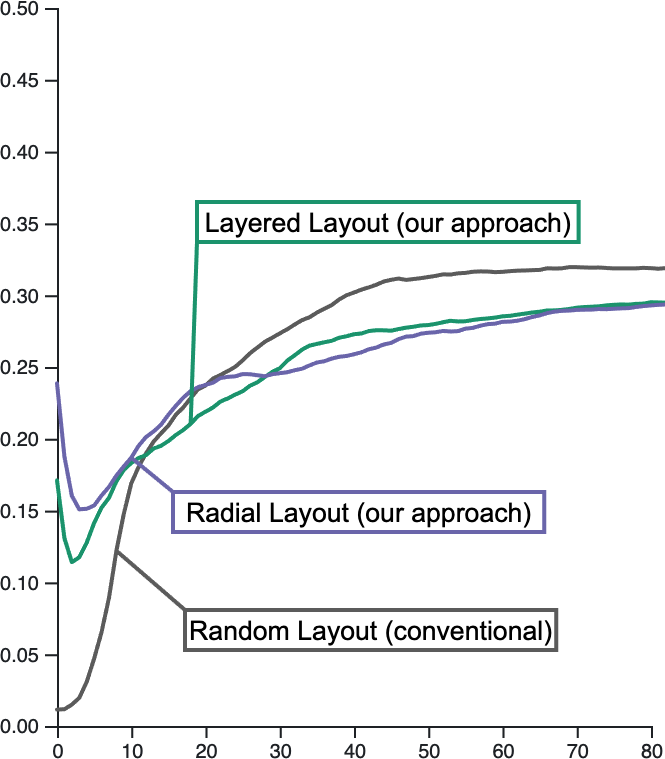}
    
    \vspace{-5pt}\hspace{55pt}
    {\tiny iterations}
    
    \vspace{-12pt}
    \subfloat[\textsc{usair 97}\label{fig:convergence:usair97}]{\hspace{\linewidth}}
    
    \end{minipage}

    \vspace{-5pt}
    \caption{Plots of $Q_{LCMC}$ against the number of iterations show that compared to the random initial layout, our approach begins with high co-ranking scores and spends most of the iterations fine-tuning.}
    \label{fig:convergence}
\end{figure}

\subsubsection{Rate of Convergence}
\label{sec:eval:h0:conv}

We compute the convergence metrics (see \autoref{sec.eval.metrics}) on all of the datasets listed in \autoref{tab:dense} and~\ref{tab:sparse}, and we focus on the convergence of the LCMC ($C_{LCMC}$). Our results show that for all except one dataset, \textsc{usair 97}, using our approach converged faster than random layouts, often significantly so. \autoref{fig:convergence} shows plots of $Q_{LCMC}$ against iterations for three example datasets, including \textsc{usair 97}. In all cases, our approach starts with a much higher $Q_{LCMC}$ score and fine-tunes the results. One interesting observation is that the $Q_{LCMC}$ sometimes starts high and dips, e.g., in \autoref{fig:convergence:usair97}. This results from our good initial layout being in a high energy state which is distorted by the force-directed layout before settling in a good quality low energy state.

\begin{figure*}[!ht]

\subfloat[\textsc{bio-diseasome}\label{fig:sparse:diseasome}]{
\begin{minipage}[t]{0.315\linewidth}
\rotatebox{90}{\tiny \hspace{13pt} Random} \hspace{-8pt}
\begin{tikzonimage}[width=0.315\linewidth]{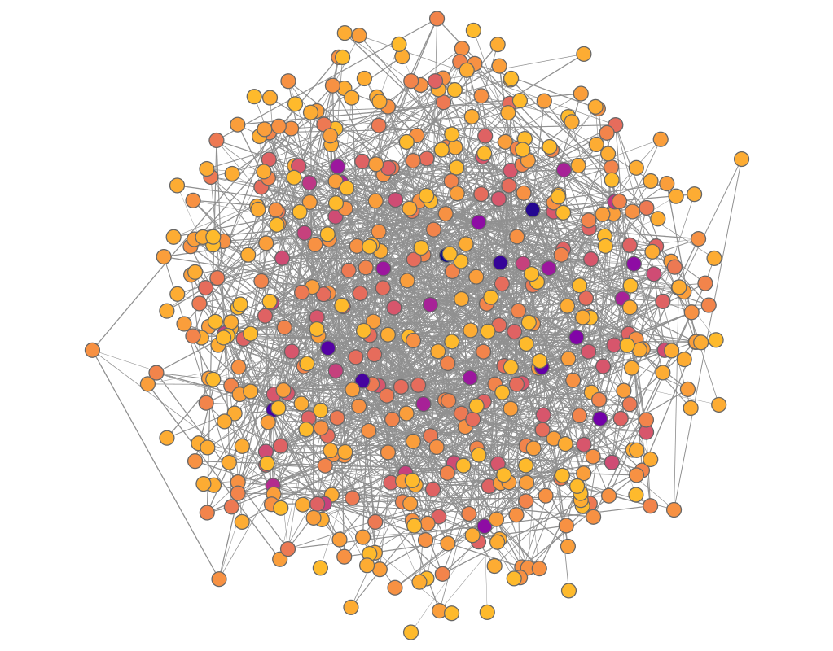}[image label]\node{$Q_{LCMC}$: 0.000};\end{tikzonimage}
\hspace{-2pt}
\begin{tikzonimage}[width=0.315\linewidth]{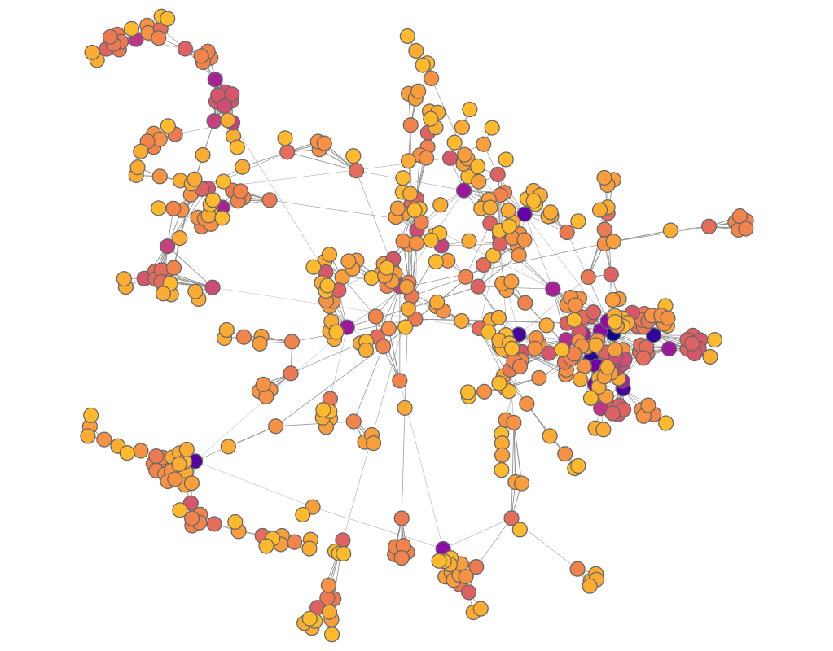}[image label]\node{$Q_{LCMC}$: 0.407};\end{tikzonimage}
\hspace{-1pt} \rule[-45pt]{.5pt}{85pt} \rotatebox{90}{\tiny \hspace{14pt} neato} \hspace{-9pt}
\begin{tikzonimage}[width=0.28\linewidth]{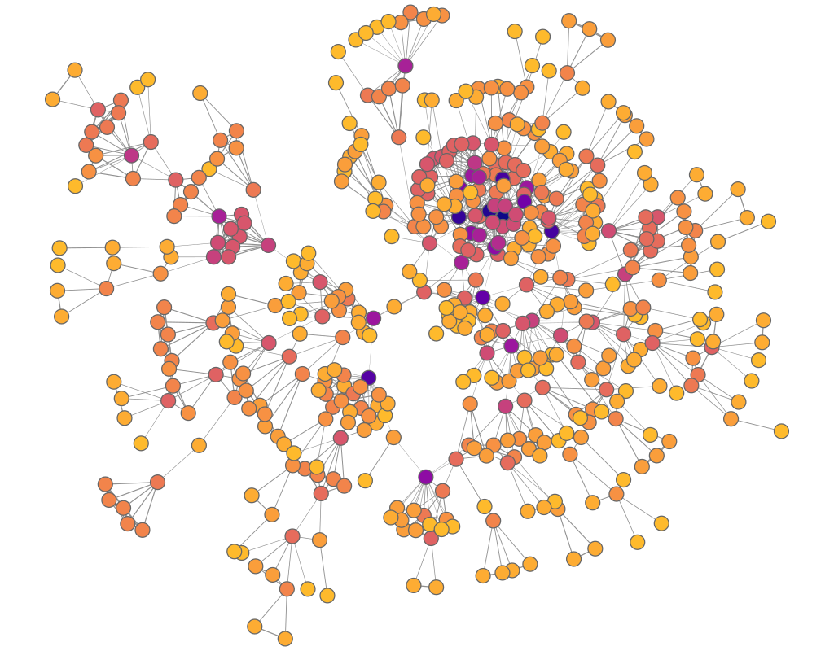}[image label]\node{$Q_{LCMC}$: 0.509};\end{tikzonimage}

\vspace{-45pt} \rotatebox{90}{\tiny \hspace{13pt} Radial} \hspace{-8pt}
\begin{tikzonimage}[width=0.315\linewidth]{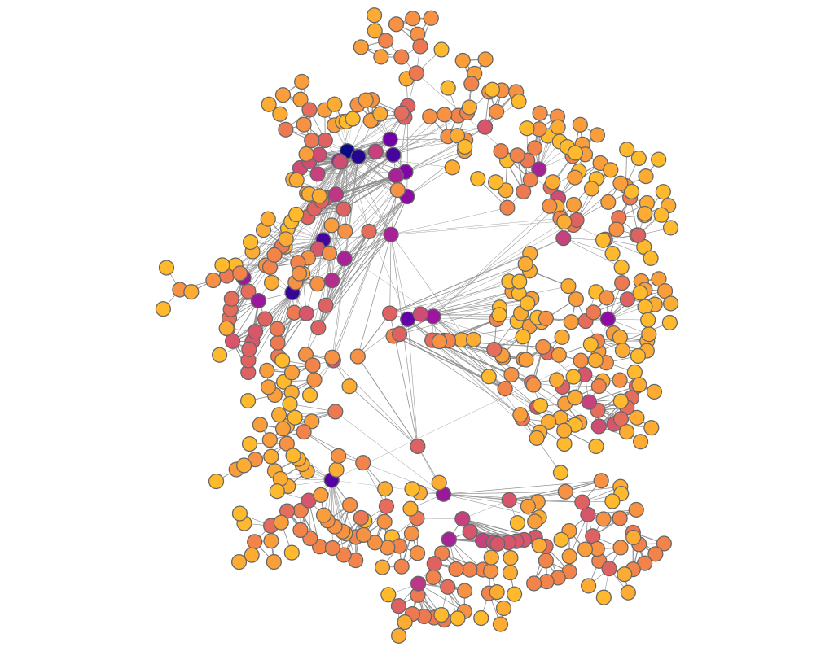}[image label]\node{$Q_{LCMC}$: 0.369};\end{tikzonimage}
\hspace{-2pt}
\begin{tikzonimage}[width=0.315\linewidth]{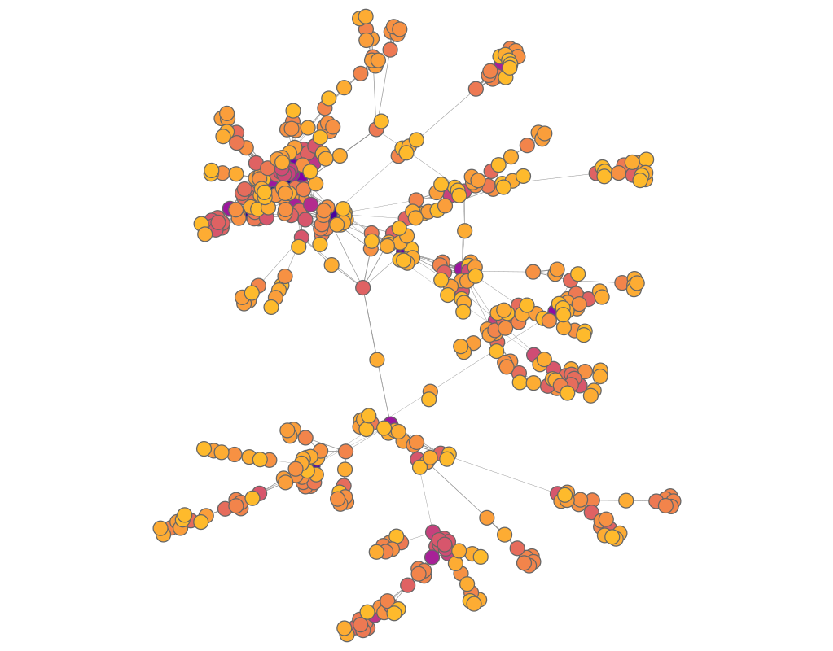}[image label]\node{$Q_{LCMC}$: 0.478};\end{tikzonimage}
\hspace{2.5pt}\rotatebox{90}{\tiny \hspace{15pt} sfdp} \hspace{-9pt}
\begin{tikzonimage}[width=0.28\linewidth]{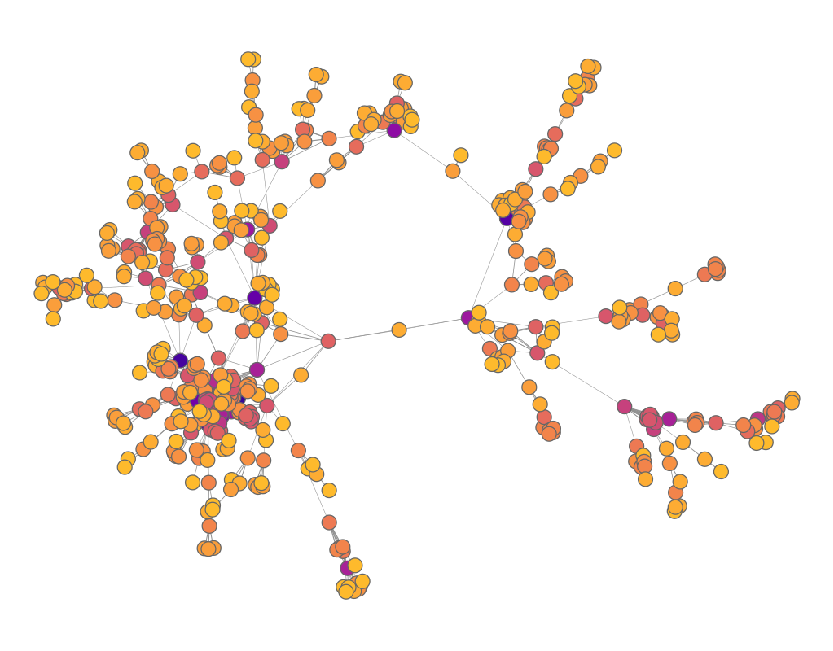}[image label]\node{$Q_{LCMC}$: 0.518};\end{tikzonimage}

{\footnotesize \hspace{23pt} Initial \hspace{33pt} Final}
\end{minipage}
}
\hfill
\subfloat[\textsc{circular ladder graph (100)}\label{fig:sparse:circular}]{
\begin{minipage}[t]{0.315\linewidth}
\rotatebox{90}{\tiny \hspace{13pt} Random} \hspace{-8pt}
\begin{tikzonimage}[width=0.315\linewidth]{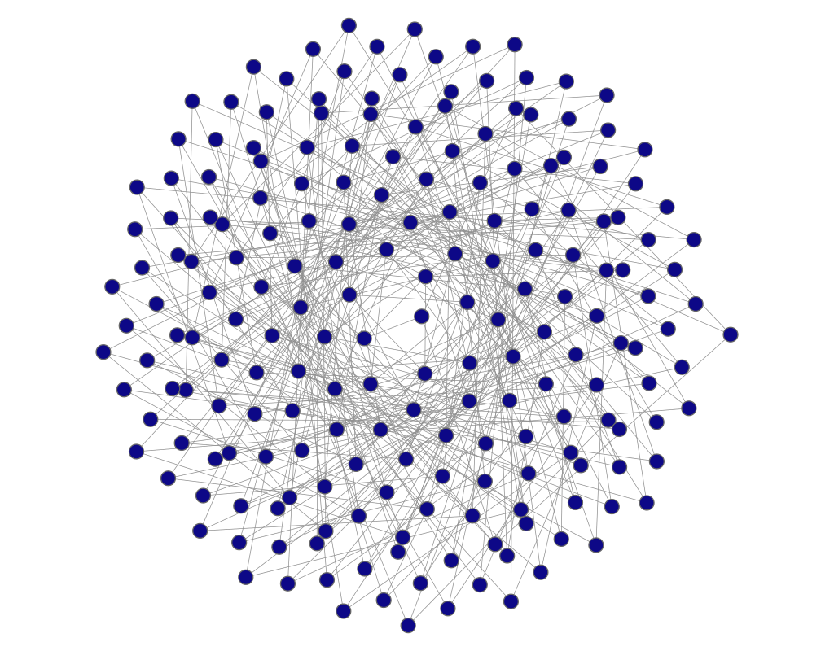}[image label]\node{$Q_{LCMC}$: -0.035};\end{tikzonimage}
\hspace{-2pt}
\begin{tikzonimage}[width=0.315\linewidth]{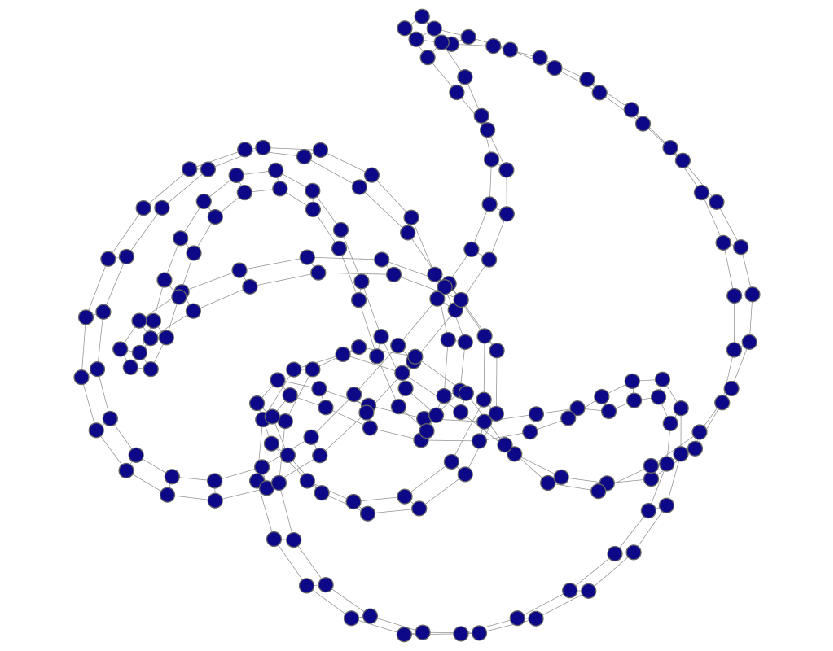}[image label]\node{$Q_{LCMC}$: 0.455};\end{tikzonimage}
\hspace{-1pt} \rule[-45pt]{.5pt}{85pt} \rotatebox{90}{\tiny \hspace{14pt} neato} \hspace{-9pt}
\begin{tikzonimage}[width=0.28\linewidth]{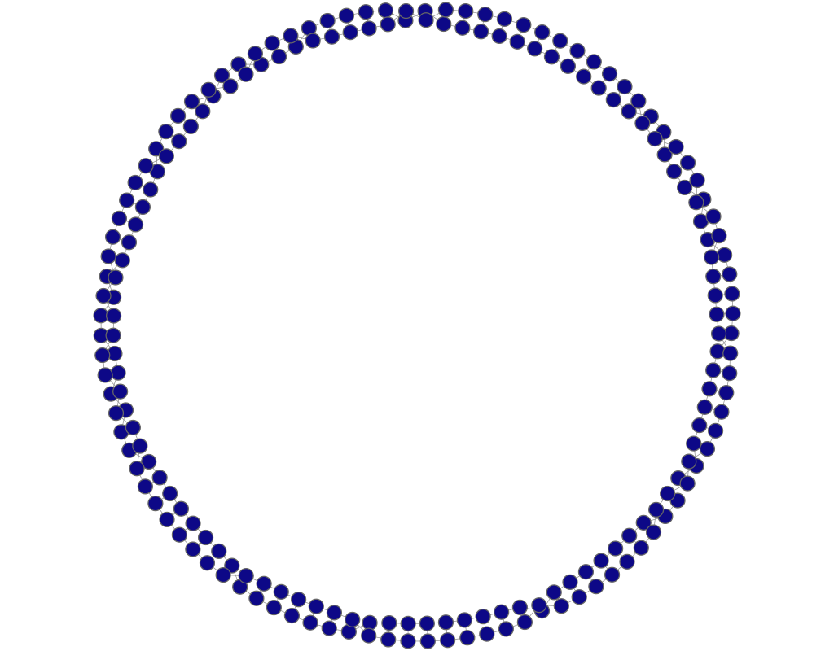}[image label]\node{$Q_{LCMC}$: 0.834};\end{tikzonimage}

\vspace{-45pt} \rotatebox{90}{\tiny \hspace{13pt} Radial} \hspace{-8pt}
\begin{tikzonimage}[width=0.315\linewidth]{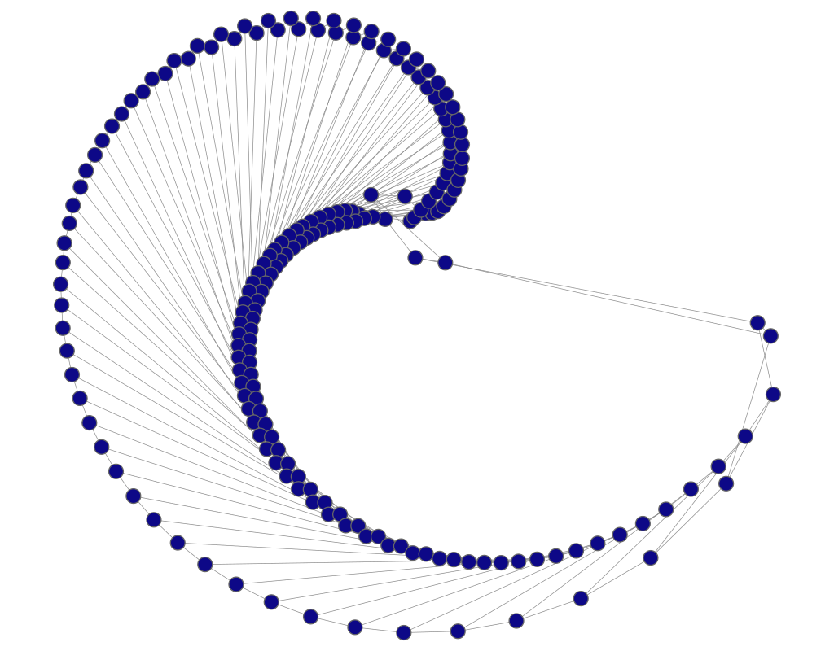}[image label]\node{$Q_{LCMC}$: 0.430};\end{tikzonimage}
\hspace{-2pt}
\begin{tikzonimage}[width=0.315\linewidth]{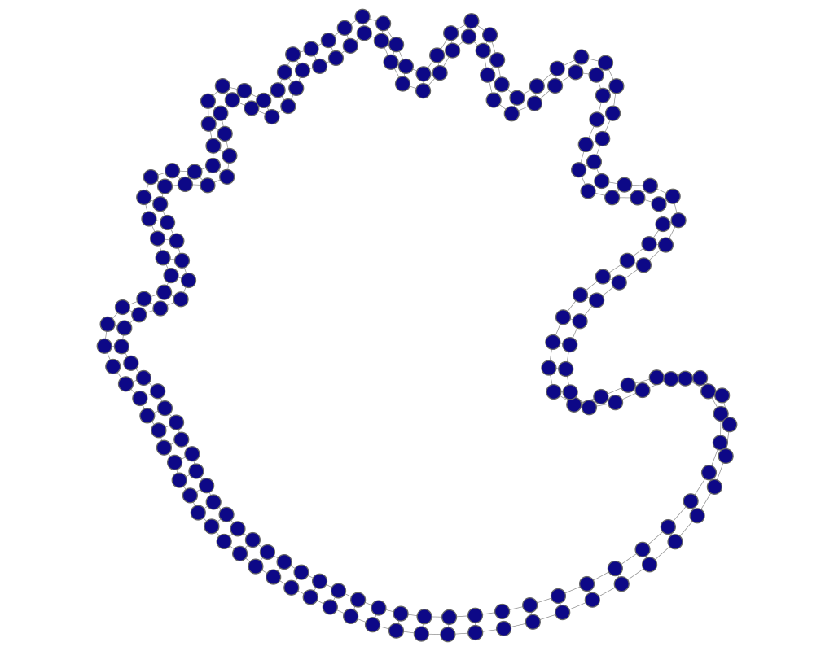}[image label]\node{$Q_{LCMC}$: 0.814};\end{tikzonimage}
\hspace{2.5pt}\rotatebox{90}{\tiny \hspace{15pt} sfdp} \hspace{-9pt}
\begin{tikzonimage}[width=0.28\linewidth]{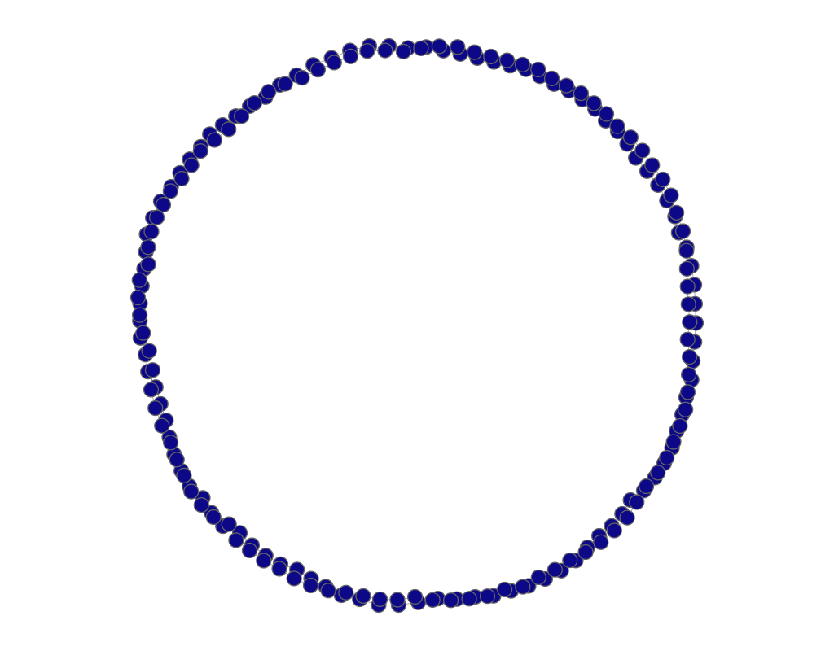}[image label]\node{$Q_{LCMC}$: 0.816};\end{tikzonimage}

{\footnotesize \hspace{23pt} Initial \hspace{33pt} Final}
\end{minipage}
}
\hfill
\subfloat[\textsc{engymes-g123}\label{fig:sparse:engymes}]{
\begin{minipage}[t]{0.315\linewidth}
\rotatebox{90}{\tiny \hspace{13pt} Random} \hspace{-8pt}
\begin{tikzonimage}[width=0.315\linewidth]{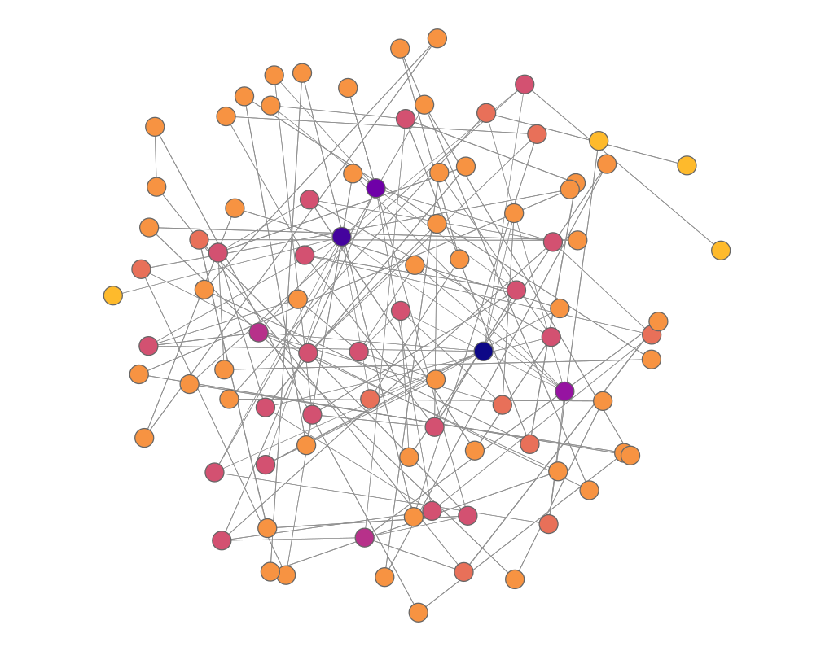}[image label]\node{$Q_{LCMC}$: -0.008};\end{tikzonimage}
\hspace{-2pt}
\begin{tikzonimage}[width=0.315\linewidth]{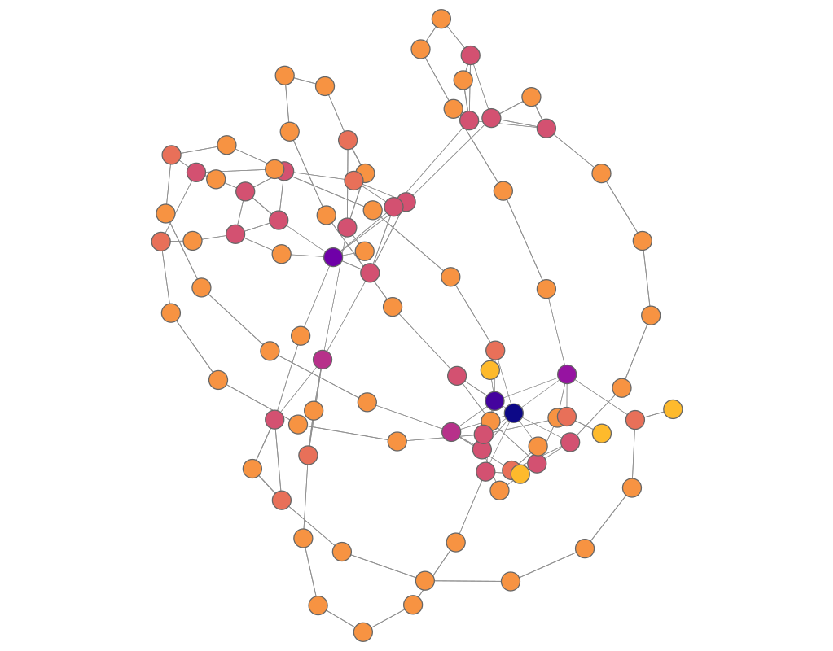}[image label]\node{$Q_{LCMC}$: 0.374};\end{tikzonimage}
\hspace{-1pt} \rule[-45pt]{.5pt}{85pt} \rotatebox{90}{\tiny \hspace{14pt} neato} \hspace{-9pt}
\begin{tikzonimage}[width=0.28\linewidth]{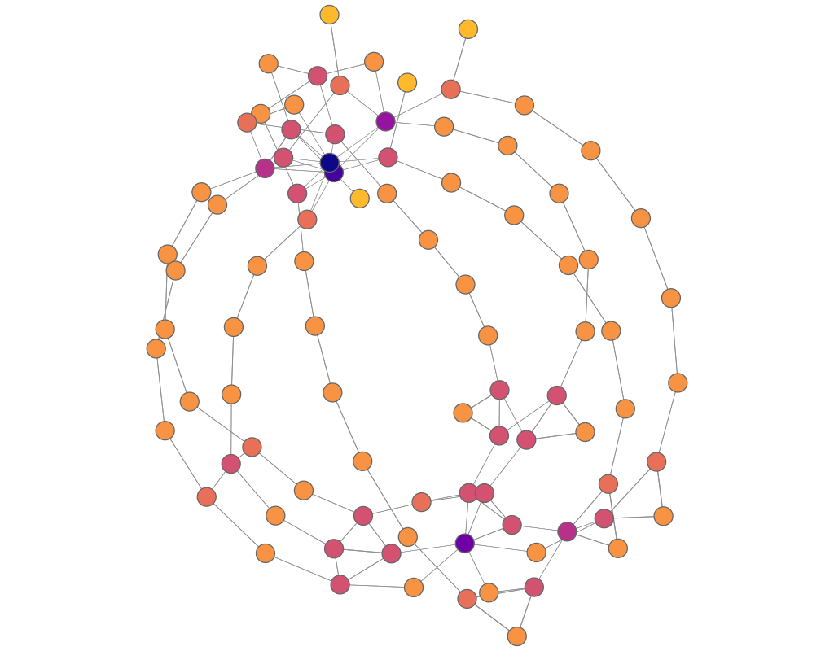}[image label]\node{$Q_{LCMC}$: 0.505};\end{tikzonimage}

\vspace{-45pt} \rotatebox{90}{\tiny \hspace{13pt} Layered} \hspace{-8pt}
\begin{tikzonimage}[width=0.315\linewidth]{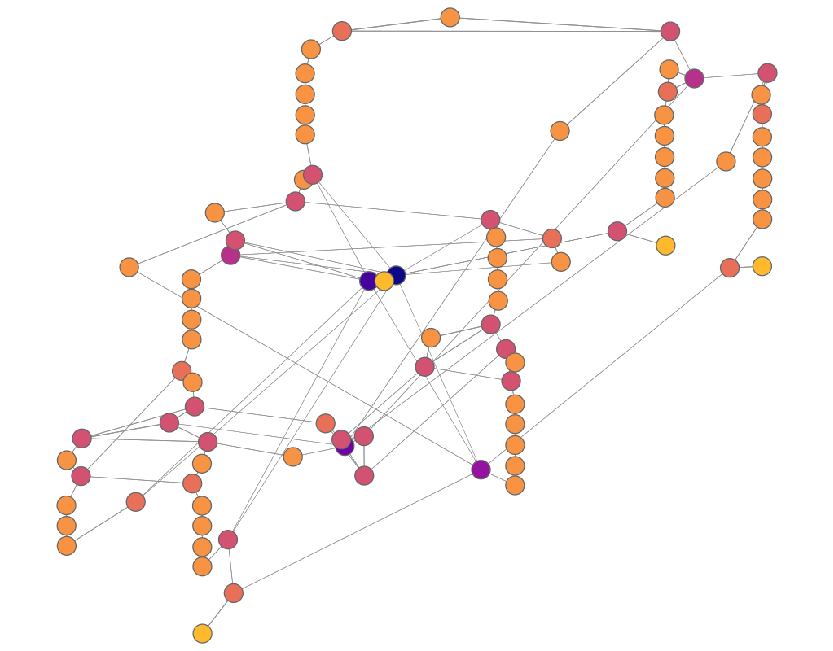}[image label]\node{$Q_{LCMC}$: 0.344};\end{tikzonimage}
\hspace{-2pt}
\begin{tikzonimage}[width=0.315\linewidth]{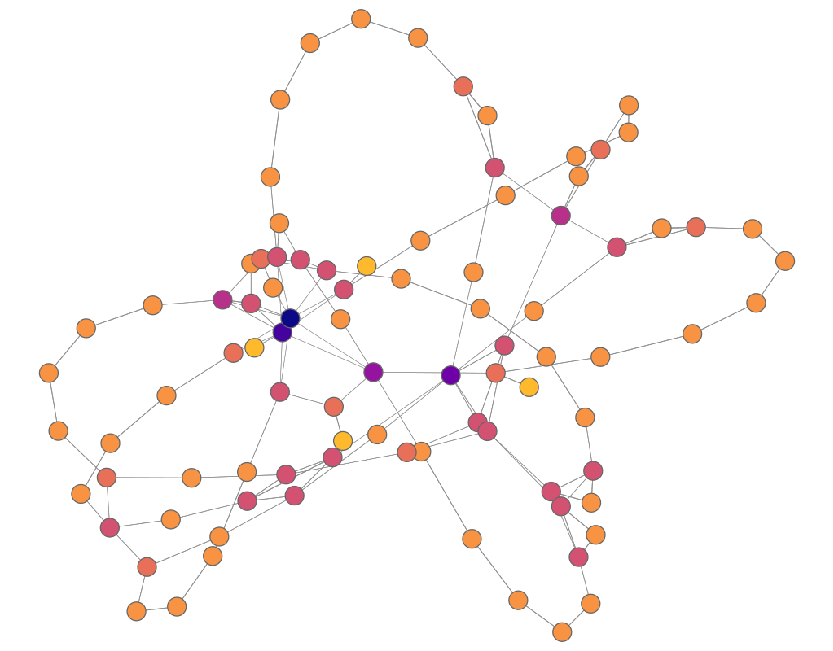}[image label]\node{$Q_{LCMC}$: 0.436};\end{tikzonimage}
\hspace{2.5pt}\rotatebox{90}{\tiny \hspace{15pt} sfdp} \hspace{-9pt}
\begin{tikzonimage}[width=0.28\linewidth]{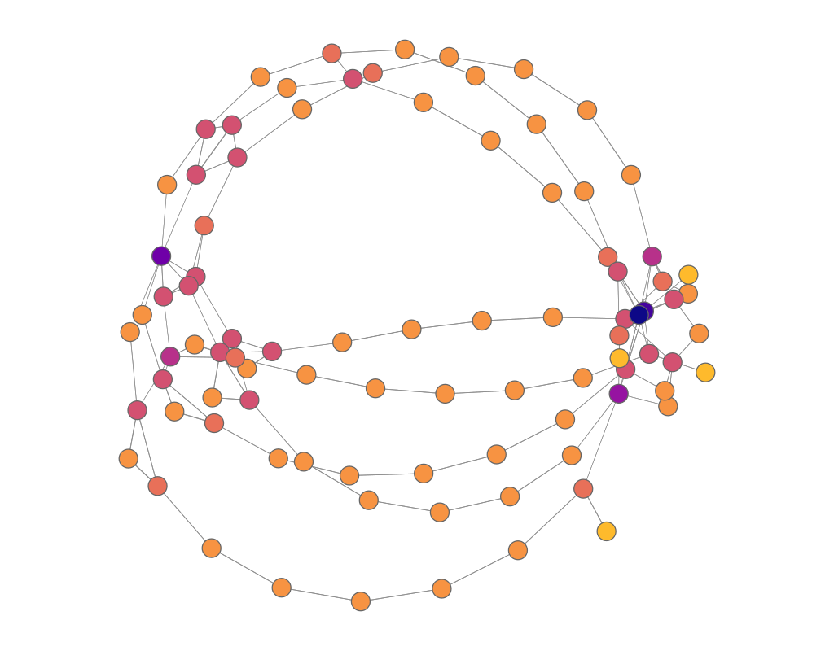}[image label]\node{$Q_{LCMC}$: 0.475};\end{tikzonimage}

{\footnotesize \hspace{23pt} Initial \hspace{33pt} Final}
\end{minipage}
}

\vspace{-3pt}
\subfloat[\textsc{lollipop (10,50)}\label{fig:sparse:lollipop}]{
\begin{minipage}[t]{0.315\linewidth}
\rotatebox{90}{\tiny \hspace{13pt} Random} \hspace{-8pt}
\begin{tikzonimage}[width=0.315\linewidth]{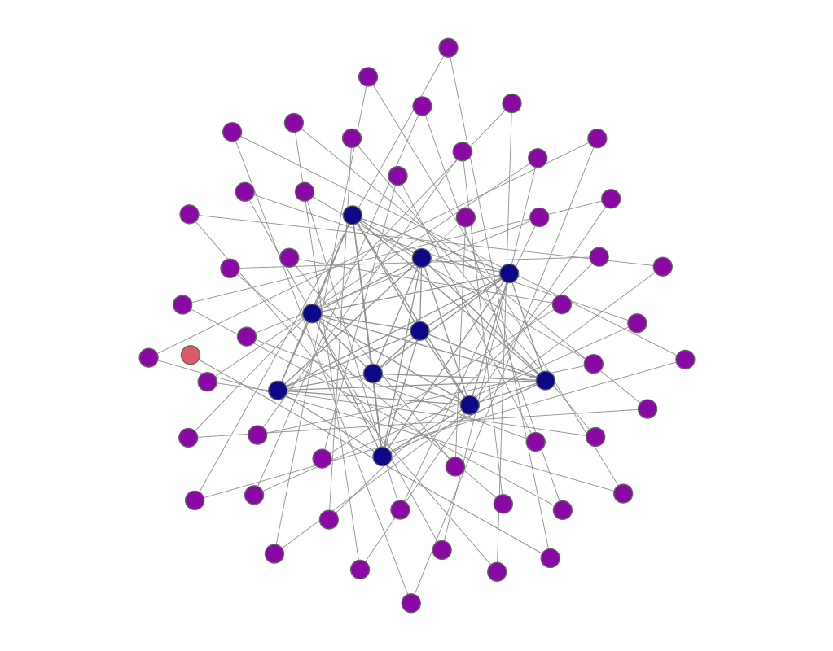}[image label]\node{$Q_{LCMC}$: 0.021};\end{tikzonimage}
\hspace{-2pt}
\begin{tikzonimage}[width=0.315\linewidth]{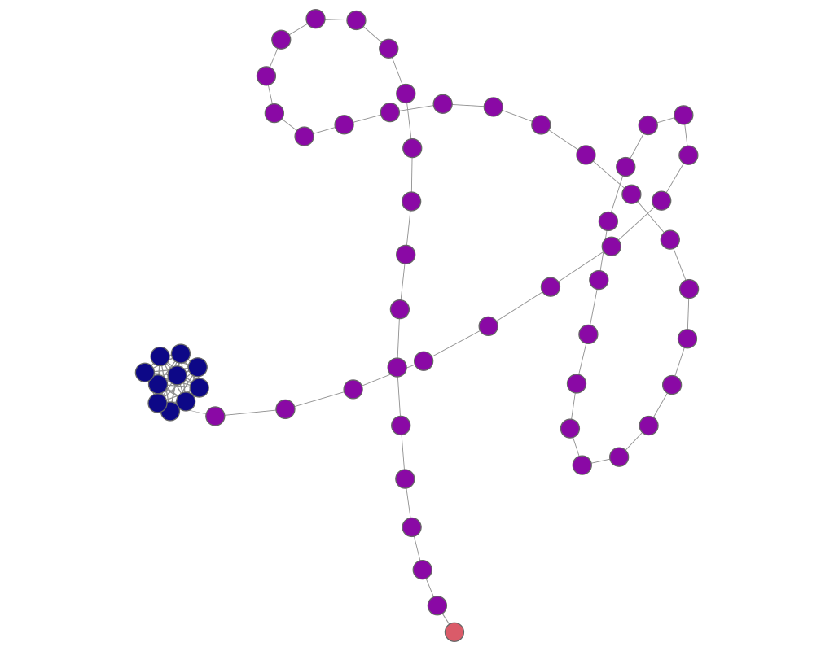}[image label]\node{$Q_{LCMC}$: 0.444};\end{tikzonimage}
\hspace{-1pt} \rule[-45pt]{.5pt}{85pt} \rotatebox{90}{\tiny \hspace{14pt} neato} \hspace{-9pt}
\begin{tikzonimage}[width=0.28\linewidth]{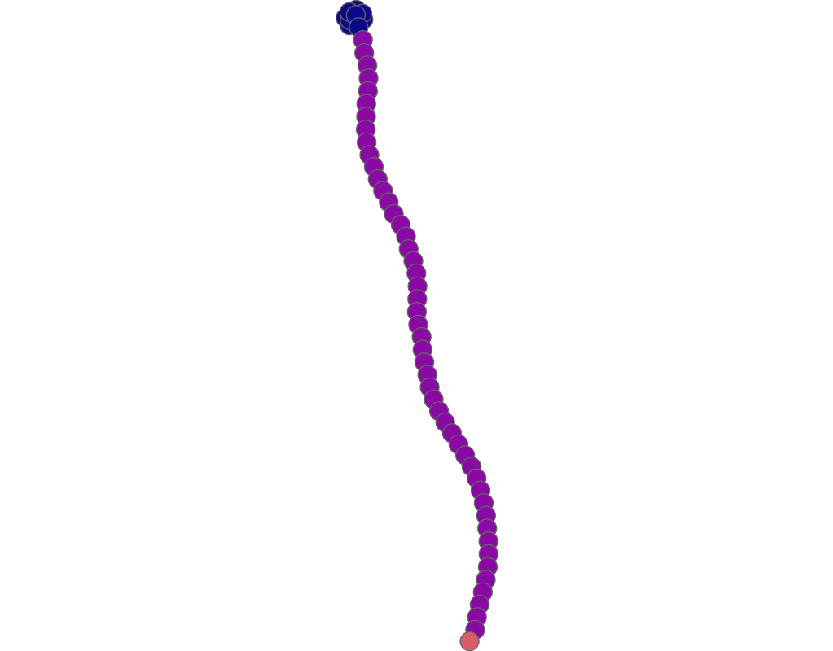}[image label]\node{$Q_{LCMC}$: 0.764};\end{tikzonimage}

\vspace{-45pt} \rotatebox{90}{\tiny \hspace{13pt} Layered} \hspace{-8pt}
\begin{tikzonimage}[width=0.315\linewidth]{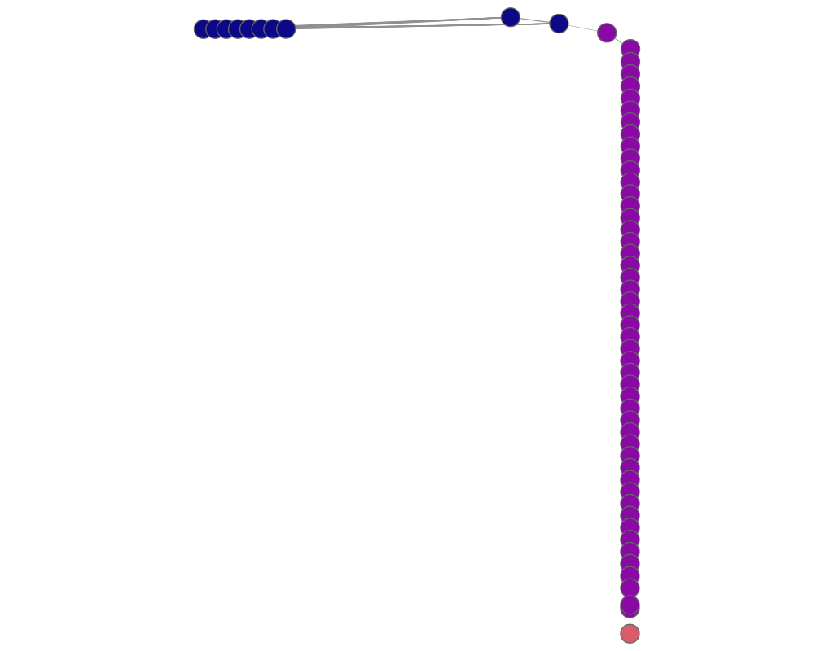}[image label]\node{$Q_{LCMC}$: 0.717};\end{tikzonimage}
\hspace{-2pt}
\begin{tikzonimage}[width=0.315\linewidth]{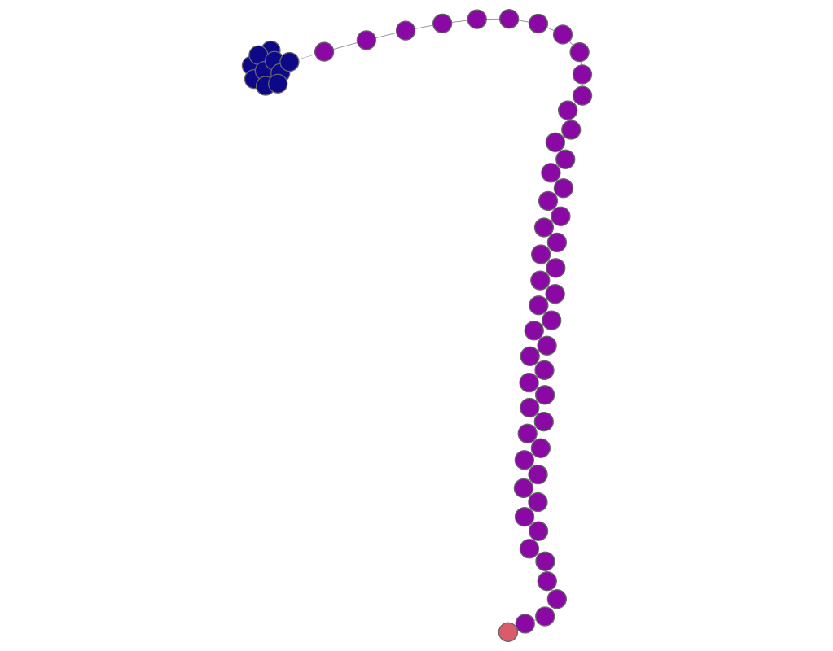}[image label]\node{$Q_{LCMC}$: 0.755};\end{tikzonimage}
\hspace{2.5pt}\rotatebox{90}{\tiny \hspace{15pt} sfdp} \hspace{-9pt}
\begin{tikzonimage}[width=0.28\linewidth]{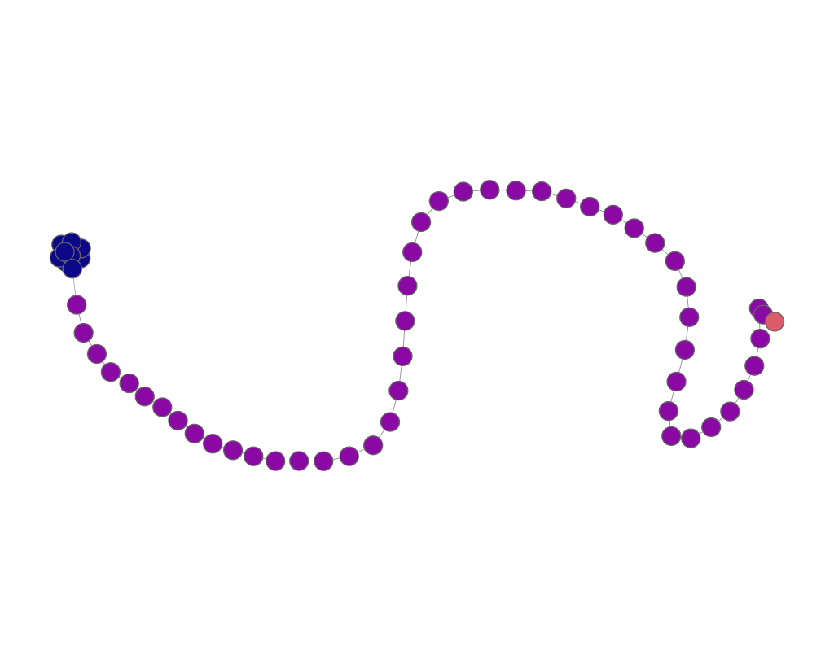}[image label]\node{$Q_{LCMC}$: 0.734};\end{tikzonimage}

{\footnotesize \hspace{23pt} Initial \hspace{33pt} Final}
\end{minipage}
}
\hfill
\subfloat[\textsc{map of science}\label{fig:sparse:map}]{
\begin{minipage}[t]{0.315\linewidth}
\rotatebox{90}{\tiny \hspace{13pt} Random} \hspace{-8pt}
\begin{tikzonimage}[width=0.315\linewidth]{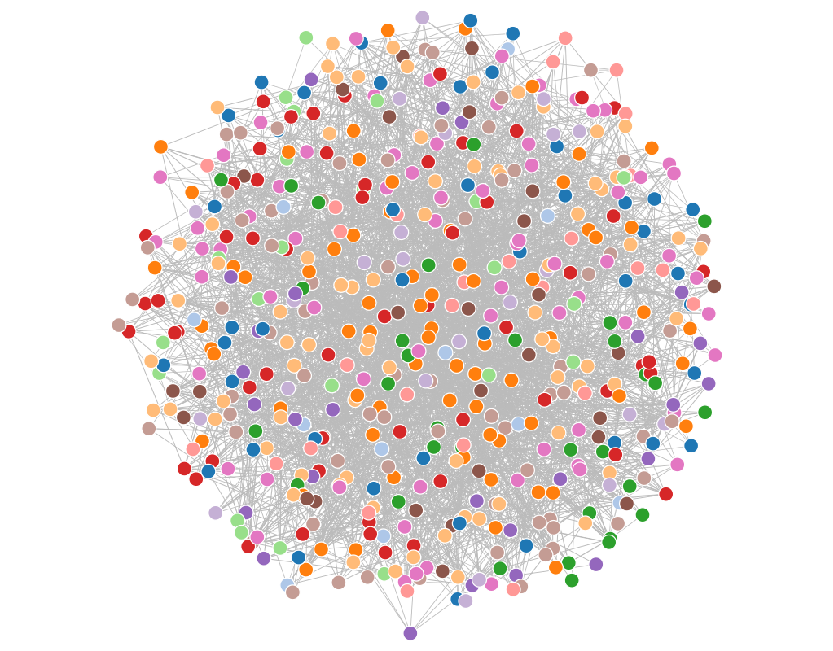}[image label]\node{$Q_{LCMC}$: 0.004};\end{tikzonimage}
\hspace{-2pt}
\begin{tikzonimage}[width=0.315\linewidth]{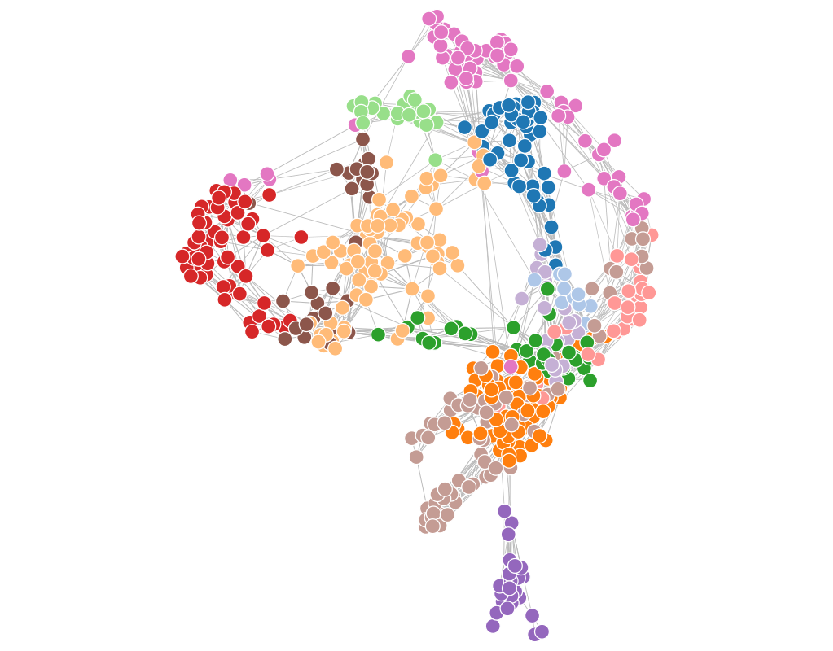}[image label]\node{$Q_{LCMC}$: 0.361};\end{tikzonimage}
\hspace{-1pt} \rule[-45pt]{.5pt}{85pt} \rotatebox{90}{\tiny \hspace{14pt} neato} \hspace{-9pt}
\begin{tikzonimage}[width=0.28\linewidth]{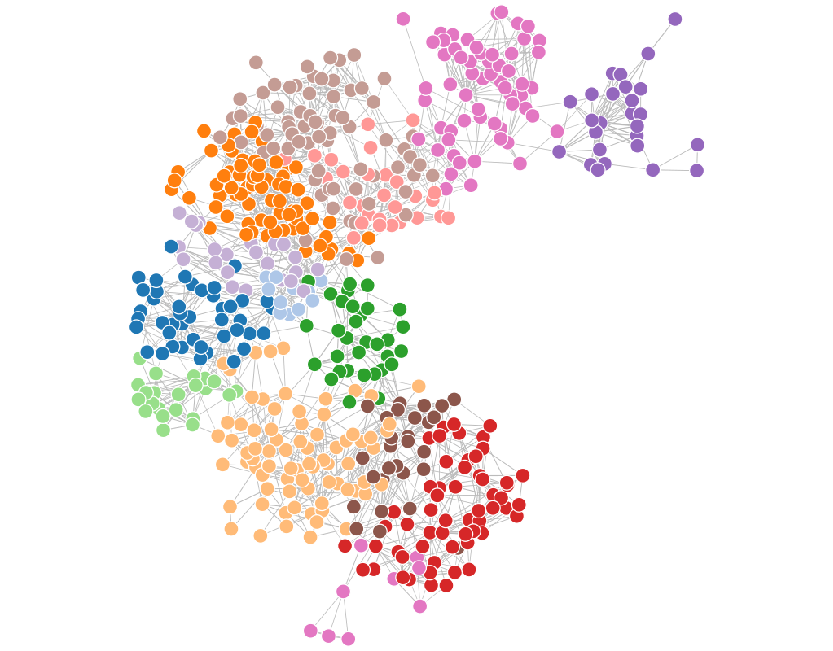}[image label]\node{$Q_{LCMC}$: 0.514};\end{tikzonimage}

\vspace{-45pt} \rotatebox{90}{\tiny \hspace{13pt} Radial} \hspace{-8pt}
\begin{tikzonimage}[width=0.315\linewidth]{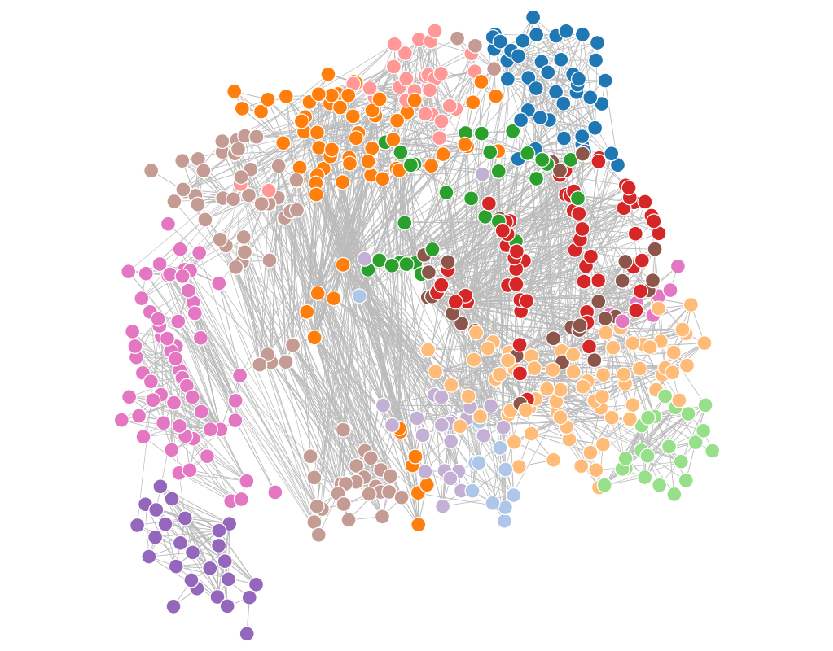}[image label]\node{$Q_{LCMC}$: 0.183};\end{tikzonimage}
\hspace{-2pt}
\begin{tikzonimage}[width=0.315\linewidth]{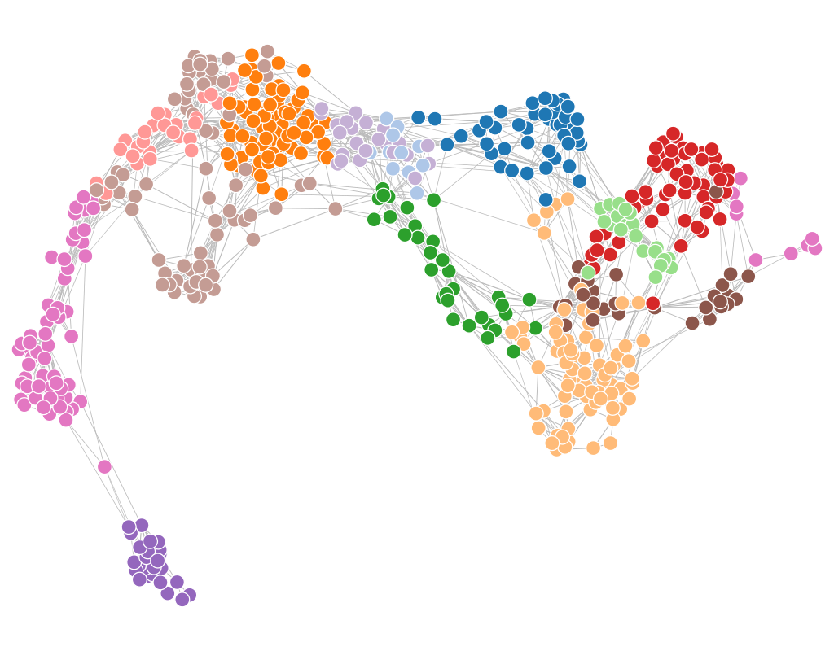}[image label]\node{$Q_{LCMC}$: 0.402};\end{tikzonimage}
\hspace{2.5pt}\rotatebox{90}{\tiny \hspace{15pt} sfdp} \hspace{-9pt}
\begin{tikzonimage}[width=0.28\linewidth]{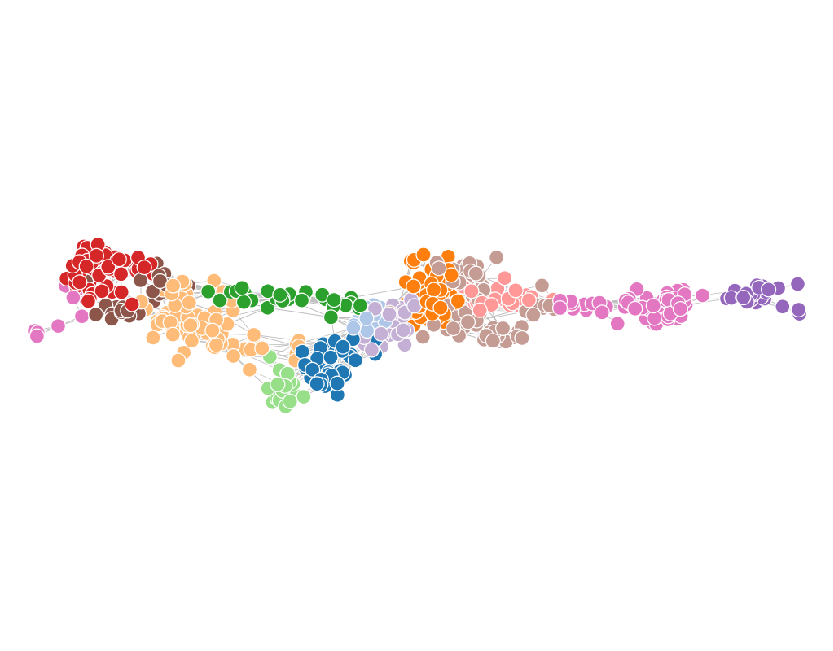}[image label]\node{$Q_{LCMC}$: 0.391};\end{tikzonimage}

{\footnotesize \hspace{23pt} Initial \hspace{33pt} Final}
\end{minipage}
}
\hfill
\subfloat[\textsc{science collaboration}\label{fig:sparse:collaboration}]{
\begin{minipage}[t]{0.315\linewidth}
\rotatebox{90}{\tiny \hspace{13pt} Random} \hspace{-8pt}
\begin{tikzonimage}[width=0.315\linewidth]{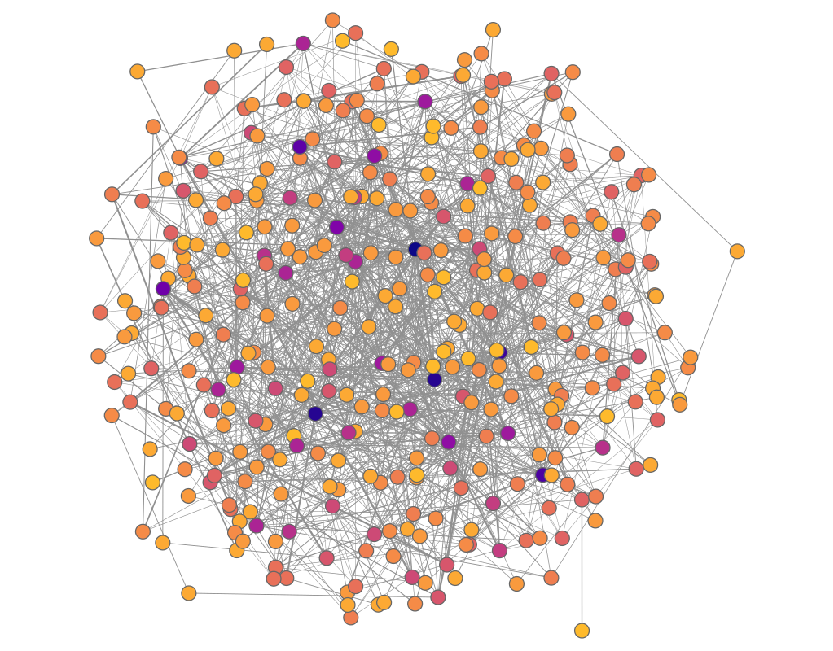}[image label]\node{$Q_{LCMC}$: 0.000};\end{tikzonimage}
\hspace{-2pt}
\begin{tikzonimage}[width=0.315\linewidth]{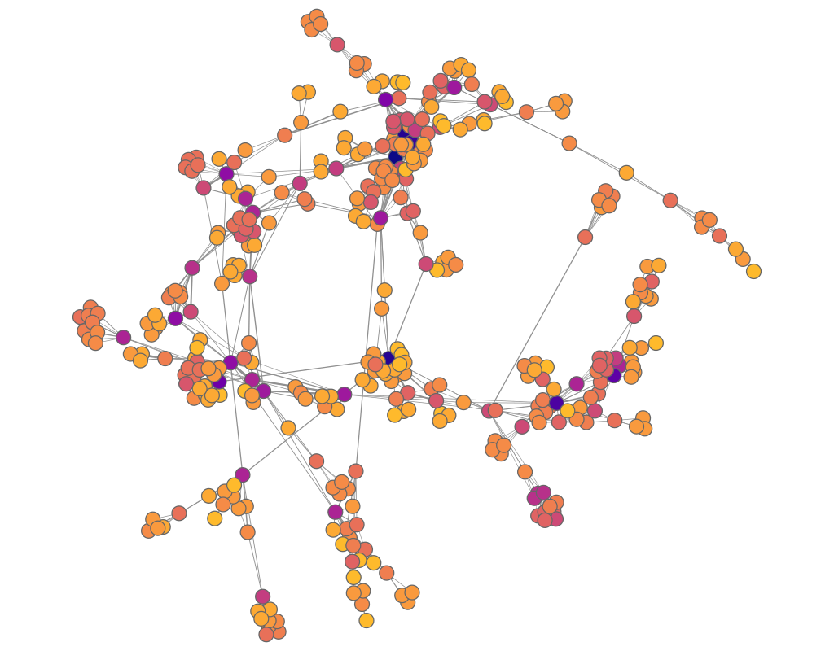}[image label]\node{$Q_{LCMC}$: 0.433};\end{tikzonimage}
\hspace{-1pt} \rule[-45pt]{.5pt}{85pt} \rotatebox{90}{\tiny \hspace{14pt} neato} \hspace{-9pt}
\begin{tikzonimage}[width=0.28\linewidth]{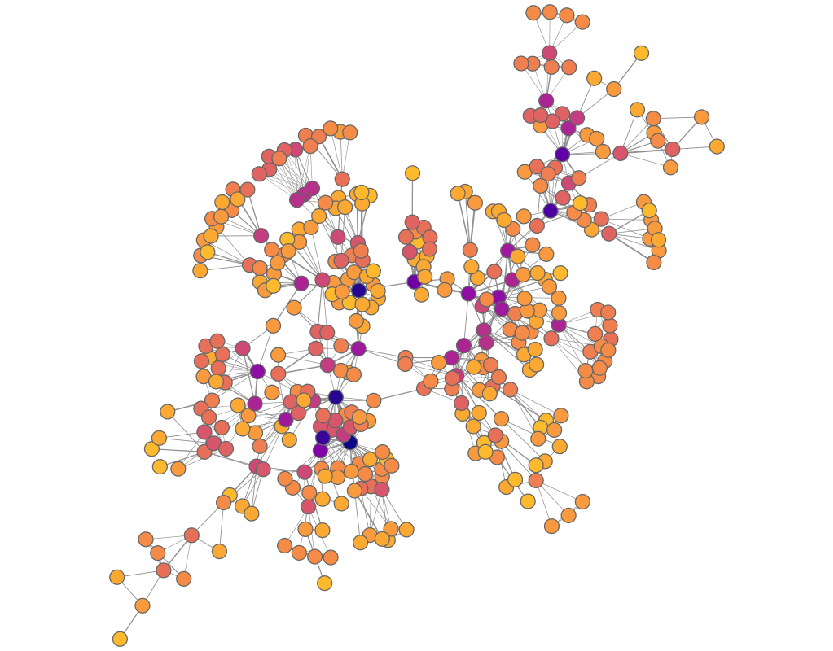}[image label]\node{$Q_{LCMC}$: 0.545};\end{tikzonimage}

\vspace{-45pt} \rotatebox{90}{\tiny \hspace{13pt} Layered} \hspace{-8pt}
\begin{tikzonimage}[width=0.315\linewidth]{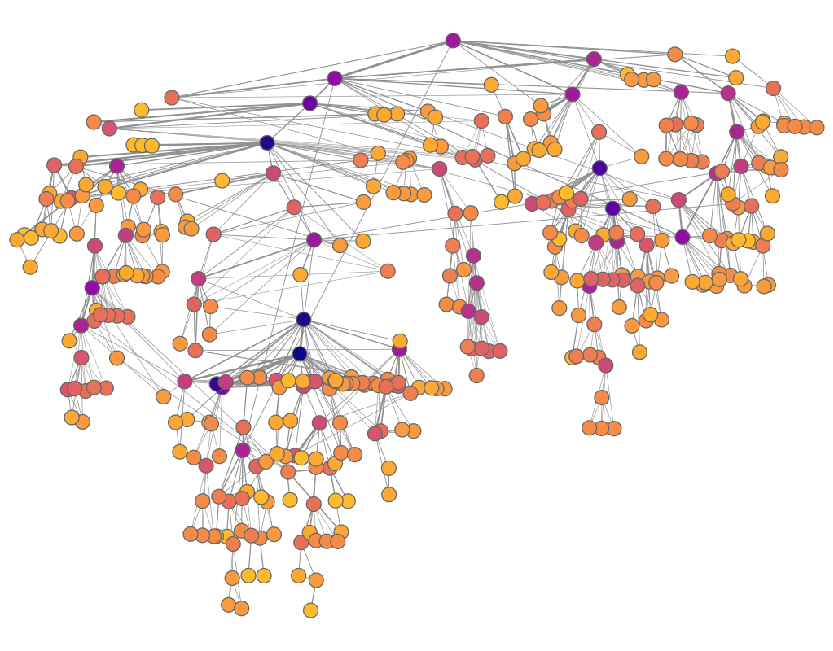}[image label]\node{$Q_{LCMC}$: 0.375};\end{tikzonimage}
\hspace{-2pt}
\begin{tikzonimage}[width=0.315\linewidth]{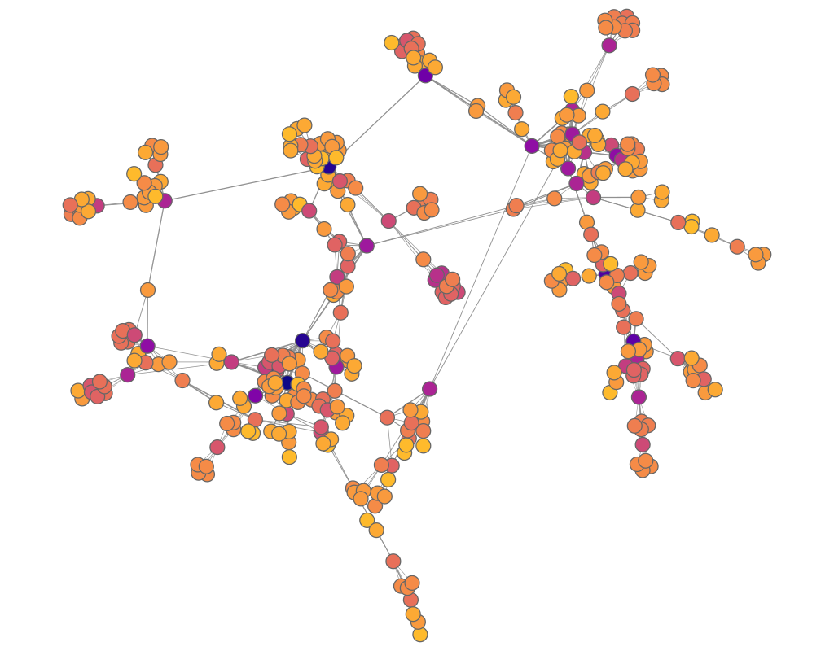}[image label]\node{$Q_{LCMC}$: 0.500};\end{tikzonimage}
\hspace{2.5pt}\rotatebox{90}{\tiny \hspace{15pt} sfdp} \hspace{-9pt}
\begin{tikzonimage}[width=0.28\linewidth]{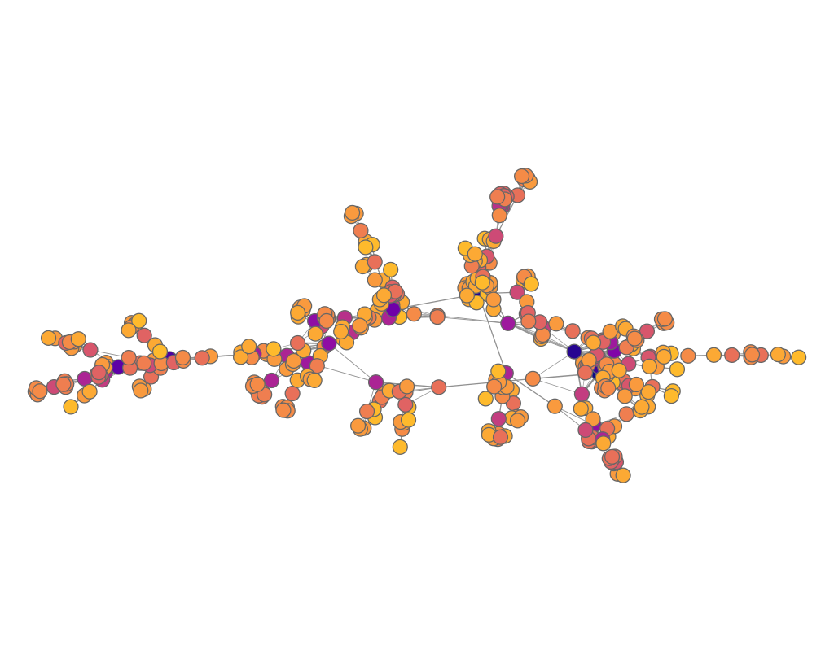}[image label]\node{$Q_{LCMC}$: 0.555};\end{tikzonimage}

{\footnotesize \hspace{23pt} Initial \hspace{33pt} Final}
\end{minipage}
}

    \caption{Example of \textit{sparse} graph initial layouts (left), final layouts (middle), and the result of running neato and sfdp (right). Graphs using our technique (bottom) converge more quickly and produce better layouts than the standard random approach (top).}
    \label{fig.sparse}
\end{figure*}

The rate of convergence is particularly important for larger graphs, where the average time per iteration is higher. For the large datasets, \textsc{airport}, \textsc{hic 1k net~6}, and \textsc{smith}, our approach converged faster than random layouts, with 37 vs.\ 76, 55 vs.\ 112, and 4 vs.\ 28 iterations, respectively. This phenomenon is also visible in \autoref{fig:teaser} and \ref{fig:large_graphs}, where graphs layouts are shown at several intervals---initial, 5 iterations, 10 iterations, etc. In all cases, the structures shown in the final graph are visible much earlier with our approach (iteration 5 or 10) than with a random layout (iteration 25 or more).

Therefore, we conclude that our approach  significantly improves convergence in most cases.

\begin{figure}[!b]

\hspace{-12pt}
\resizebox{15pt}{!}{\rotatebox{90}{\hspace{0pt}\subfloat[\textsc{science collab.}]{\hspace{80pt}}}}
\hspace{-8pt}\begin{tikzonimage}[width=0.24\linewidth]{fig/auto/init_layout_science_collab_network_random_final}[image label]\node{$Q_{LCMC}$: 0.433};\end{tikzonimage}
\hspace{-5pt}\begin{tikzonimage}[width=0.24\linewidth]{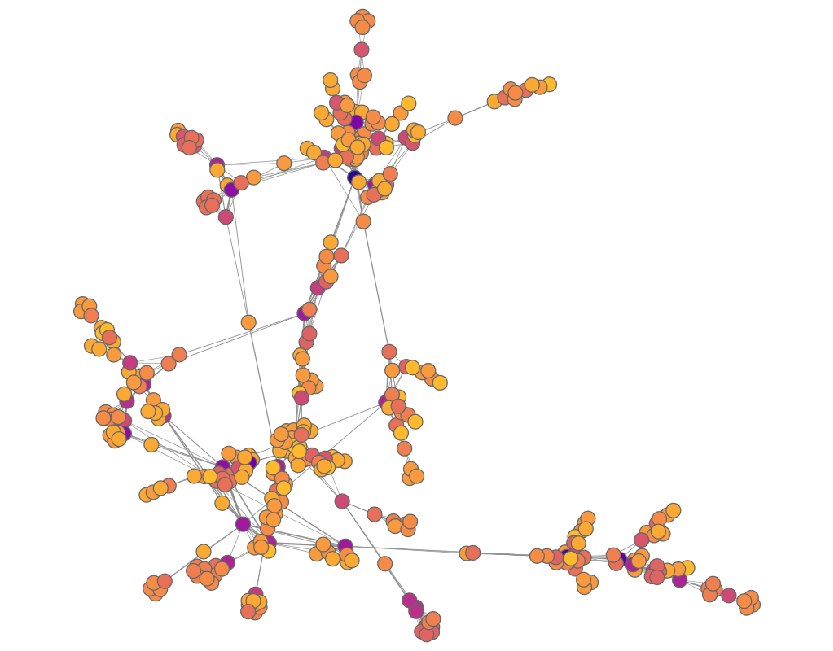}[image label]\node{$Q_{LCMC}$: 0.467};\end{tikzonimage}
\hspace{-6pt}
\textcolor{hlineColor}{\rule[-55pt]{.5pt}{100pt}}
\hspace{-7pt}\begin{tikzonimage}[width=0.3\linewidth]{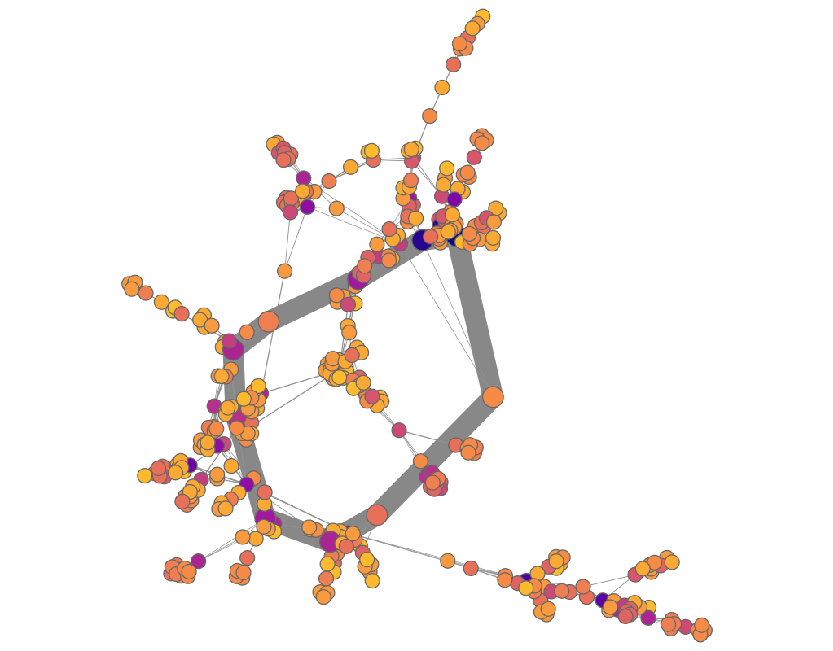}[image label]\node{$Q_{LCMC}$: 0.510};\end{tikzonimage}
\hspace{-12pt}\begin{tikzonimage}[width=0.3\linewidth]{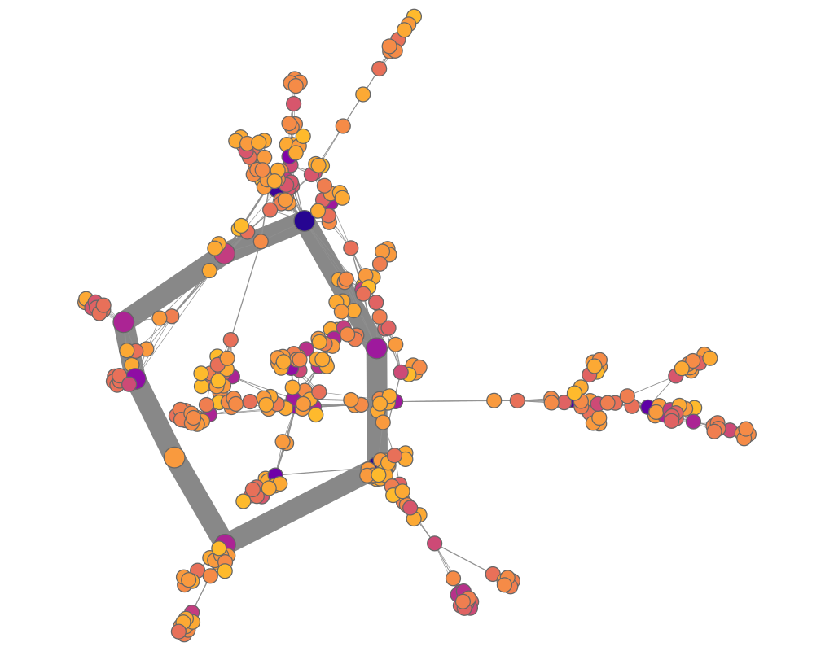}[image label]\node{$Q_{LCMC}$: 0.512};\end{tikzonimage}

\vspace{-55pt}
\hspace{-12pt}
\resizebox{15pt}{!}{\rotatebox{90}{\hspace{15pt}\subfloat[\textsc{retweet}\label{fig:h1_exp:retweet}]{\hspace{60pt}}}}
\hspace{-12pt}
\begin{tikzonimage}[width=0.24\linewidth]{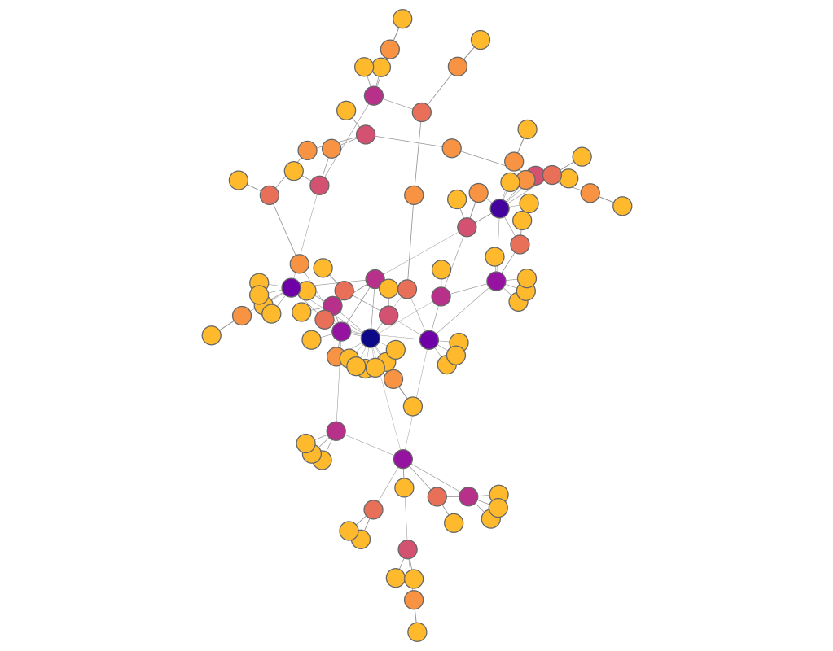}[image label]\node{$Q_{LCMC}$: 0.440};\end{tikzonimage}
\hspace{-10pt}\begin{tikzonimage}[width=0.24\linewidth]{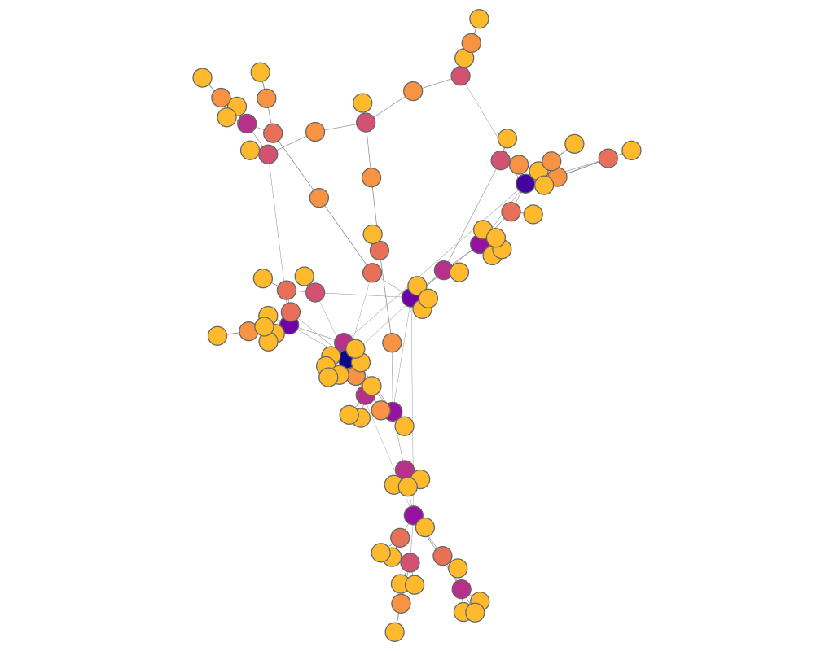}[image label]\node{$Q_{LCMC}$: 0.451};\end{tikzonimage}
\hspace{-17pt}
\hspace{0.5pt}
\begin{tikzonimage}[width=0.3\linewidth]{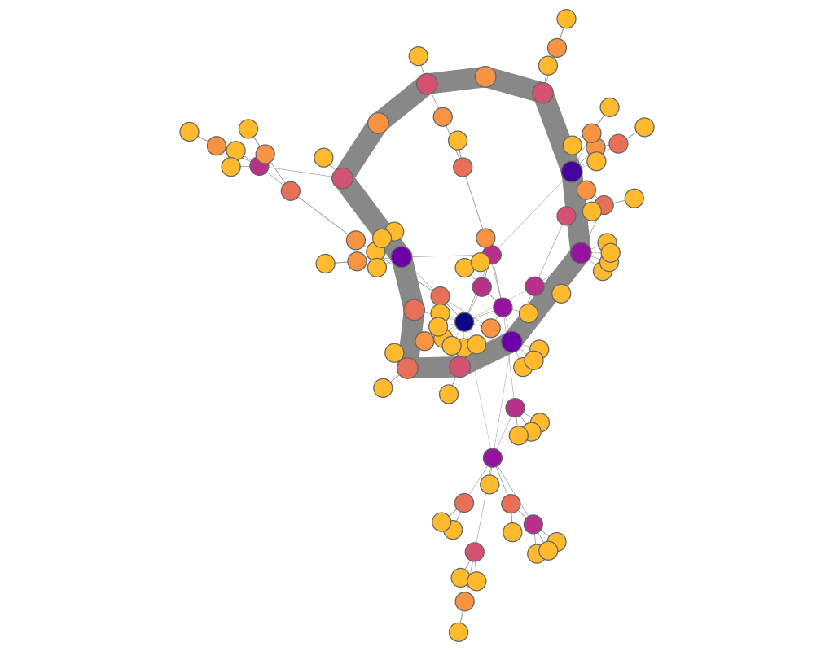}[image label]\node{$Q_{LCMC}$: 0.481};\end{tikzonimage}
\hspace{-25pt}
\begin{tikzonimage}[width=0.3\linewidth]{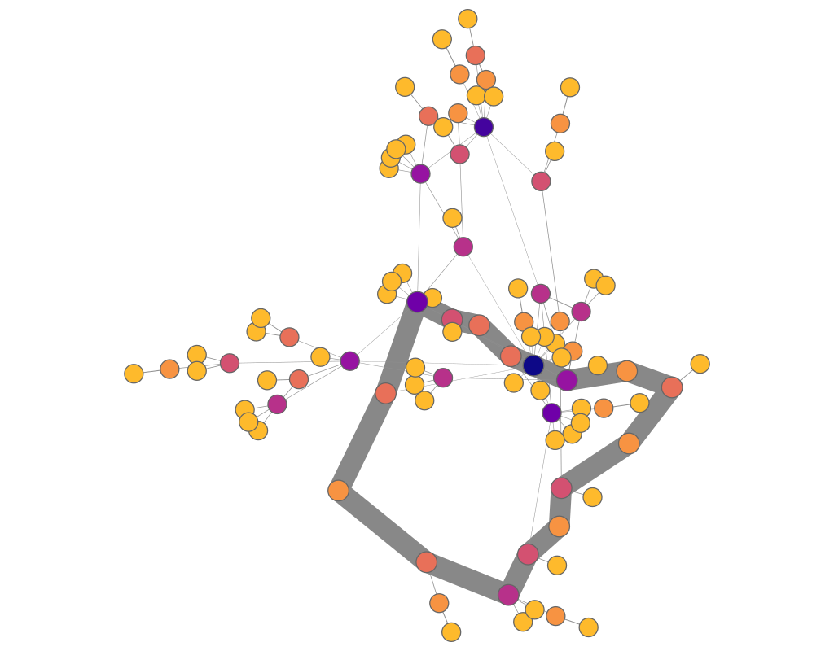}[image label]\node{$Q_{LCMC}$: 0.423};\end{tikzonimage}

\vspace{-5pt}
\hspace{15pt}{\tiny Random Layout}
\hspace{20pt}{\tiny $\Hgroup_0$ Forces~\cite{suh2019persistent}}
\hspace{50pt}{\tiny $\Hgroup_1$ Forces (our approach)}

    \caption{A comparison of graphs using a random initialization and random initialization + $\Hgroup_0$ forces, and examples of our approach on different cycles, random initialization + $\Hgroup_1$ forces. The results show that our approach reveals cycles otherwise difficult to observe in the data, and in most cases, our approach improves the overall graph layout.}
    \label{fig:h1_exp-1}
\end{figure}

\subsubsection{Compute Time}

The improved rate of convergence our method offers does not come for free. An initialization time penalty (i.e., $T_{IT}$), albeit small, must be paid. Due to the imprecise time measurements offered by the web browser, we forgo discussing small graphs and focus on larger graphs instead, namely \textsc{airport}, \textsc{hic 1k net~6}, and \textsc{smith}. For these graphs, there was an additional initialization cost ($T_{IT}$) of $15-20\%$ over a random initialization. 
Since we made no modifications to the force calculations, the average time per iteration, $T_{AIT}$, was virtually identical.

However, the time benefit of our approach is placed in context when considering the time to convergence, $T_{LCMC}$. Due to the low additional overhead and significant reduction in number of iterations, our approach offers a speed-up of $\sim2\times$ for \textsc{airport} and \textsc{hic 1k net~6}, and a speed-up of $\sim5.6\times$ for \textsc{smith}.

Given these observations, we conclude the benefits of fast convergence far outweigh the additional initialization time required for the spanning tree calculation, particularly for larger graphs.

\begin{figure*}[!ht]
    \centering \footnotesize 
    
    \subfloat[\textsc{airport} dataset comparing random initial layout to layered over 300 iterations\label{fig:large_graphs.airport}]{
    \begin{minipage}[b]{10pt}
        \centering 
        \rotatebox{90}{
        \hspace{18pt} Layered 
        \hspace{35pt} Random}
    \end{minipage}
    \begin{minipage}[b]{0.136\linewidth}
        \centering 
        \begin{tikzonimage}[width=\linewidth]{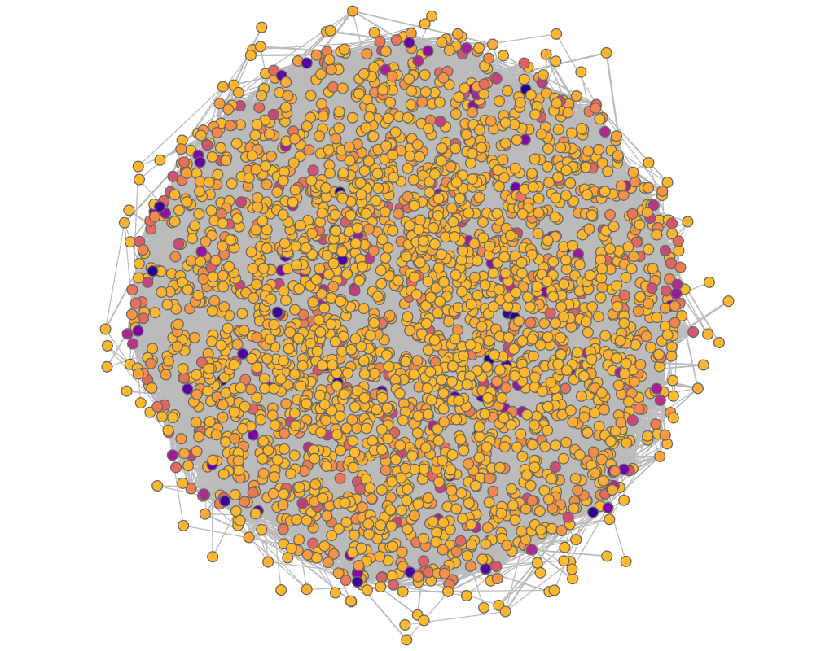}[image label]         \node{$Q_{LCMC}$: 0.004};     \end{tikzonimage}
        \begin{tikzonimage}[width=\linewidth]{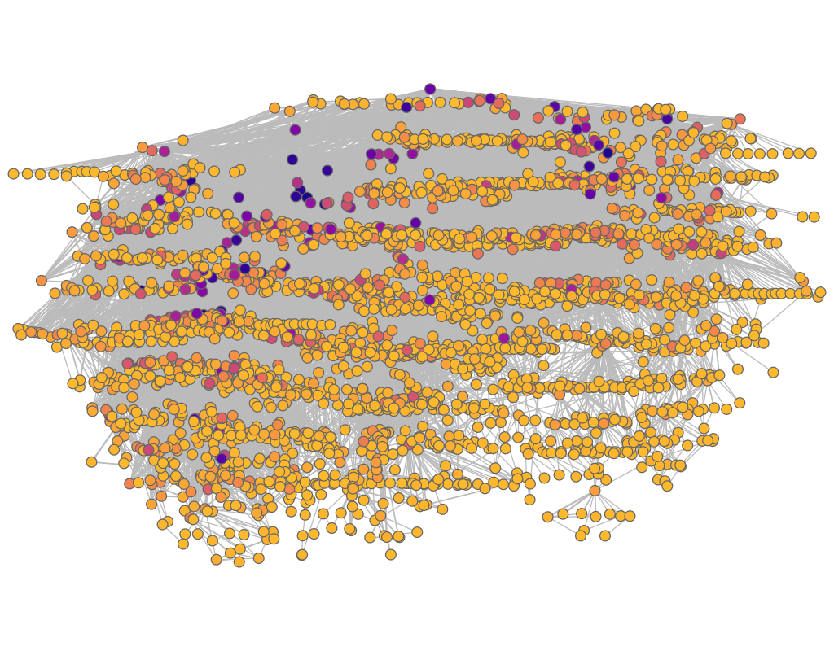}[image label]         \node{$Q_{LCMC}$: 0.070};     \end{tikzonimage}
        Initial Layout
    \end{minipage}
    \begin{minipage}[b]{0.136\linewidth}
        \centering 
        \begin{tikzonimage}[width=\linewidth]{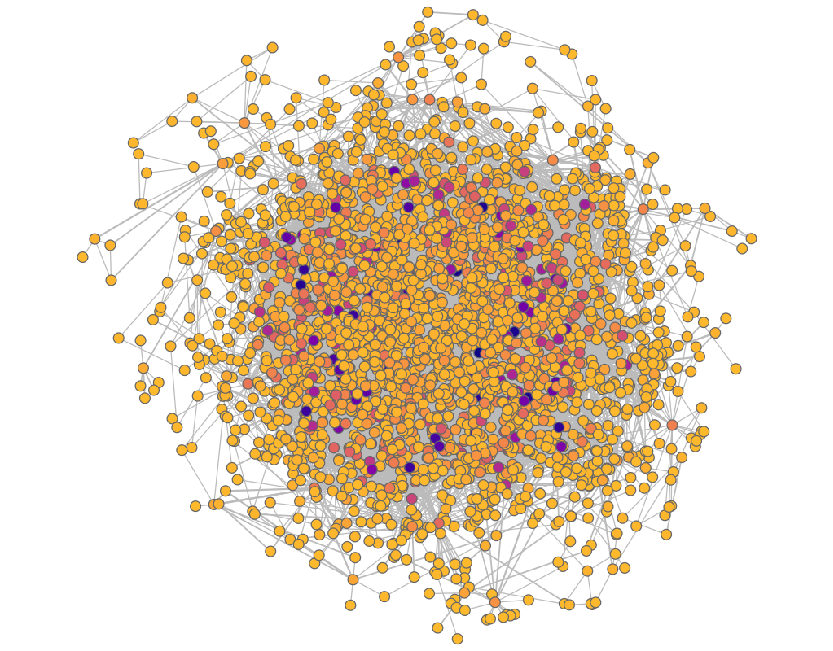}[image label]         \node{$Q_{LCMC}$: 0.032};     \end{tikzonimage}
        \begin{tikzonimage}[width=\linewidth]{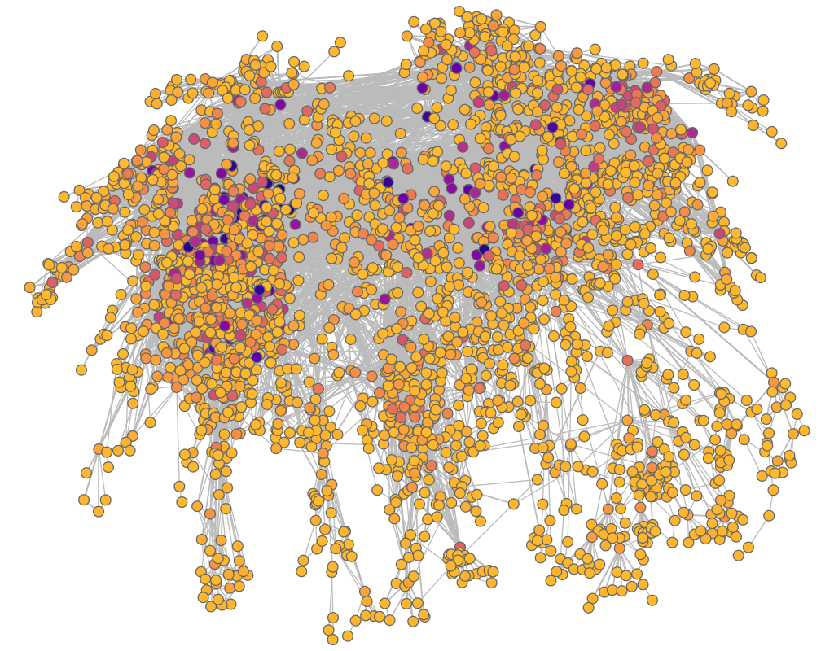}[image label]         \node{$Q_{LCMC}$: 0.150};     \end{tikzonimage}        
        5 Iterations
    \end{minipage}
    \begin{minipage}[b]{0.136\linewidth}
        \centering 
        \begin{tikzonimage}[width=\linewidth]{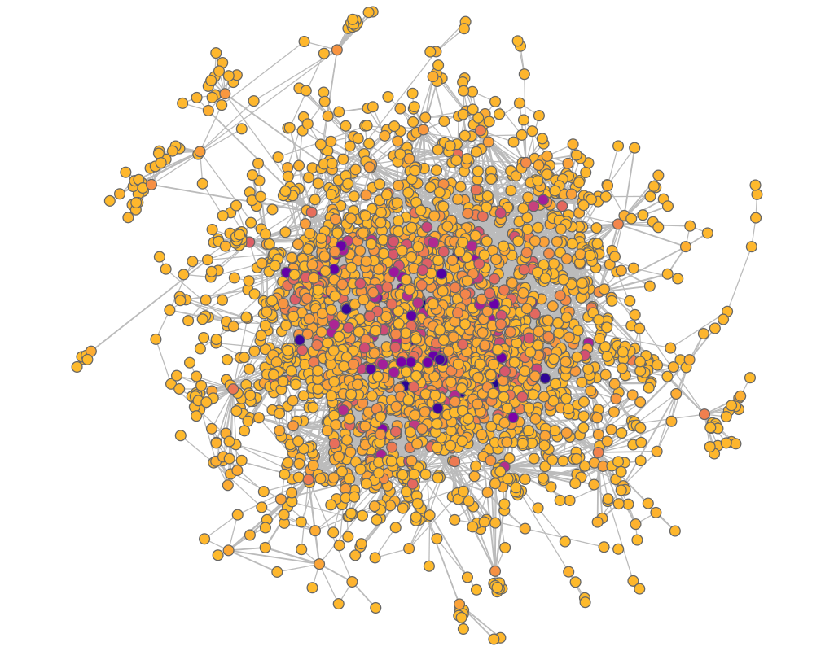}[image label]         \node{$Q_{LCMC}$: 0.115};     \end{tikzonimage}
        \begin{tikzonimage}[width=\linewidth]{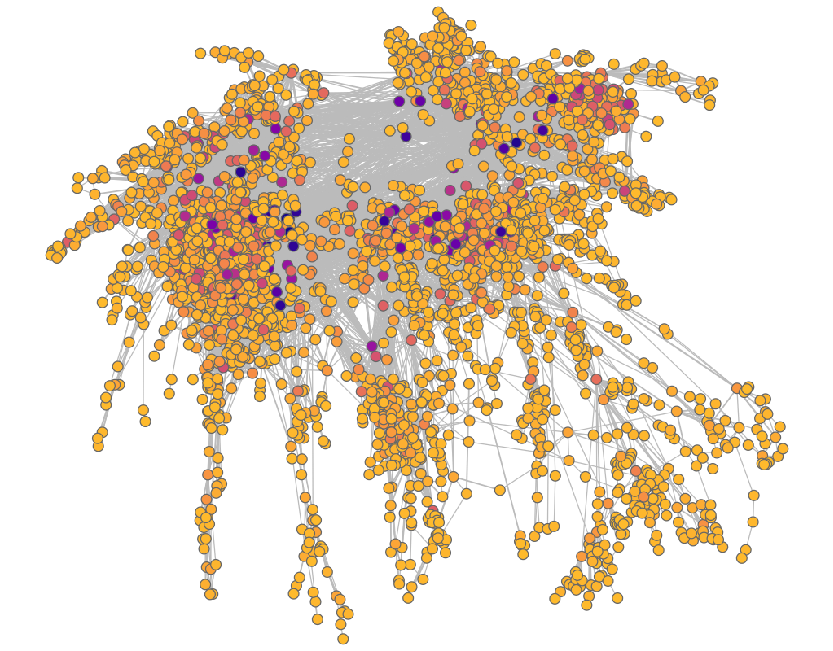}[image label]         \node{$Q_{LCMC}$: 0.197};     \end{tikzonimage}
        10 Iterations
    \end{minipage}
    \begin{minipage}[b]{0.136\linewidth}
        \centering 
        \begin{tikzonimage}[width=\linewidth]{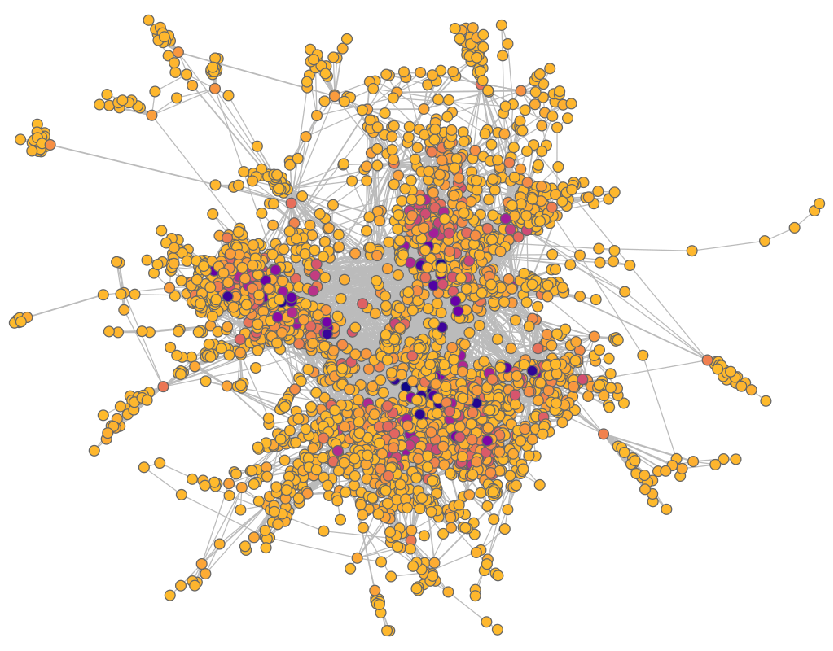}[image label]         \node{$Q_{LCMC}$: 0.187};     \end{tikzonimage}
        \begin{tikzonimage}[width=\linewidth]{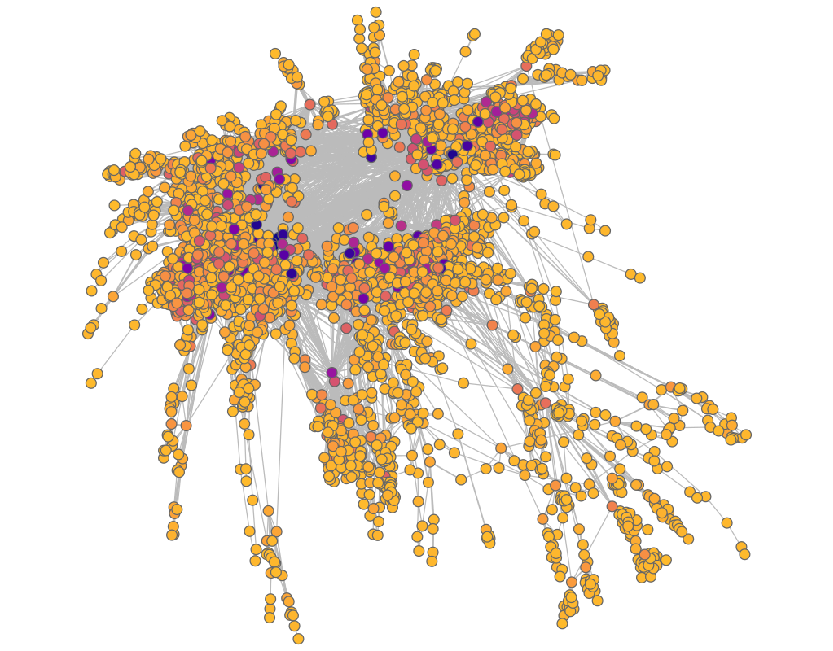}[image label]         \node{$Q_{LCMC}$: 0.230};     \end{tikzonimage}
        25 Iterations
    \end{minipage}
    \begin{minipage}[b]{0.136\linewidth}
        \centering 
        \begin{tikzonimage}[width=\linewidth]{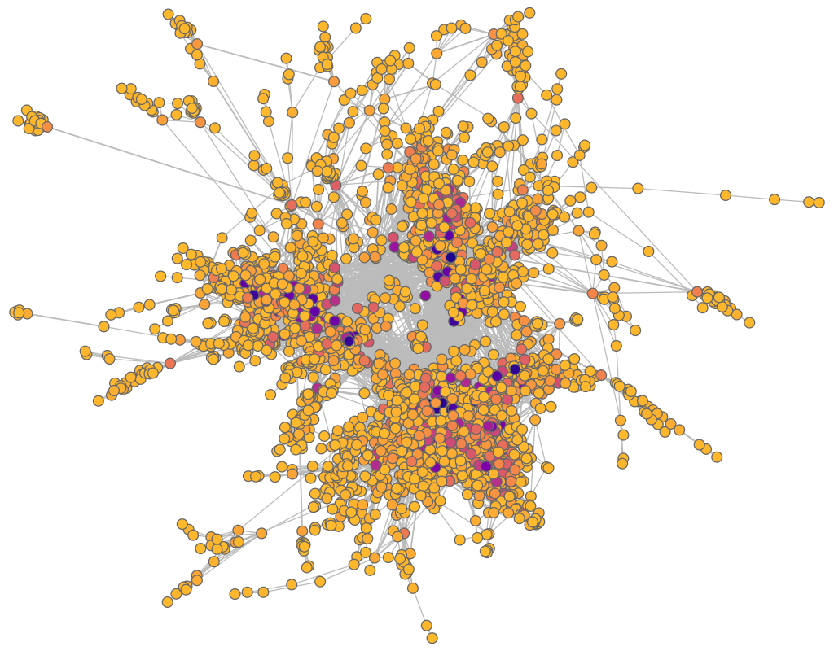}[image label]         \node{$Q_{LCMC}$: 0.224};     \end{tikzonimage}
        \begin{tikzonimage}[width=\linewidth]{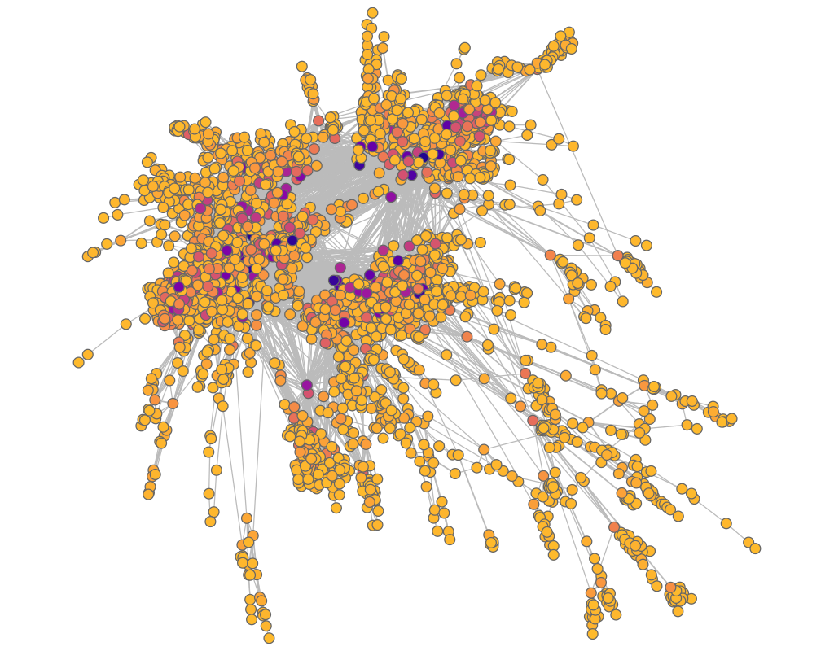}[image label]         \node{$Q_{LCMC}$: 0.249};     \end{tikzonimage}
        50 Iterations
    \end{minipage}
    \begin{minipage}[b]{0.136\linewidth}
        \centering 
        \begin{tikzonimage}[width=\linewidth]{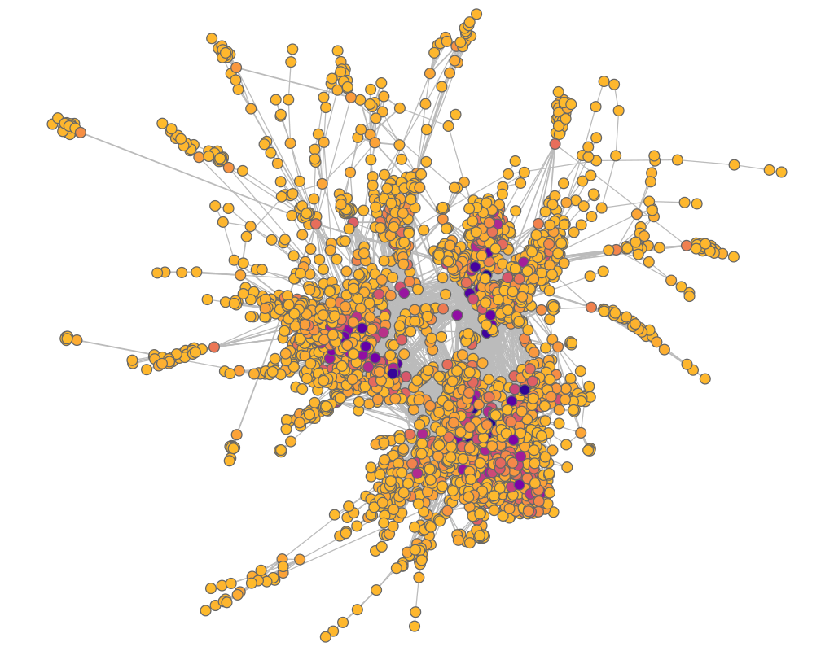}[image label]         \node{$Q_{LCMC}$: 0.245};     \end{tikzonimage}
        \begin{tikzonimage}[width=\linewidth]{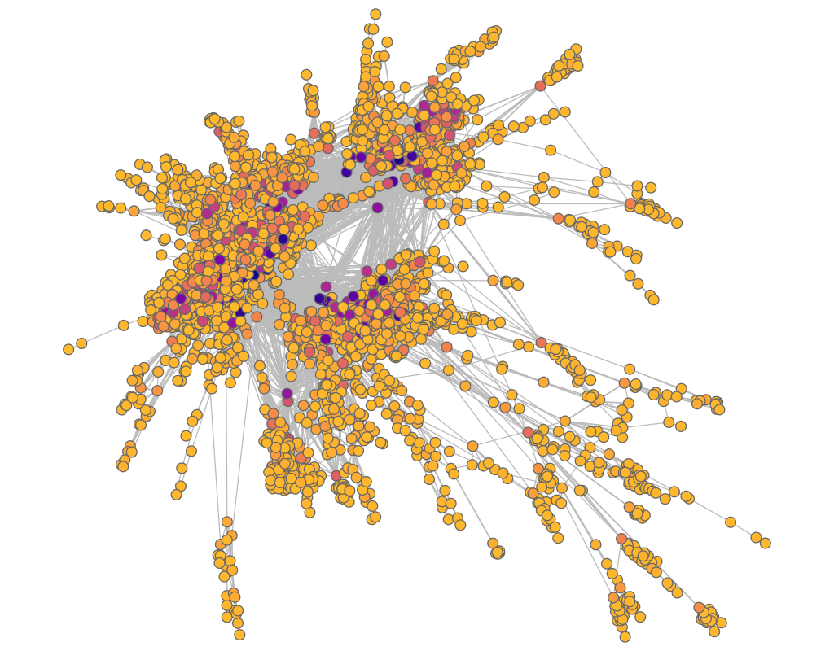}[image label]         \node{$Q_{LCMC}$: 0.254};     \end{tikzonimage}    
        Final Layout
    \end{minipage}
    \begin{minipage}[b]{10pt}
         \centering 
         \rotatebox{90}{
         \hspace{10pt} sfdp 
         \hspace{28pt} fdp
         \hspace{26pt} neato}
    \end{minipage}
    \begin{minipage}[b]{0.095\linewidth}
        \centering 
        \begin{tikzonimage}[width=\linewidth]{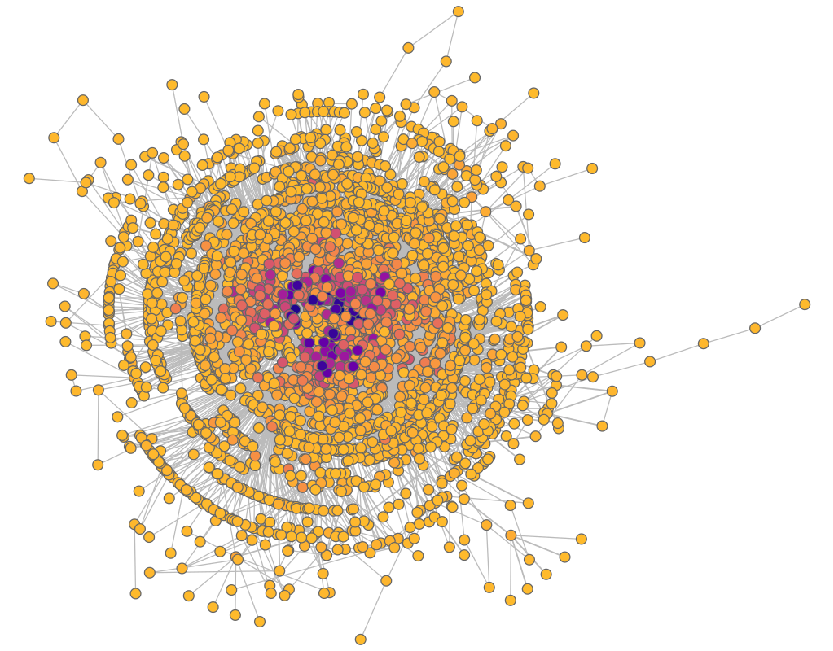}[image label]         \node{$Q_{LCMC}$: 0.269};     \end{tikzonimage}
        \begin{tikzonimage}[width=\linewidth]{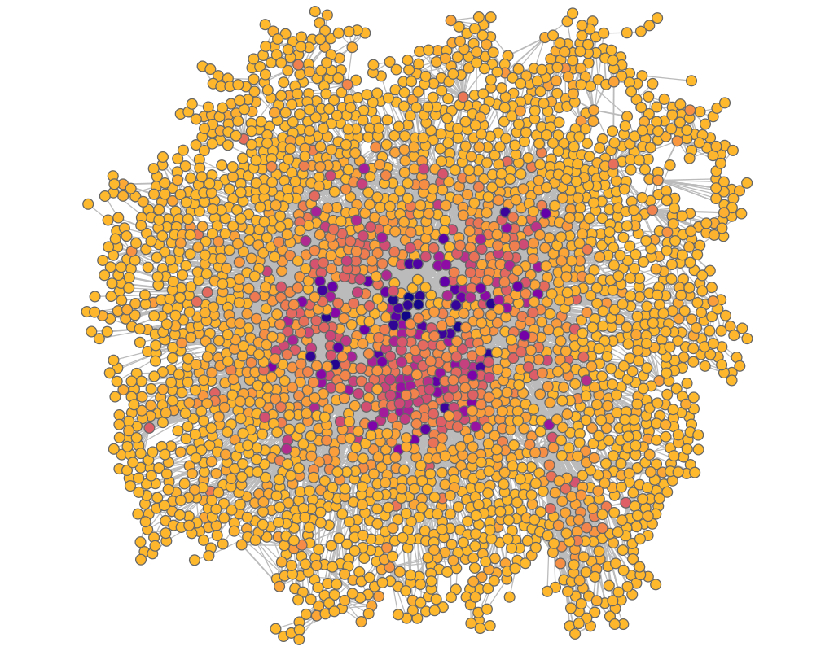}[image label]         \node{$Q_{LCMC}$: 0.213};     \end{tikzonimage}
        \begin{tikzonimage}[width=\linewidth]{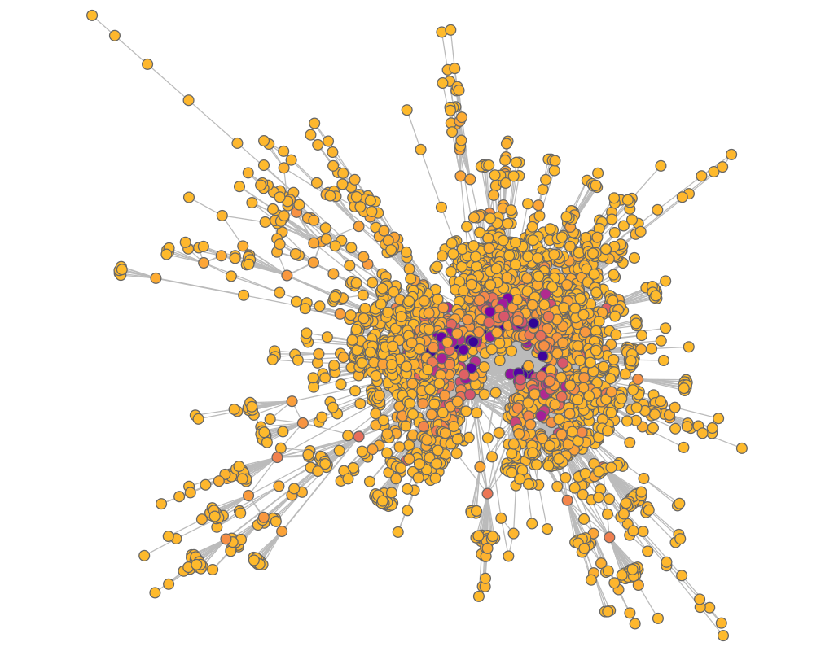}[image label]         \node{$Q_{LCMC}$: 0.289};     \end{tikzonimage}
    \end{minipage}
    }

    \vspace{-5pt}
    \subfloat[\textsc{smith} dataset comparing random initial layout to radial over 300 iterations\label{fig:large_graphs.smith}]{
    \begin{minipage}[b]{10pt}
        \centering 
        \rotatebox{90}{
        \hspace{24pt}Radial 
        \hspace{36pt} Random}
    \end{minipage}
    \begin{minipage}[b]{0.136\linewidth}
        \centering 
        \begin{tikzonimage}[width=\linewidth]{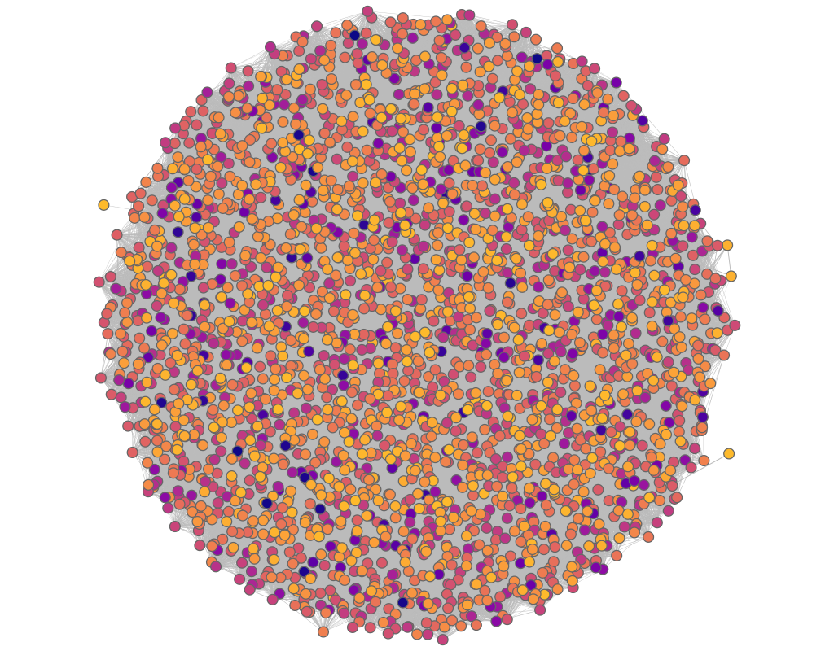}[image label]         \node{$Q_{LCMC}$: 0.000};     \end{tikzonimage}
        \begin{tikzonimage}[width=\linewidth]{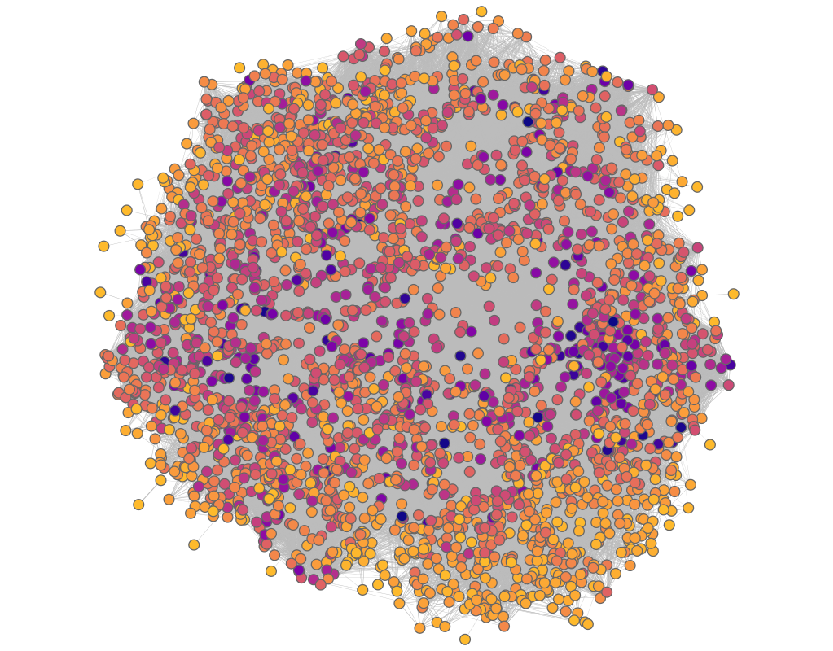}[image label]         \node{$Q_{LCMC}$: 0.080};     \end{tikzonimage}
        Initial Layout
    \end{minipage}
    \begin{minipage}[b]{0.136\linewidth}
        \centering 
        \begin{tikzonimage}[width=\linewidth]{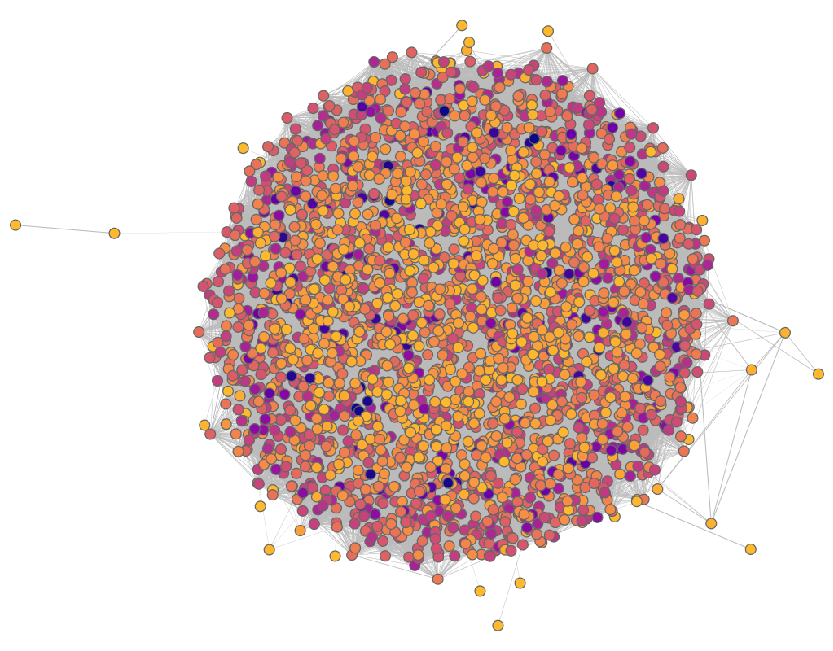}[image label]         \node{$Q_{LCMC}$: 0.001};     \end{tikzonimage}
        \begin{tikzonimage}[width=\linewidth]{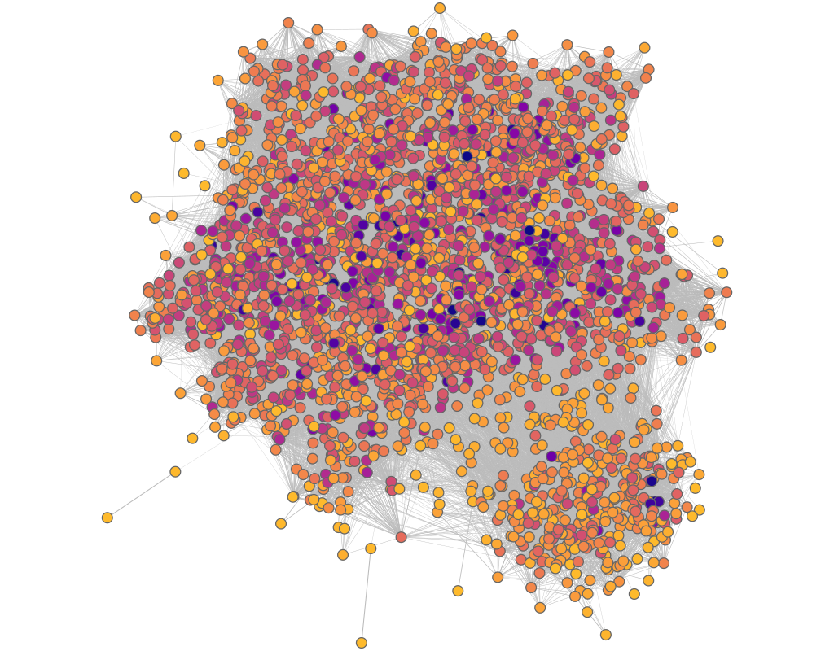}[image label]         \node{$Q_{LCMC}$: 0.050};     \end{tikzonimage}
        5 Iterations
    \end{minipage}
    \begin{minipage}[b]{0.136\linewidth}
        \centering 
        \begin{tikzonimage}[width=\linewidth]{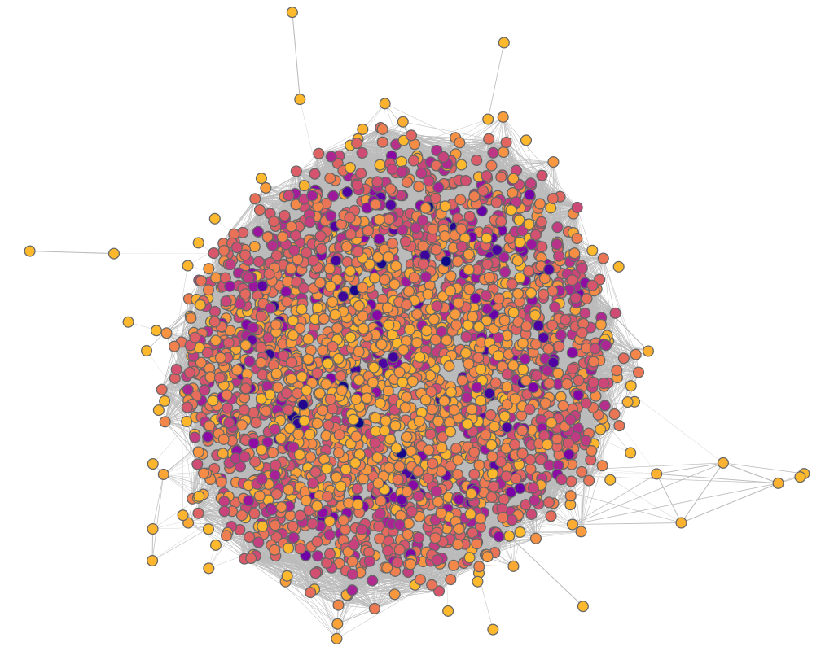}[image label]         \node{$Q_{LCMC}$: 0.006};     \end{tikzonimage}
        \begin{tikzonimage}[width=\linewidth]{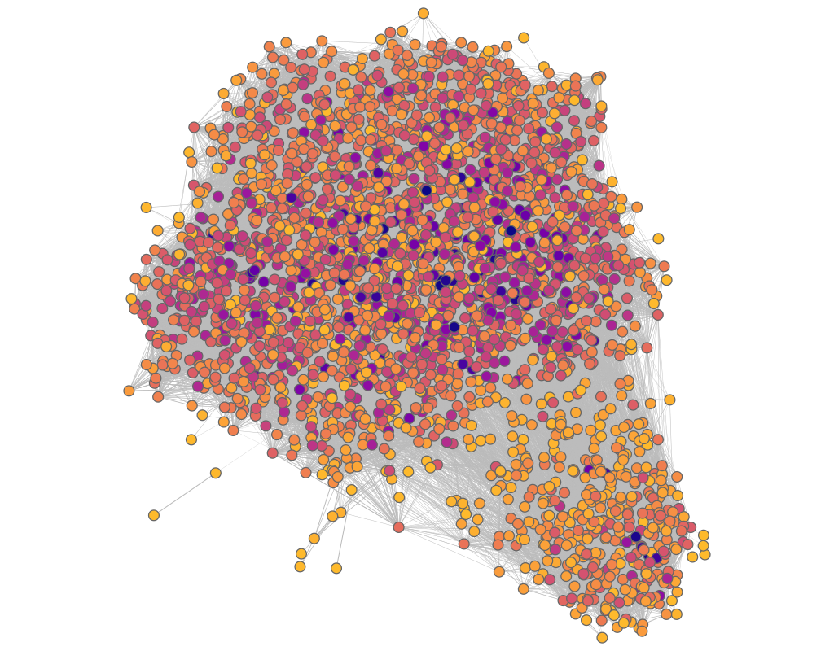}[image label]         \node{$Q_{LCMC}$: 0.044};     \end{tikzonimage}
        10 Iterations
    \end{minipage}
    \begin{minipage}[b]{0.136\linewidth}
        \centering 
        \begin{tikzonimage}[width=\linewidth]{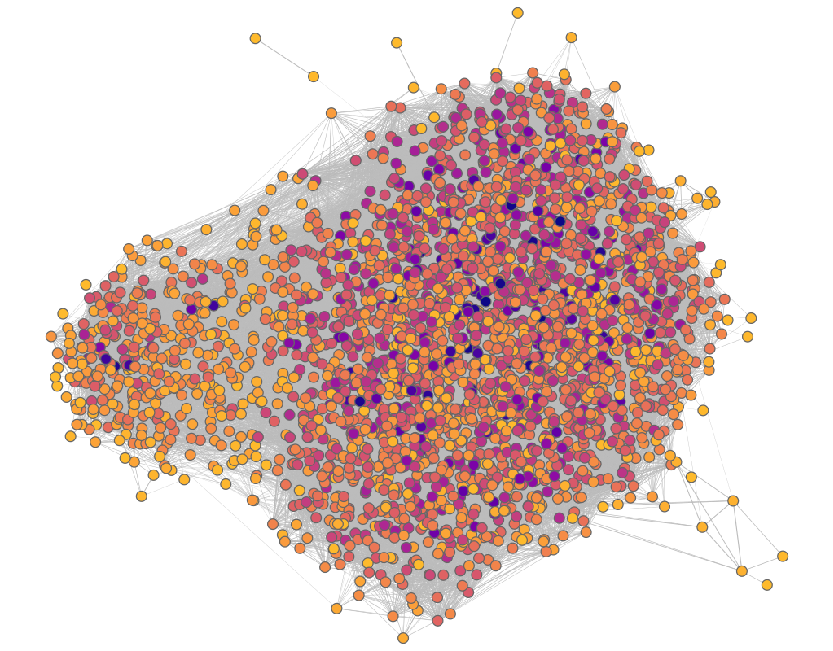}[image label]         \node{$Q_{LCMC}$: 0.031};     \end{tikzonimage}
        \begin{tikzonimage}[width=\linewidth]{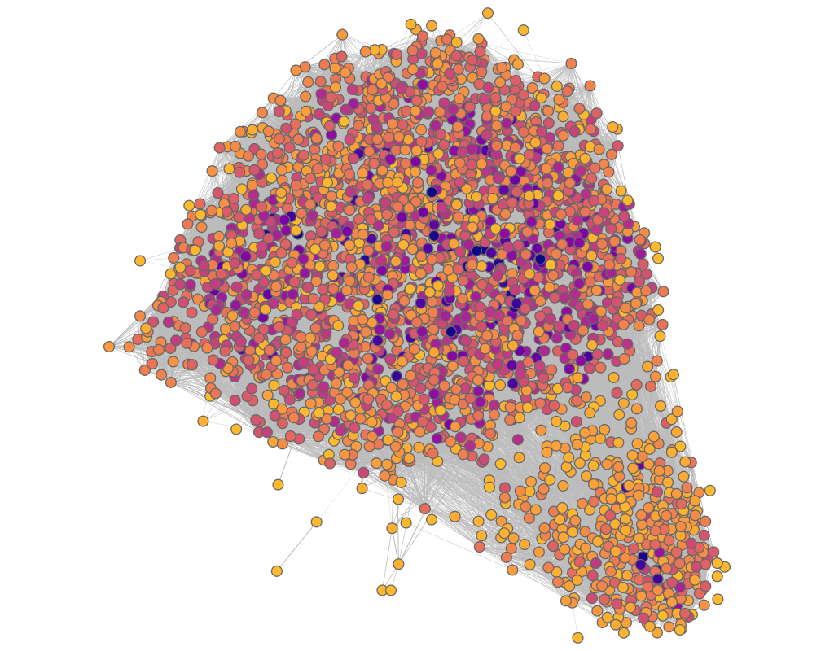}[image label]         \node{$Q_{LCMC}$: 0.042};     \end{tikzonimage}
        25 Iterations
    \end{minipage}
    \begin{minipage}[b]{0.136\linewidth}
        \centering 
        \begin{tikzonimage}[width=\linewidth]{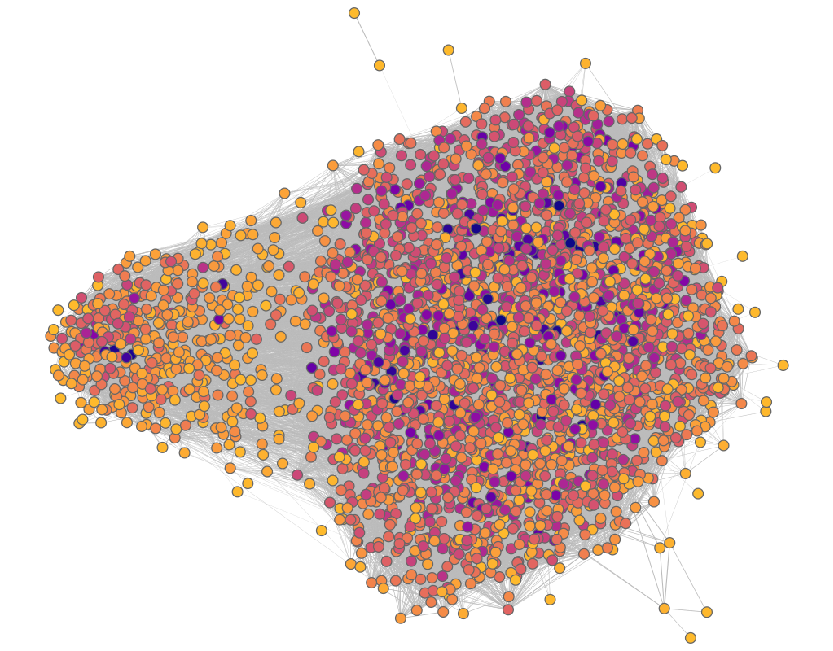}[image label]         \node{$Q_{LCMC}$: 0.039};     \end{tikzonimage}
        \begin{tikzonimage}[width=\linewidth]{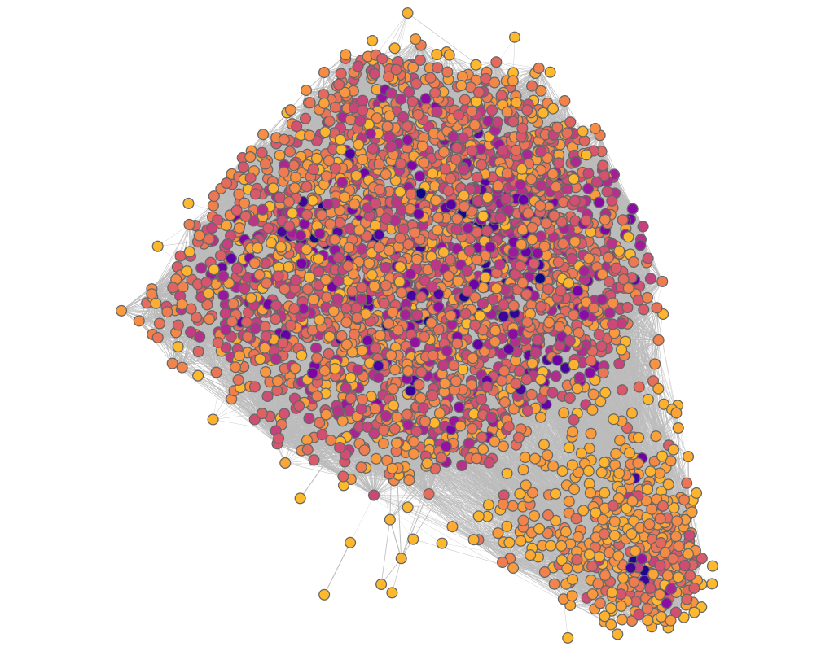}[image label]         \node{$Q_{LCMC}$: 0.044};     \end{tikzonimage}
        50 Iterations
    \end{minipage}
    \begin{minipage}[b]{0.136\linewidth}
        \centering 
        \begin{tikzonimage}[width=\linewidth]{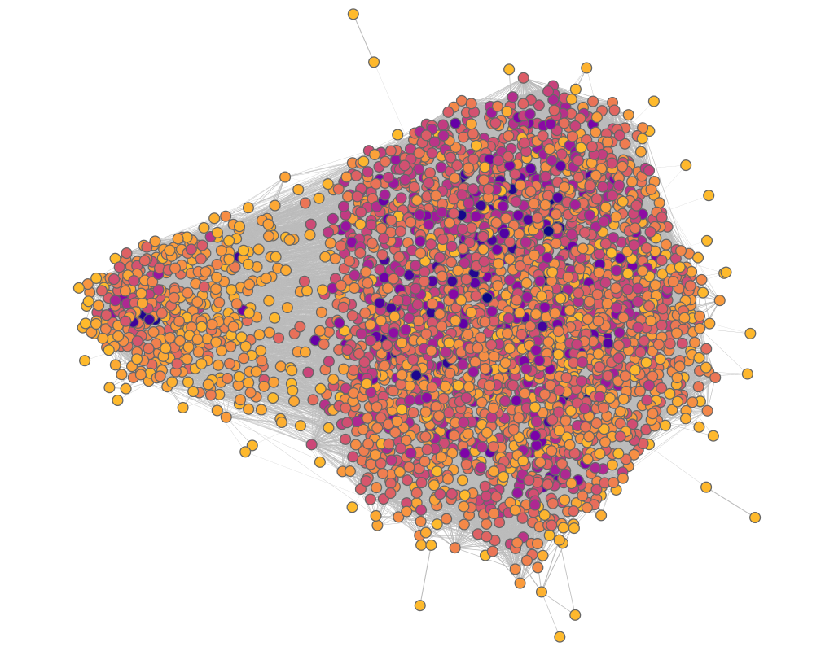}[image label]         \node{$Q_{LCMC}$: 0.045};     \end{tikzonimage}
        \begin{tikzonimage}[width=\linewidth]{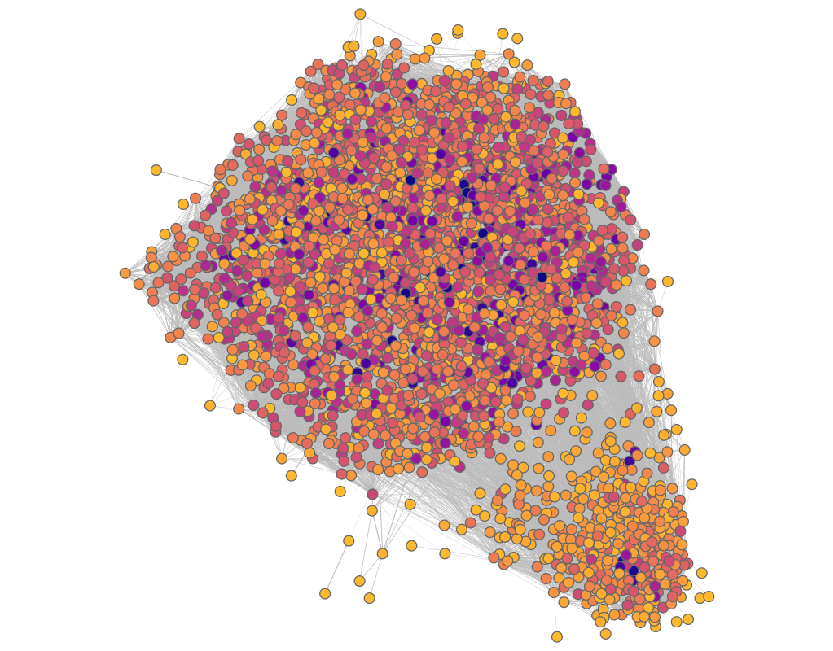}[image label]         \node{$Q_{LCMC}$: 0.045};     \end{tikzonimage}
        Final Layout
    \end{minipage}
    \begin{minipage}[b]{1pt}
         \centering 
         \rotatebox{90}{
         \hspace{10pt} sfdp 
         \hspace{28pt} fdp
         \hspace{26pt} neato}
    \end{minipage}
    \begin{minipage}[b]{0.095\linewidth}
        \begin{tikzonimage}[width=\linewidth]{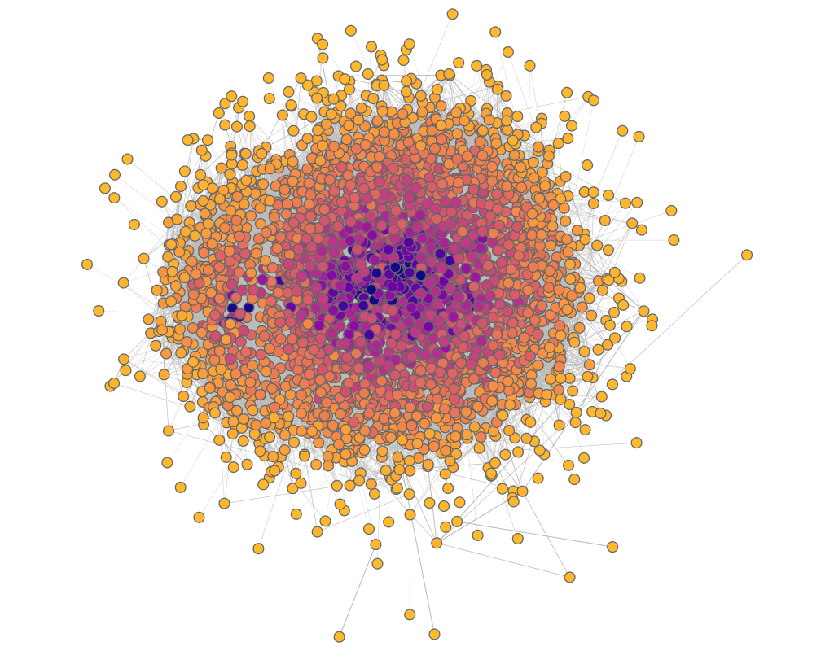}[image label]         \node{$Q_{LCMC}$: 0.046};     \end{tikzonimage}
        \begin{tikzonimage}[width=\linewidth]{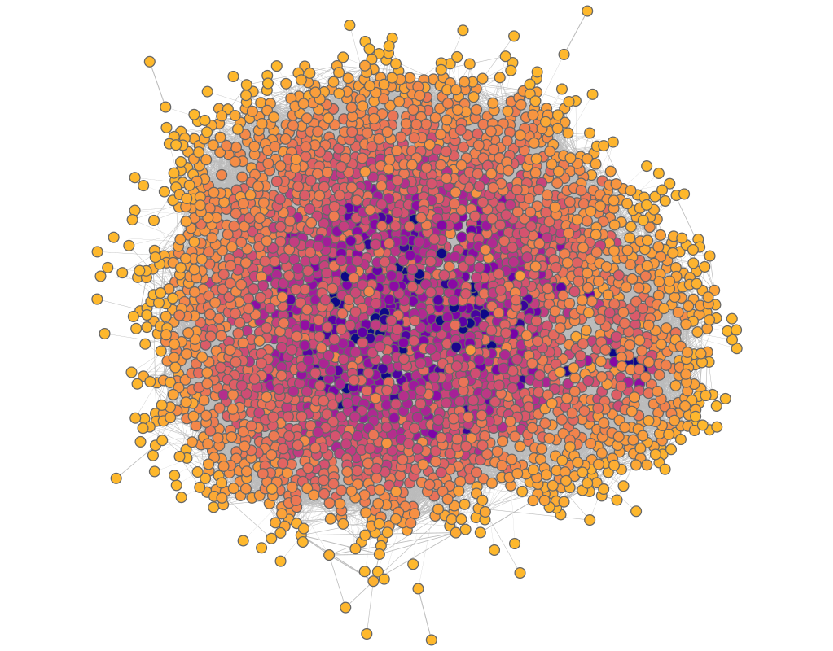}[image label]         \node{$Q_{LCMC}$: 0.048};     \end{tikzonimage}
        \begin{tikzonimage}[width=\linewidth]{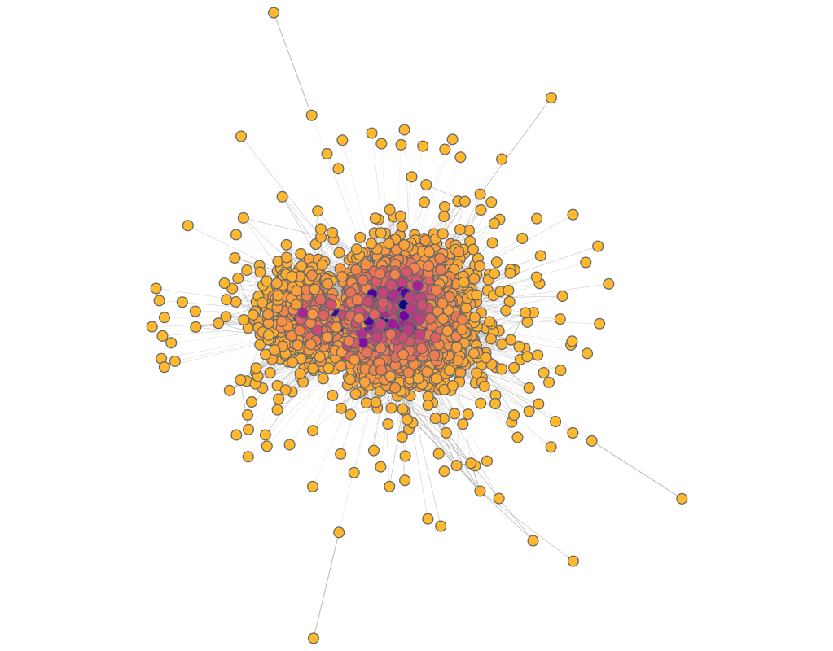}[image label]         \node{$Q_{LCMC}$: 0.063};     \end{tikzonimage}
    \end{minipage}
    }

    \caption{A comparison of random initial layout to our approach for large graphs shows the layouts at various stages of processing. Importantly, our approach shows the graph structure much earlier in the process, and the final layout is similar or better quality. We also compare to the results of non-interactive methods of neato, fdp, and sfdp (right).}
    \label{fig:large_graphs}
    \vspace{2pt}
\end{figure*}

\subsubsection{Comparison to Other Algorithms}
\label{sec:eval:other_algos}

We also compared our results to other graph layout algorithms, namely, \emph{neato}, \emph{fdp}, and \emph{sfdp}. Due to space consideration, the majority of results are presented in our supplemental document. However, results on larger graphs can be seen in \autoref{fig:teaser} and \ref{fig:large_graphs}, and on smaller graphs in \autoref{fig.dense} and \ref{fig.sparse}. Generally speaking, one or more of these methods produced graph layouts with similar or slightly better $Q_{LCMC}$ scores than our approach. Therefore, our method can be viewed as closing the gap between random initialization force-directed layout and these more advanced methods. Ultimately, we are still limited by the capacity of the D3.js force-directed layout to produce high-quality final layouts. Our method  provides only a boost. Finally, an important aspect of our approach is that the layouts are intended to be interactive. Users are supposed to explore $\Hgroup_0$ and $\Hgroup_1$ features to learn the graph's structure, which is a capability not necessarily offered by these other methods.

\begin{figure}[!b]

\vspace{8pt}

\hspace{2pt}
\rotatebox{90}{\tiny\hspace{5pt}Random Layout}
\hspace{5pt}
\begin{tikzonimage}[width=0.21\linewidth]{fig/auto/init_layout_bio-diseasome_random_final}[image label]\node{$Q_{LCMC}$: 0.407};\end{tikzonimage}
\hspace{0pt}
\begin{tikzonimage}[width=0.21\linewidth]{fig/auto/init_layout_circular_ladder_graph__100___random_final}[image label]\node{$Q_{LCMC}$: 0.455};\end{tikzonimage}
\hspace{0pt}
\begin{tikzonimage}[width=0.21\linewidth]{fig/auto/init_layout_engymes_g123_random_final}[image label]\node{$Q_{LCMC}$: 0.374};\end{tikzonimage}
\hspace{0pt}
\begin{tikzonimage}[width=0.21\linewidth]{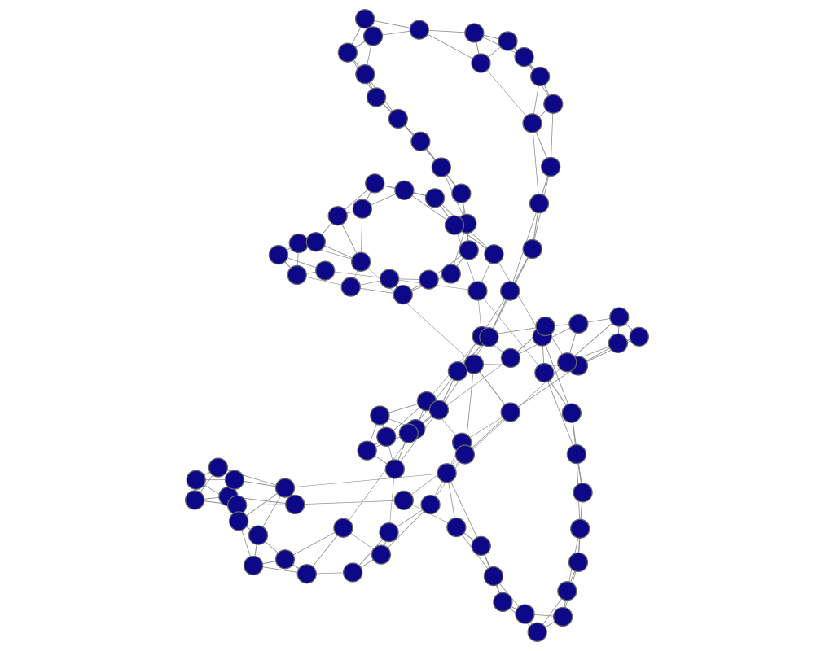}[image label]\node{$Q_{LCMC}$: 0.479};\end{tikzonimage}

\vspace{3pt}
\hspace{2pt}
\rotatebox{90}{\tiny\hspace{5pt}$\Hgroup_0$ Forces~\cite{suh2019persistent}}
\hspace{5pt}
\begin{tikzonimage}[width=0.21\linewidth]{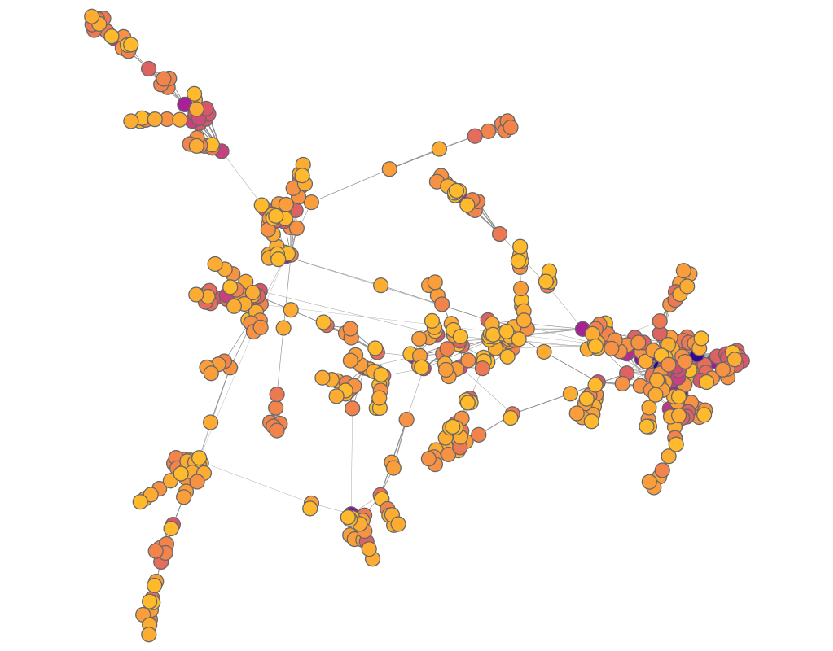}[image label]\node{$Q_{LCMC}$: 0.402};\end{tikzonimage}
\hspace{0pt}
\begin{tikzonimage}[width=0.21\linewidth]{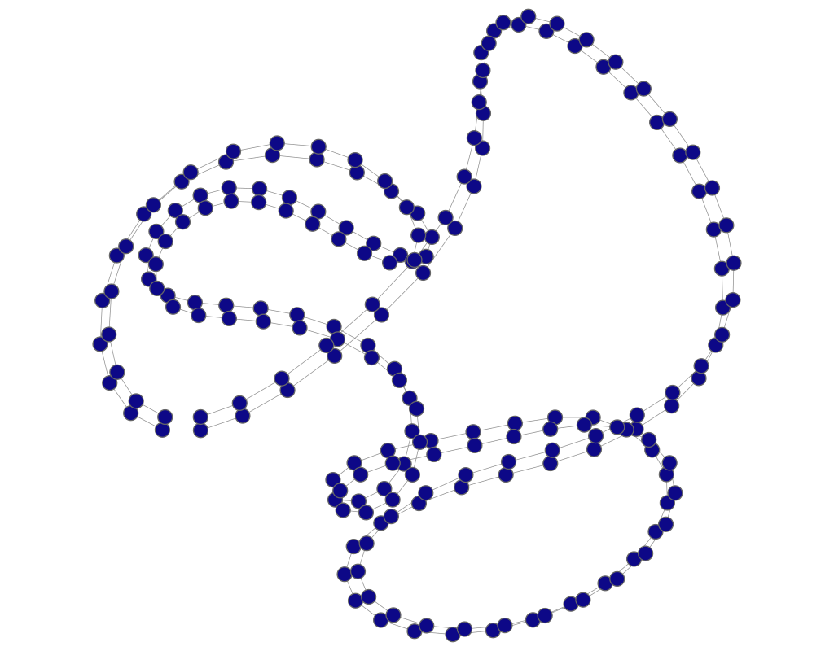}[image label]\node{$Q_{LCMC}$: 0.600};\end{tikzonimage}
\hspace{0pt}
\begin{tikzonimage}[width=0.21\linewidth]{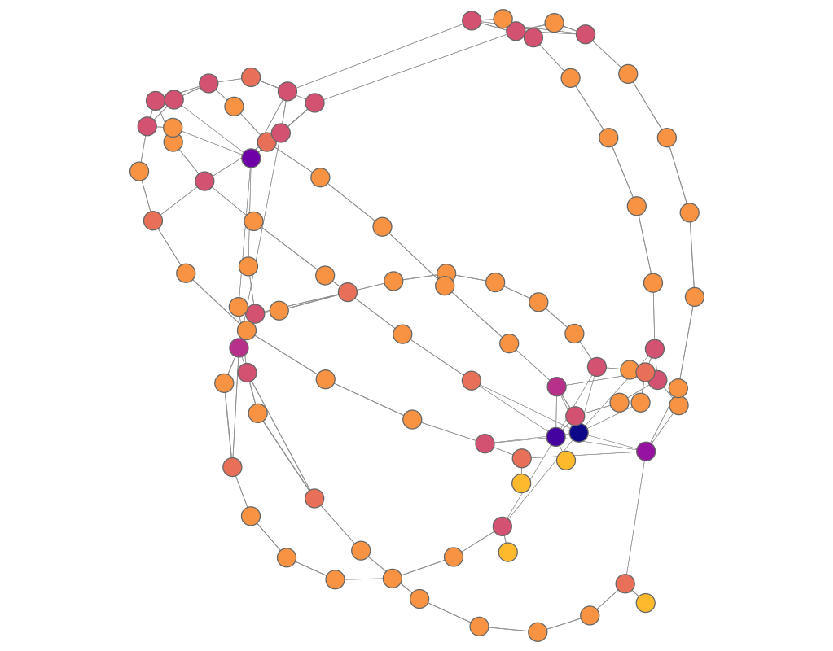}[image label]\node{$Q_{LCMC}$: 0.438};\end{tikzonimage}
\hspace{0pt}
\begin{tikzonimage}[width=0.21\linewidth]{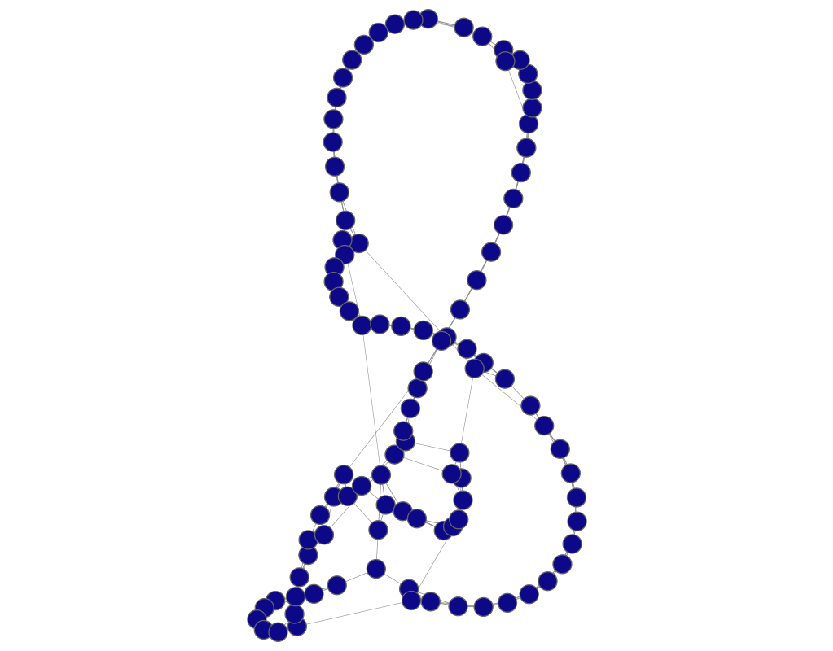}[image label]\node{$Q_{LCMC}$: 0.594};\end{tikzonimage}

\vspace{3pt}
\rotatebox{90}{\tiny\hspace{9pt}$\Hgroup_1$ Forces}
\rotatebox{90}{\tiny\hspace{6pt}(our approach)}
\begin{tikzonimage}[width=0.25\linewidth]{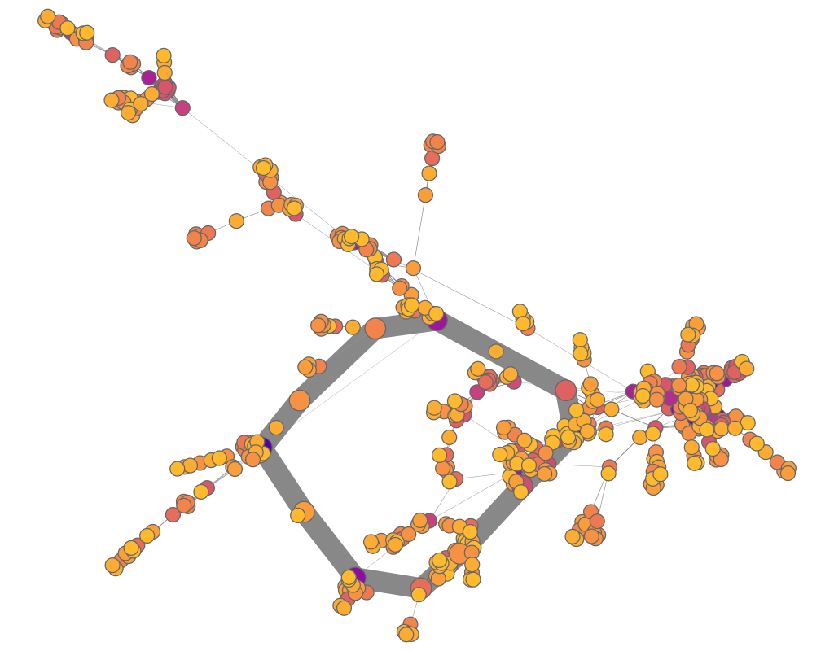}[image label]\node{$Q_{LCMC}$: 0.477};\end{tikzonimage}
\hspace{-10pt}
\begin{tikzonimage}[width=0.25\linewidth]{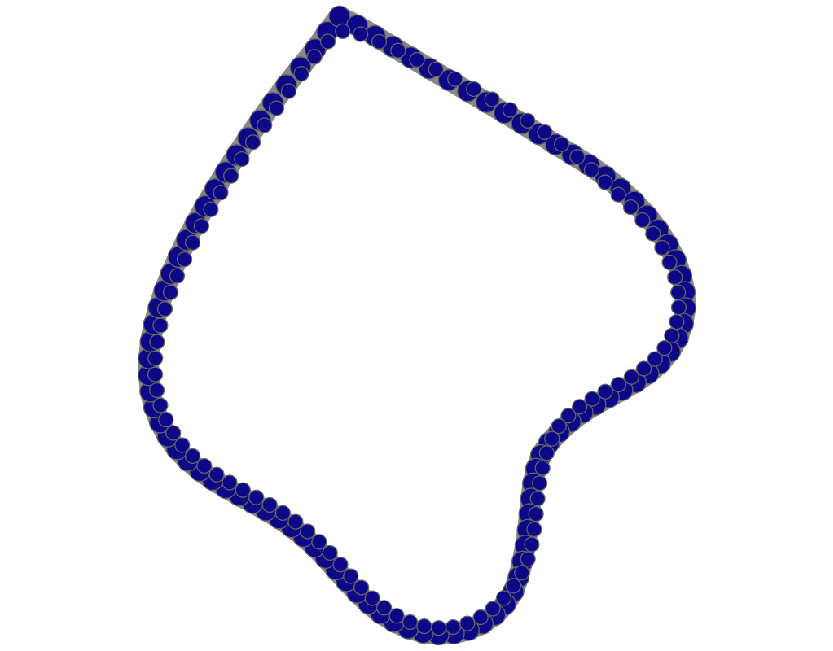}[image label]\node{$Q_{LCMC}$: 0.831};\end{tikzonimage}
\hspace{-10pt}
\begin{tikzonimage}[width=0.25\linewidth]{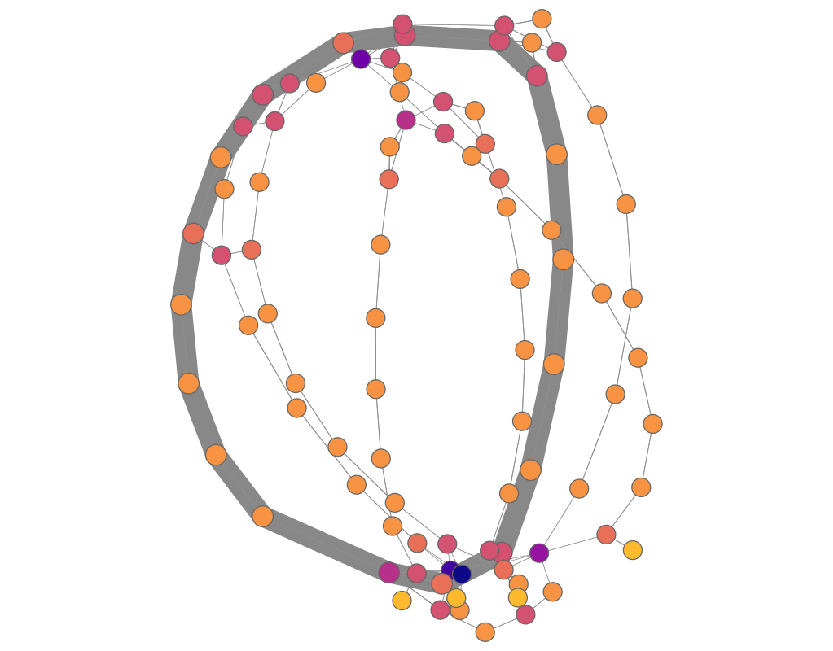}[image label]\node{$Q_{LCMC}$: 0.499};\end{tikzonimage}
\hspace{-10pt}
\begin{tikzonimage}[width=0.25\linewidth]{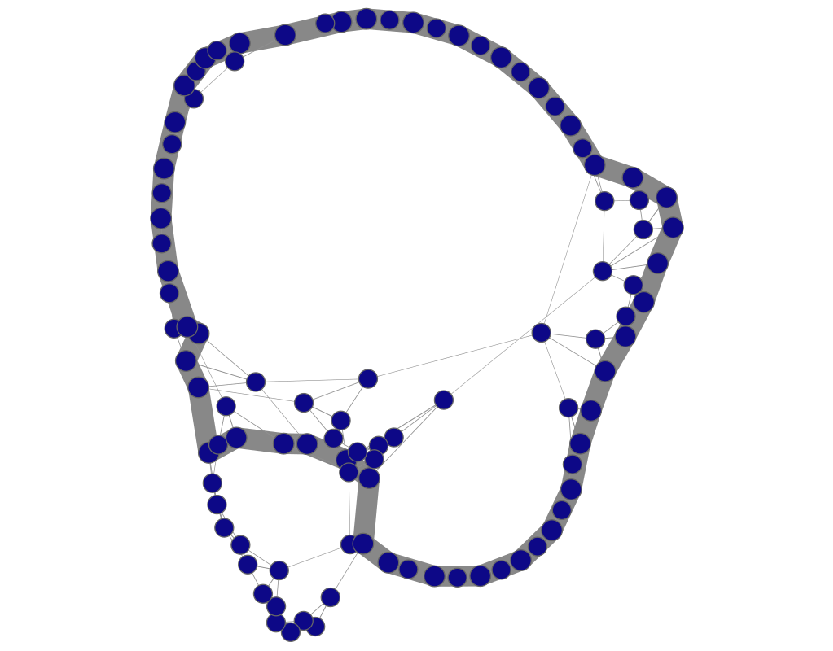}[image label]\node{$Q_{LCMC}$: 0.689};\end{tikzonimage}
\hspace{-7pt}

\vspace{-10pt}
\resizebox{\linewidth}{!}{\hspace{30pt} 
    \subfloat[\textsc{bio-diseasome}]{\hspace{75pt}}
    \hspace{12pt} 
    \subfloat[\textsc{circular ladder}\label{fig:h1_exp:circ}]{\hspace{75pt}}
    \hspace{12pt} 
    \subfloat[\textsc{engymes-g123}]{\hspace{75pt}}
    \hspace{8pt} 
    \subfloat[\textsc{watts strogatz}\label{fig:h1_exp:watts}]{\hspace{75pt}
    \hspace{15pt}}
}

    \caption{Additional examples of graphs using a random initialization, random initialization + $\Hgroup_0$ forces, or our approach, random initialization + $\Hgroup_1$ forces.}
    \label{fig:h1_exp-2}
\end{figure}

\subsection{Interactive Untangling with \texorpdfstring{$\Hgroup_1$}{H1} Features}
\label{sec.eval.h1}

We evaluate whether the interaction with $\Hgroup_1$ features reveals anything about the graph structure previously available using a random initialization force-directed layout or by using the $\Hgroup_0$ forces introduced in~\cite{suh2019persistent}.

\para{Generating \texorpdfstring{$\Hgroup_0$}{H0} Examples} To compare to $\Hgroup_0$ persistent homology feature forces, we start with the random initial layout (i.e., the D3.js default) and allow the graph to converge to a stable configuration. In other words, they look like the random final graph layouts in the upper middle of \autoref{fig.dense} and \ref{fig.sparse}. We then apply a force to \textit{all} $\Hgroup_0$ features and again allow the graph to converge.

\para{Generating \texorpdfstring{$\Hgroup_1$}{H1} Examples} To determine the efficacy of $\Hgroup_1$ persistent homology feature forces, we configure them similarly. Starting with the random initial layout, we allow the graph to converge to a stable configuration. We then apply a force to a single $\Hgroup_1$ feature, which is selected by considering $\Hgroup_1$ features with longer cycle lengths, and again the graph is allowed to converge. 

\para{Results}
The results are visible primarily in \autoref{fig:h1_exp-1} and \ref{fig:h1_exp-2}, and additionally in \autoref{fig:teaser2}. First, we can see that our approach reveals cycle structures that are often hidden in standard graph layout and only sometimes revealed using $\Hgroup_0$ features. In other words, our approach reveals topology of the graph that is otherwise hidden or difficult to see. 
The second observation we make is that whereas $\Hgroup_0$ features tend to improve the overall presentation of the graph, i.e., higher $Q_{LCMC}$ scores, our approach of applying forces to $\Hgroup_1$ features has a much stronger impact, resulting in even better graph layouts, sometimes dramatically so, e.g., in \textsc{circular ladder} (\autoref{fig:h1_exp:circ}) or \textsc{watts strogatz} (\autoref{fig:h1_exp:watts}) graphs.

An important aspect of interacting with $\Hgroup_1$ features is that highlighting different features may lead to very different graph layouts, e.g., see \autoref{fig:h1_exp:retweet}. Taken in isolation, it may be difficult to visually relate the structures of each. However, when interacting, animation provides the context for relating different structures. Our demo\footnote{Demo at \demoURL.} includes the ability to interactively evaluate additional cycles from each dataset.

Given these observations, we see that our approach has the ability to capture and highlight important cycle structures, and it can use those structures to further untangle a force-directed layout.

\section{Discussion and Conclusions}

In this paper, we have evaluated two new uses of persistent homology on force-directed layouts. We first investigate using $\Hgroup_0$ persistent homology for initializing graph layouts. Although the implementation itself relies on maximal spanning trees, persistent homology provides a theoretical foundation for justifying its use. At the same time, our experimental results show that it  indeed improves the convergence rate and quality of force-directed layouts. 

Second, we investigate using $\Hgroup_1$ features for highlighting and modifying the forces of a force-directed layout. Here again, in addition to the algorithmic contribution of efficiently extracting the $\Hgroup_1$ features, we observe that using these features reveals hidden features and improves graph layouts in many situations.

Beyond our current work, there is potentially room for developing additional initial layout schemes or perhaps automatically identifying which scheme would work best for a given dataset. Still, a good balance between performance and final quality remains of the utmost importance. In addition, our scheme for utilizing $\Hgroup_1$ features could be utilized in a more elaborate manner. Additionally, it would be interesting to study whether other simplicial complexes could be used with persistent homology to capture topological information about other graph structures, e.g., cliques, stars and trees. Finally, our approach is implemented using the D3.js force-directed layout, but we believe the approach would work with other state-of-the-art techniques and frameworks. However, the exact implementation and the ultimate performance and quality gain require additional study.

\vspace{5pt}
\noindent
Demo: \demoURL \\
Source: \sourceURL

\acknowledgments{The authors wish to thank Ashley Suh, Mustafa Hajij, and Carlos Scheidegger for initial discussions. 
This work was supported in part by
grants from the National Science Foundation (NSF) IIS-1513616 and IIS-1845204, and Department of Energy (DOE) DE-SC0021015. 
}

\bibliographystyle{abbrv}
\bibliography{main}

\end{document}